\documentclass[journal,twocolumn]{IEEEtran}

\usepackage{palatino}
\usepackage{mathtools}
\usepackage{url}
\usepackage{amsmath}
\usepackage{amssymb}
\usepackage{wrapfig}
\usepackage{framed}
\usepackage{mdframed}
\usepackage{verbatim}
 \usepackage{amsthm}
 \usepackage{cite}
\usepackage{subfigure}
\newcommand\independent{\protect\mathpalette{\protect\independenT}{\perp}}
\def\independenT#1#2{\mathrel{\rlap{$#1#2$}\mkern2mu{#1#2}}}

\usepackage{palatino}
\usepackage{mathtools}
\usepackage{url}
\usepackage{amsmath}
\usepackage{amssymb}
\usepackage{wrapfig}
\usepackage{framed}
\usepackage{mdframed}
\usepackage{bbm}
\usepackage{verbatim}\usepackage{subfigure}
\usepackage[utf8]{inputenc}
\usepackage{amssymb}
\usepackage{amsthm}
\usepackage{dsfont}
\usepackage[normalem]{ulem}

\usepackage{bbm}
\usepackage{subfigure}
\usepackage{times}
\usepackage{multirow}
\usepackage[final]{graphicx}
\usepackage{upgreek}
\usepackage{latexsym,amssymb,mathrsfs}
\usepackage{comment}
\usepackage{url}

\usepackage{algorithmicx}
\usepackage[ruled]{algorithm}
\usepackage{algpseudocode}
\usepackage{algpascal}
\usepackage{algc}
\usepackage{xcolor}
\usepackage{cite}

 \usepackage{cite}
\newtheorem{thm}{Theorem}
\newtheorem{defi}{Definition}

\newtheorem{exmpl}{Example}
\def\independenT#1#2{\mathrel{\rlap{$#1#2$}\mkern2mu{#1#2}}}

\title{Online Context-aware Data Release with Sequence Information Privacy}
\author{\IEEEauthorblockN{Bo~Jiang  \quad Ming~Li \quad Ravi~Tandon\\
}
\IEEEauthorblockA{Department of Electrical and Computer Engineering}\\
\IEEEauthorblockA{University of Arizona, Tucson, AZ, 85721\\
Email: $\{\textit{bjiang, lim, tandonr }\}$@email.arizona.edu}
}
\begin{document}
\maketitle
\begin{abstract}
Publishing streaming data in a privacy-preserving manner has been a key research focus for many years. This issue presents considerable challenges, particularly due to the correlations prevalent within the data stream. Existing approaches either fall short in effectively leveraging these correlations, leading to a suboptimal utility-privacy tradeoff, or they involve complex mechanism designs that increase the computation complexity with respect to the sequence length.
In this paper, we introduce Sequence Information Privacy (SIP), a new privacy notion designed to guarantee privacy for an entire data stream, taking into account the intrinsic data correlations. We show that SIP provides a similar level of privacy guarantee compared to local differential privacy (LDP), and it also enjoys a lightweight modular mechanism design.
We further study two online data release models (instantaneous or batched) and propose corresponding privacy-preserving data perturbation mechanisms. We provide a numerical evaluation of how correlations influence noise addition in data streams. Lastly, we conduct experiments using real-world data to compare the utility-privacy tradeoff offered by our approaches with those from existing literature. The results reveal that our mechanisms offer utility improvements more than twice those based on LDP-based mechanisms.

\end{abstract}
\section{Introduction}
In the era of big data, data sharing has become extensive and pervasive across various industries, transforming the way businesses and organizations operate. The data-sharing mechanisms play a critical role in enabling decision-making, analytics, and automation. The setting of these mechanisms can be broadly classified into two categories: offline and online. The offline setting considers static data/dataset, such as database queries, which involve accessing and utilizing stored data for various applications. The online setting, on the other hand,  often involves real-time processing and dissemination of data generated by IoT devices or cloud-based systems \cite{7867732}. These mechanisms encompass varied applications such as real-time heart rate monitoring via smartwatches \cite{medical}, instant updates from cloud-based infrastructure, and smart grid management for efficient energy distribution, etc. 

As the shared data may contain personally sensitive information, investigating data-sharing methods in a privacy-preserving manner becomes critical. Differential Privacy (DP)\cite{Dwork20061,Dwork2006,Dwork2008}, which is \textit{de facto} standard in the privacy research community, has achieved great success in the offline data-sharing setting and has led many real-world implementations, such as surveying demographics and commuting patterns \cite{Machanavajjhala:2008:PTM:1546682.1547184}, and the  2020 U.S. Census  \cite{inproceedingsCensus}. DP is well-suited for answering aggregated queries and requires a trusted server. Conversely, Local DP (LDP) mechanisms\cite{Dwork2008} allow for the publication of individual records without reliance on a trusted server. They can be used to answer both statistical and individual queries. LDP-based mechanisms have been successfully adopted by  Google's RAPPOR \cite{Rappor} for collecting web browsing behavior, and Apple's MacOS to identify popular emojis and media preferences in Safari \cite{Cormode:2018:PSL:3183713.3197390, LPSApple}. However, previous research has indicated that when independent $\epsilon$-LDP mechanisms are applied to correlated data, the actual leakage for each mechanism is significantly greater than $\epsilon$ for highly correlated data\cite{free_lunch, Kifer2012ARA, Kifer2014PufferfishAF} (the leakage upper bound is $k\epsilon$ when releasing $k$ consecutive correlated data points). 
A strict way to upper bound the privacy leakage is to properly allocate the global privacy budget  to each LDP-based mechanism by sequential composition. However, the privacy budget allocated to each mechanism may be too small to ensure an ideal utility, because LDP provides strong (worst-case) privacy guarantees and fails to \textit{leverage correlations} in their definitions. This oversight can potentially lead to less-than-optimal \textit{utility-privacy tradeoffs}.  

Context-aware privacy notions, which incorporate context information (typically the data distribution) in privacy definition, offer a more relaxed and adaptive way to measure privacy leakage\cite{context}. Mutual Information Privacy (MIP), for instance, gauges the mutual information between the raw data and its release\cite{7498650}. MIP evaluates the Kullback-Leibler (KL) divergence, a statistical measure that quantifies the expected distance between two distributions, thereby naturally incorporating the data's prior distribution and correlations. However, MIP provides an average-case privacy protection, which may not be sufficient in practice\cite{relation}. Moreover, MIP is not sequentially composable, making it unsuitable for the online setting. In such settings, the requirement for \textit{sequential composability} means that the decomposed privacy guarantee of each time step must remain independent of subsequent steps. Lastly, MIP-based mechanisms, which require averaging over the input/output's support, typically introduce substantial \textit{computational complexity} due to the exponential growth of the support with the data sequence length. Notably, as online data release demands the timely sharing of data, the privacy protection mechanisms need to be lightweight and have low computation complexity to avoid causing any delays in the data release.

\begin{figure*}[htp]
\begin{small}
\centering 
{ \includegraphics[width=0.65\textwidth]{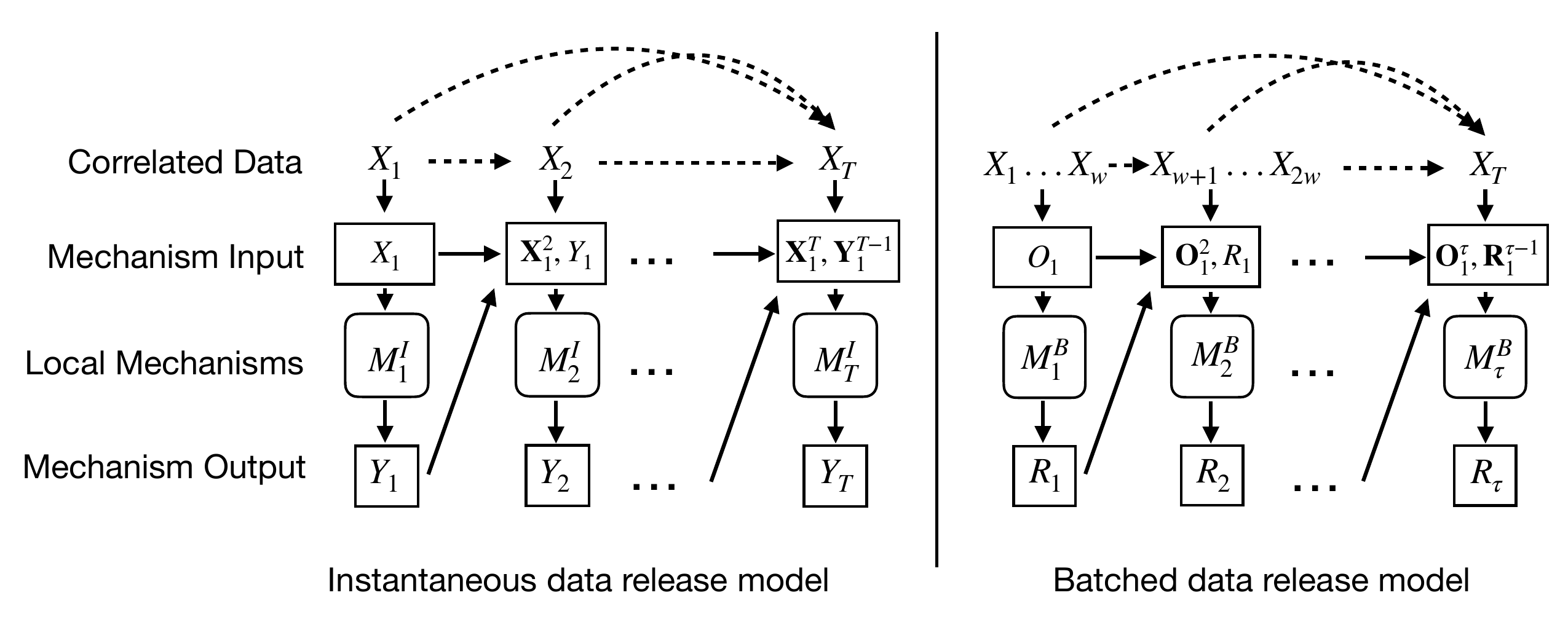}
\caption{Online privacy-preserving data release systems, a side-by-side comparison between the instantaneous data release model and the batched data release model.}}
\label{sys_model} 
\end{small}
\end{figure*}

Another context-aware privacy notion, Local Information Privacy (LIP)\cite{LIP1,LIP2}. LIP bounds the privacy leakage via the ratio between the posterior and the prior data distribution, which provides a worst-case privacy guarantee similar to LDP, while the utility of LIP-based mechanisms can be significantly enhanced compared to LDP (we provide a further comparison in section VI). However, LIP and LIP-based privacy mechanisms were originally proposed/designed for offline/one-time data release. In this paper, we extend it to handle the time series data in the online setting and propose Sequential Information Privacy (SIP). We show that SIP has sequential composition properties similar to LDP (while the privacy is bounded by a factor of 2 from LDP, meaning it also achieves strong privacy). At the same time, SIP enjoys modular mechanism design with low complexity. We develop novel mechanisms for two data release settings: instantaneous release setting,  where data is released as it is generated; and batched release setting, where data is accumulated and released periodically. Both approaches have their own merits and use cases. Our proposed mechanisms ensure real-time privacy preservation while maintaining data utility.

Our main contributions in this paper are three-fold:
\begin{itemize}
\item To quantify the privacy leakage in the data stream, we introduce a novel privacy notion, termed Sequence Information Privacy (SIP), which measures the overall privacy leakage in the data sequence. We consider two common real-world online data release settings: instantaneous data release and batched data release. Subsequently, we define two metrics, Instantaneous Information Leakage (IIL) and Batched Information Leakage (BIL), corresponding to the aforementioned settings. We show that SIP enjoys similar composition theorems as LDP, which is upper bounded by the sum of IIL and BIL in each time step, in accordance with the release setting in a linear or sub-linear fashion (advanced composition).
\item We propose privacy protection mechanisms corresponding to each data release setting. For the instantaneous release setting, we derive optimal mechanism and its parameters in closed form; furthermore, we demonstrate correlation-dependent noise through an example. For the batched release setting, we first show the problem can be degraded to a sub-optimal problem by simplifying the mechanism parameter. Then, we propose a data release  algorithm with simplified parameters based on gradient descent. We study the influence of batch size on data utility and computational complexity.
\item We provide extensive experimental results, utilizing two real datasets with different application types (and correspondingly, different utility measures). We evaluate the utility-privacy tradeoff provided by both mechanisms and compare these results with existing solutions. Our analysis shows that, while the privacy guarantee offered by $\epsilon$-SIP is strictly stronger than $2\epsilon$-LDP, the utility provided by $\epsilon$-SIP based mechanisms significantly exceeds that of $2\epsilon$-LDP.
\end{itemize}

\section{System and Threat Model}
\subsection{System Setup}
Let us consider the scenario of releasing time-series data in an online manner. Denote the raw data at each time stamp $k$ as $X_k$ that takes value from a finite support $\mathcal{X}$. Denote the data stream up to time $T$ as $\bold{X}_{1}^T=\{X_1,..., X_T\}$, and $\bold{x}_{1}^T$ as a realization of $\bold{X}_{1}^T$. We use the bold symbol  to denote a vector. In the context-aware setting, the data stream is considered as a correlated random vector with $\text{Pr}(X_1=x)=P_{1}(x)$, and $\text{Pr}(X_{k+1}=\bold{v}|\bold{X}_{1}^k=\bold{u})=C^{\bold{v}}_{\bold{u}}$, for all $u$, $v$ sequence. Further, we consider two types of data release scenarios: 1) instantaneously  release and 2) batched release. Instantaneous release means each of the data in the sequence is released instantaneously. One example of the instantaneous release is the navigation app on the smartphone, users are sending location data to the server and accessing location-based services on the fly. Another example is online games, users' operation data is collected by the server continuously (usually less than every 20 ms). On  the other hand, we also consider data to be released in a batched manner, for the applications where moderate delay is allowed to minimize the communication cost. Backend software in smartphones or PCs implements batched release by periodically sending collected logs to the server. This allows for efficient data management and analysis. Traffic monitoring systems utilize batched release to collect and upload aggregated traffic flow data at regular intervals. This approach ensures accurate monitoring of traffic density and enables effective analysis of traffic patterns. Models of these two release settings are depicted in Fig. 1.

\textbf{Instantaneous release setting:} In the instantaneous release setting, to protect data privacy, data at each time step (for example, time $k$) is perturbed to $Y_k$ before being released to the public. Assume that $Y_k$ takes a value of $y$ from the same domain as $X$, and the perturbation is done by a randomized mechanism $\mathcal{M}^I_k$, where the superscript $I$ in the notation denotes instantaneous release setting. The mechanism outputs $Y_k$ by considering the whole data sequence till $k$ as well as all previous outputs, i.e., $Y_k=\mathcal{M}_k^I(\bold{X}_1^k,\bold{Y}_1^{k-1})$.
Also, it is natural to assume that $X_{k+t}\independent{Y_k}|\{\bold{X}_{1}^{k},\bold{Y}_1^{k-1}\}$, for all $t\ge{1}$ (we use $\independent$ to denote 
 independent between random variables), as the current release should not depend on data at future time steps.  More specifically, the mechanism $\mathcal{M}^I_k$ is defined as follows 
\begin{equation}\label{action}
    \begin{aligned}
    a^I_{k}(y_{k}| x_{k}, &\bold{x}_{1}^{k-1}, \bold{y}_{1}^{k-1}) =\\ &\operatorname{Pr}\left(Y_{k}=y_{k} \mid \mathbf{X}_{1}^{k}=\mathbf{x}_{1}^{k}, \mathbf{Y}_{1}^{k-1}= \bold{y}_{1}^{k-1}\right).
    \end{aligned}
\end{equation}
% \begin{align}\label{action}
    
% \end{align}

\textbf{Batched release setting:} The raw data sequence is partitioned into different batches. Denote $\bold{O}_{l} = \bold{X}_{(l-1)w +1}^{lw}$ as one batch generated from the raw sequence. $w$ here denotes the length of the batch, and $l\in[1,\tau]$ represents the batch index. In this paper, we assume $w$ of each batch to be the same, however, it is straightforward to extend our analysis and results to the case where $w$ for each batch is different from each other.

% from the input and $\mathbf{O}$ as a batched input sequence, $\mathsf{B}$ as the total number of batches, then the sequence of $\bold{X}_{1}^{T}$ can be represented as $\bold{O}_1^{\tau}$, where $\bold{O}_{l} = \bold{X}_{(l-1)w}^{lw}$, and $w$ denotes the length of each batch.

Similarly to the instantaneous setting, for data privacy consideration, at time $l$, a privacy protection mechanism considers all previous input/output and releases a perturbed version $\bold{R}_l$. Denote the mechanism for batched release as $\mathcal{M}_l^B(\bold{O}_1^l,\bold{R}_1^{l-1})$, then $\bold{R}_l =\mathcal{M}_l^B(\bold{O}_1^l,\bold{R}_1^{l-1})$, it is also natural to assume that $\bold{O}_{l+t}\independent{\bold{R}_l}|\{\bold{O}_{1}^{l},\bold{R}_1^{l-1}\}$, for all $t>0$. As a result, $\mathcal{M}_l^B$ is defined as follows:
\begin{equation}\label{action2}
    \begin{aligned}
    a^B_{l}(\bold{r}_{l}| \bold{o}_{l}, &\bold{o}_{1}^{l-1}, \bold{r}_{1}^{l-1}) = \\
&\operatorname{Pr}\left(\bold{R}_{l}=\bold{r}_{l} \mid \mathbf{O}_{1}^{l}=\mathbf{o}_{1}^{l}, \mathbf{R}_{1}^{l-1}= \bold{r}_{1}^{l-1}\right).
    \end{aligned}
\end{equation}

\subsection{Adversary Model}
In this paper, the adversary can be anyone who has access to the released data. e.g. the server, or anyone in the public. We assume the adversary is honest but curious. He/She does not have access to the user's release system, and can only passively receive and observe the output sequence from the privacy-preserving mechanism. It is assumed that the adversary is interested in learning the raw data sequence. His/Her inference model stems from the Bayesian posterior probability distribution, given all historical observations. For the instantaneous release setting, for a set of observations $\bold{y}_{1}^{k-1}$, the adversary's belief of $\bold{X}_{1}^{k}$ is defined as:
\begin{align}
    \beta^I_{k}(\bold{x}_1^k|\bold{y}_1^{k-1}) = \operatorname{Pr}(\bold{X}_{1}^{k} = \bold{x}_{1}^{k} | \bold{Y}_{1}^{k-1} = \bold{y}_{1}^{k-1}). 
\end{align}
We denote this posterior belief as the belief state of the adversary in the instantaneous release setting.
For the batched release setting, similarly, the adversary's belief of $\bold{O}_{1}^{l}$ after receiving $\bold{r}_{1}^{l-1}$ is 
\begin{align}
    \beta^B_{l}(\bold{o}_1^l|\bold{r}_1^{l-1}) = \operatorname{Pr}(\bold{O}_{1}^{l} = \bold{o}_{1}^{l} | \bold{R}_{1}^{l-1} = \bold{r}_{1}^{l-1}),
\end{align}
which is defined as the adversary's belief state in the batched release setting. 

Besides, it is assumed that the adversary has the following abilities:
\begin{itemize}
    \item The adversary knows the initial prior distribution of the first data $P_1(\cdot)$ and the data correlation in the data sequence. Hence, the adversary's belief state can be updated from time to time. This is a common assumption used in nearly all information-theoretic approaches. Note that this is a worst-case assumption. According to \cite{LIP2}, if the adversary has a different knowledge from the true data prior distribution, the privacy leakage is decreased, and the reduced amount is proportional to the deviation from the true prior. 
    \item The adversary knows the privacy-protection mechanism, including the release setting, current data index (time step), and the perturbation parameters.
\end{itemize}

\begin{table*}[h!]
\centering
\renewcommand{\arraystretch}{2} % Adjusts the row height
\caption{Comparison of different Privacy Notions in the Online Setting}
\begin{tabular}{c|ccc|cc}
\hline
 & & \textbf{ Definition} & &~~~~~~~~~~~~~~~~~~~~~~~~~~~  \textbf{Mechanism} & \\
\hline
\textbf{Privacy Notion} & \textbf{Sequential Composability} & \textbf{Privacy Guarantee} & \textbf{Leveraging Correlation} & \textbf{Computation Complexity} & \textbf{Online Release}\\
\hline
LDP &Yes  &Worst-case  & No & $\mathcal{O}(|\mathcal{X}|)$\cite{Conti_release_LDP} & Yes\\
% \hline
Pufferfish &Yes, with Markovity   & Worst-case & Yes & $\mathcal{O}(T^3|\mathcal{X}|^3)$ \cite{Pufferfish_mec} & $No$\\
% \hline
MIP &{No}  & Average-case & Yes & $\mathcal{O}(|\mathcal{X}^{2T}|)$\cite{9715086} & Yes\\
% \hline
SIP (this paper) &Yes  & Worst-case & Yes &$\mathcal{O(}{|\mathcal{X}|})$ & Yes\\
\hline
\end{tabular}
\label{table:1}
\end{table*}

\section{Privacy Definition and comparison with existing privacy notions}
\subsection{Context-aware Sequence Information Leakages}
For data privacy, we start by defining the Sequence Information Leakage that occurs after a series of successive outputs. We then demonstrate how this leakage can be decomposed into each local leakage at various time steps.
\begin{defi}
The sequence information leakage (SIL) for releasing $\bold{Y}_1^T$ as privatized version of $\bold{X}_1^T$ is defined as: 
\begin{small}
\begin{equation}
\begin{aligned}
   \mathcal{L}(\bold{Y}_1^T\to{\bold{X}_1^T})=\max_{\bold{x}_1^T,\bold{y}_1^T\in{\mathcal{X}^T}}\left|\log\frac{\operatorname{Pr}(\bold{X}_1^T=\bold{x}_1^T|\bold{Y}_1^T=\bold{y}_1^{T})}{\operatorname{Pr}(\bold{X}^T_1=\bold{x}_1^T)}\right|.
\end{aligned}
\end{equation}
\end{small}
\end{defi}
This can be interpreted as the adversary's maximum information gain about $\bold{X}_1^T$ after observing the output sequence of $\bold{Y}_1^T$ compared to his prior knowledge. We say a privacy-preserving mechanism satisfies $(\epsilon,\delta)$-Sequence Information Privacy (SIP) if the following condition holds:
\begin{equation}
    \text{Pr}(\mathcal{L}(\bold{Y}_1^T\to{\bold{X}_1^T})> \epsilon)<\delta.
\end{equation}

We next define the instantaneous information leakage at each time step. 
\begin{defi}
The instantaneous information leakage (IIL) at time $k$ is defined as: 
\begin{small}
\begin{equation}
   \mathcal{L}({Y}_k\to{\bold{X}_1^k})=\max_{\bold{x}_1^k\in\mathcal{X}^k,y_k\in\mathcal{X}}\left|\log\frac{\operatorname{Pr}(\bold{X}^k_1=\bold{x}^k_1|\bold{Y}_1^k=\bold{y}_{1}^{k})}{\operatorname{Pr}(\bold{X}_1^k=\bold{x}_1^k|\bold{Y}_1^{k-1}=\bold{y}_{1}^{k-1})}\right|.
\end{equation}
\end{small}
\end{defi}
The operational meaning of the instantaneous information leakage can be interpreted as the adversary's additional belief on the data sequence $\bold{X}_1^k$ after observing $Y_k$ at time $k$ compared to the belief before taking this observation. Such a definition is presented based on an online data release manner, i.e., before observing $Y_k$, $k-1$
outputs have already been published. 

Similarly, we define the information leakage in the setting of batched data release:

\begin{defi}
The batched information leakage (BIL) when releasing the $l$-th batched data sequence is defined as: 
\begin{small}
\begin{equation}\label{eq:BIL}
   \mathcal{L}(\bold{R}_l\to{\bold{O}_1^l})=\max_{\bold{o}_1^l\in\mathcal{X}^{lw},\bold{r}_l\in\mathcal{X}^w}\left|\log{\frac{\operatorname{Pr}(\bold{O}^l_1=\bold{o}^l_1|\bold{R}_1^l=\bold{r}_{1}^{l})}{\operatorname{Pr}(\bold{O}_1^l=\bold{o}_1^l|\bold{R}_1^{l-1}=\bold{r}_{1}^{l-1})}}\right|.
\end{equation}
\end{small}
\end{defi}
With similar operational meaning as the instantaneous information leakage: the adversary's additional knowledge on $\bold{O}_1^l$ after observing $\bold{R}_l$.

Regarding the relationship between the IIL/BIL and SIL we have the following theorem:
\begin{thm}\label{thm1}
For a sequence of instantaneous-release privacy protection mechanisms $\mathcal{M}^I(1:T)$, such that $\mathcal{M}^I_k$ releases $Y_k$ at time $k$, if the IIL $\mathcal{L}(Y_k\to{\bold{X}_1^k})\le{\epsilon_k}$, $\forall{k}\in{1,2,..,T}$, then $\mathcal{M}^I(1:T)$ satisfies (${\sum_{k=1}^T\epsilon_k},0$)-SIP. Similarly, for a sequence of batched-release privacy protection mechanisms $\mathcal{M}^B(1:\tau)$, if each BIL satisfies $\mathcal{L}(\bold{R}_l\to{\bold{O}_1^l})\le{\epsilon_l}$, $\forall{l}\in{1,2,..,\tau}$, then $\mathcal{M}^B(1:\tau)$ satisfies (${\sum_{l=1}^{\tau}\epsilon_l},0$)-SIP.
\end{thm}

Theorem 1 posits that the privacy budget, as it pertains to a sequence of mechanisms, decomposes linearly in relation to the amount of data disclosed. Additionally, the aggregation of local leakages contributes to the global sequence leakage. Importantly, the introduction of a minor failure probability, denoted as $\delta$, allows for the achievement of sub-linear growth in the privacy budget. This property is summarized in the following Theorem.

% Which means if the instantaneous leakage at each time $k$ is bounded by $\epsilon_k$, $\forall{k}\in{1,2,..,T}$, the sequence leakage is bounded by $\sum_{k=1}^T\epsilon_k$.
% Similarly, the relationship between batched information leakage and sequence information leakage is shown in the following proposition.

% Suppose at each time step $k$, the information is released by mechanism $\mathcal{M}_k(\bold{x}_1^k,\bold{y}_1^{k-1})$, which can be represented by the probability of $\text{Pr}(Y_{k}=y_k|\bold{X}_1^k=\bold{x}_1^k,\bold{Y}_1^{k-1}=\bold{y}_1^{k-1})$.  Define the perturbation probabilities $\text{Pr}(Y_{k}=y_k|\bold{X}_1^k=\bold{x}_1^k,\bold{Y}_1^{k-1}=\bold{y}_1^{k-1})$ as $ a_{k}(y_{k}| x_{k}, x_{1}^{k-1}, y_{1}^{k-1}) $.

\begin{thm}\label{coro1}
    The sequence of mechanisms $\mathcal{M}^I(1:T)$ in Theorem \ref{thm1} satisfies $(T\epsilon(e^{\epsilon}-1)+\sqrt{T}\epsilon\sqrt{2\ln(1/\delta)},\delta)$-SIP; similarly, the mechanisms $\mathcal{M}^B(1:\tau)$ satisfy $(\tau\epsilon(e^{\epsilon}-1)+\tau\epsilon\sqrt{2\ln(1/\delta)},\delta)$-SIP.
\end{thm}

{Remark: Theorem 2 is similar to the sequential composition of LDP, and the proof is done in a similar way.} Theorem 1 and Theorem 2 effectively break down a global task into manageable local goals. Specifically, to design either the instantaneous $\mathcal{M}^I$ or batched $\mathcal{M}^B$ mechanism at each moment under a total SIP budget, it suffices to limit the IIL or BIL. Detailed proof of Theorem 1 and Theorem 2  are provided in the appendix.

% \section{Comparison with existing  privacy notions}
% In this section, we compare the privacy notions proposed in this paper to other existing ones.

\subsection{Comparison with Existing Privacy Notions}

\textbf{Local Differential Privacy:} The decentralized version of DP, Local Differential Privacy (LDP)\cite{Dwork2008}, has gained much attention since its introduction. It adopts a similar structure as Differential Privacy but considers the input as each individual's data, so the privacy-utility tradeoff of each individual is customizable. The definition of LDP, when adapted to the sequential data release model, can be summarized as follows:

\begin{defi}
    A privacy protection mechanism $\mathcal{M}$, is said to be $\epsilon$-local differentially private for the sequence of $\bold{X}_1^T$, if for all $\bold{x}_1^T, \tilde{\bold{x}}_1^T\in\{\mathcal{X}\}_1^T$, and for all $\bold{y}_1^T\in \text{Range}(\mathcal{M}(\bold{X}_1^T))$, 
    \begin{equation}\label{eq:dp}
        \frac{\text{Pr}(\bold{Y}_1^T=\bold{y}_1^T|\bold{X}_1^T=\bold{x}_1^T)}{\text{Pr}(\bold{Y}_1^T=\bold{y}_1^T|\tilde{\bold{X}}_1^T=\tilde{\bold{x}}_1^T)}\in{[e^{-\epsilon},e^{\epsilon}]}.
    \end{equation}
\end{defi}

The relationship between $\epsilon$-LDP and $\epsilon$-SIP is shown in the following Theorem:
\begin{thm}
    If a privacy-preserving mechanism $\mathcal{M}$ satisfies $\epsilon$-LDP, then it satisfies $\epsilon$-SIP. Conversely, if $\mathcal{M}$ satisfies  ${\epsilon}$-SIP, then it satisfies $2\epsilon$-LDP.
\end{thm}

LDP provides worst-case privacy protection, which means the condition in \eqref{eq:dp} must hold for every possible $\bold{x}_1^T, \tilde{\bold{x}}_1^T,$ and $\bold{y}_1^T$. {Note that, the $2 \epsilon$ bound from LDP is in fact an advantage of the SIP notion since it also bounds worst-case leakage while considering context (close to the privacy guarantee of LDP), however, the utility-privacy tradeoff is much better than LDP due to the exploitation of data prior/context in our mechanisms.} LDP is also sequentially decomposable by applying independent mechanisms \cite{Composition}, and can be adapted in the online release setting. Formally, for a sequence of privacy-preserving mechanisms $\mathcal{M}(1:T)$, if each mechanism $\mathcal{M}$ satisfies $\epsilon_k$-LDP for $X_k$. Then $\mathcal{M}(1:T)$ satisfies $\sum_{k=1}^T \epsilon_k$-LDP for $\bold{X}_1^T$\cite{Dwork_DPcontinualOb}.
 When implementing LDP with sequential composition, the total privacy budget could drastically inflate if each local mechanism is using a budget that is reasonable for ensuring utility. Conversely, the utility of each local data release could significantly diminish if a more restrained budget is applied to maintain robust privacy protection.

 Recent studies have begun to incorporate data correlation into their model assumptions {to enhance the utility-privacy tradeoff provided by the LDP mechanism}. For example, Cao et al. proposed a method to quantify DP leakage for data with temporal correlations \cite{Quantify_DP}. They developed DP mechanisms for location aggregations under temporal correlation \cite{DP_mec}, under the assumption that data correlation can be discerned by an adversary. Subsequently, they define a subset of the output's support to decrease the DP sensitivity. However, such an approach may result in an increased failure probability of DP, also the correlation itself is still not leveraged in the mechanism design. Further studies \cite{Quantify_LP, prolong_DP, 6949097, CGM, 9488899} also explored scenarios where adversaries possess knowledge of data correlations. These works primarily concentrate on creating novel mechanisms to protect a single user's privacy by extending DP. Nevertheless, unlike the setting of this paper, these studies consider sequential data release in an offline setting.
In contrast, Wang et al. put forward an online data release mechanism in \cite{Conti_release_LDP}, satisfying either DP or LDP, depending on the specific circumstances. However, still, their proposed noise-adding mechanisms for each instantaneous data point operate independently.

% \subsection{Comparison with }

\textbf{Pufferfish Privacy:} Another privacy notion that provides privacy protection over a set of self-defined secrets is Pufferfish privacy. When adapted into the data release model, Pufferfish privacy is defined as:

\begin{defi}($\epsilon$-Local Pufferfish Privacy \cite{Kifer2012ARA})
Given set of potential secrets $\mathbb{S}$, a set of discriminative pairs $\mathbb{S}_{pairs}$, a set of data evolution scenarios $\mathscr{P}_{\{\mathcal{X}\}_1^T,\mathcal{S}}$, and a privacy parameter $\epsilon\in{{R}^+}$,
an (potentially randomized) algorithm $\mathcal{M}$ satisfies $\epsilon$-PufferFish ($\mathbb{S}$, $\mathbb{S}_{pairs}$, $\mathscr{P}_{\{\mathcal{X}\}_1^T,\mathcal{S}}$) privacy if
\begin{itemize}
    \item for all possible outputs $\bold{y}_1^T \in \text{range}(\mathcal{M})$,
    \item for all pairs $(s_i,s_j)\in\mathbb{S}_{pairs}$ of potential secrets,
    \item for all distributions $\textbf{P}_{\{\mathcal{X}\}_1^T,\mathcal{S}}\in \mathscr{P}_{\{\mathcal{X}\}_1^T,\mathcal{S}}$ for which $\operatorname{Pr}(s_i| \textbf{P}_{\{\mathcal{X}\}_1^T,\mathcal{S}})\neq{0}$ and $\operatorname{Pr}(s_j|\textbf{P}_{\{\mathcal{X}\}_1^T,\mathcal{S}})\neq{0}$
\end{itemize}
the following holds:
\begin{equation*}
    e^{-\epsilon}\le\frac{\operatorname{Pr}(\mathcal{M}(\bold{X}_1^T)=\bold{y}_1^T|\textbf{P}_{\{\mathcal{X}\}_1^T,\mathcal{S}},s_i)}{\operatorname{Pr}(\mathcal{M}(\bold{X}_1^T)=\bold{y}_1^T|\textbf{P}_{\{\mathcal{X}\}_1^T,\mathcal{S}},s_j)}\le{e^{\epsilon}}.
\end{equation*}
\end{defi}
The relationship between Pufferfish Privacy and SIP isn't directly deducible as they operate under different assumptions. Pufferfish assumes the possibility of multiple data generation scenarios, captured by $\textbf{P}_{{\mathcal{X}}_1^T,\mathcal{S}}$. Conversely, SIP presumes that the correlation among data is given. In the context of streaming data release, these correlations can naturally be learned from previous releases, thus SIP is more suitable for the online release setting.

Pufferfish privacy further distinguishes itself by protecting a latent variable $S$ that correlates with the data stream, as opposed to the data itself. This model has been the subject of other studies, primarily from an information-theoretic perspective, such as those in \cite{7918623, 8262832, 8438536}. Contrarily, SIP postulates that each individual data value is privately sensitive. As a result of this assumption, Pufferfish isn't as intuitively decomposable like LDP and SIP, { without adopting additional assumptions of Markovity in the data sequence.}

Among the existing literature on Pufferfish Privacy, \cite{Pufferfish_mec} proposed a privacy protection mechanism that also considers data correlation. The proposed mechanism operates under the assumption that the data distribution adheres to the Markov Quilt properties. This premise simplifies data correlation and reduces computational complexity, while also allowing the mechanism to be sequentially composable. However, the data release mechanism only functions when the secrets are sequentially obtained, whereas the data sequence is predetermined. {This assumption renders the mechanism unsuitable for an online setting. Furthermore, the algorithm necessitates an exhaustive search of all possible combinations of proximate and distant nodes in the Markov Quilt, resulting in only a modest reduction in complexity.}

\textbf{Information Theoretical Approaches:} We next compare SIP with privacy notions borrowed from information theory. The first definition to compare with is Mutual Information Privacy (MIP), which is measured by the mutual information between the input and the output of the privacy-protection mechanism:
\begin{defi}\cite{7498650}
    A mechanism $\mathcal{M}$ satisfies $\epsilon$-MIP for some $\epsilon\in{\mathbf{R}^+}$, if the mutual information between $\bold{X}_1^T$ and $\bold{Y}_1^T$ satisfies $I(\bold{X}_1^T;\bold{Y}_1^T)\le{\epsilon}$.
\end{defi}
It is evident that a privacy-preserving mechanism $\mathcal{M}$ based on SIP, guaranteeing $\mathcal{L}(\bold{Y}_1^T\to{\bold{X}_1^T})\le {\epsilon}$, would also ensure that the mutual information $I(\bold{X}_1^T;\bold{Y}_1^T)\le{\epsilon}$. This is due to the mutual information being a statistical average of the SIP. MIP provides a relatively weak average-case privacy guarantee. { Also, MIP is not sequentially composable: after applying the chain rule, $I(\bold{X}_1^T;\bold{Y}_1^T)\le{\epsilon}$ can only be decomposed into bounding each $I(\bold{X}_1^T; Y_t|\bold{Y}_1^{t-1})$, which depends on the whole input sequence.}  Nevertheless, as mutual information or conditional mutual information intuitively captures data correlation and can be easily decomposed using the chain rule, MIP has been extensively explored as a privacy metric for time-series data, as demonstrated in \cite{10.1145/2714576.2714577}.

Several studies on privacy-preserving online data release, such as \cite{9153837,9715086}, have also capitalized on data correlation, releasing aggregated location data online while ensuring individual or group privacy measured by MIP is constrained. As mutual information functions are convex, they can conveniently be integrated into optimization problems. The obfuscation mechanism at different time stamps is selected by a reinforcement learning algorithm. {However, the complexity is still high depending on the length of the sequence.}

SIP can be perceived as an online sequential version of local Information Privacy (LIP)\cite{Centrlized_IP, LIP1,LIP2}, defined as the ratio between the prior and posterior probabilities following observation of the output from the mechanism. In \cite{LIP1}, comprehensive results on mechanism design based on LIP are offered, showing superior utility-privacy trade-offs compared to LDP, MIP, and Pufferfish. However, none of these mechanisms take into account the context of online sequential data release. 

Considering the four typical challenges for sequential data release in the online setting: utility-privacy tradeoff, sequential composability, leveraging correlation, and low computation complexity, we summarize different privacy notions described in this section and how these notions address the four challenges in Table \ref{table:1}.

%%%%%%%%%%%%%%%%%%%%%%%%%%%%%%%%%%%%%%%%%%%%%

\section{Utility Privacy Tradeoff for Instantaneous Release}
\subsection{Problem Formulation}

In this paper, we define data utility as the expected distance between the arbitrary Quality of Service (QoS) function of the input ($Q(X)$) and output ($Q(Y)$). In the instantaneous release model, the utility is gauged by the expected distance between $Q(X_k)$ and $Q(Y_k)$ at each time stamp $k$. This measurement is known as the Instantaneous Expected Distance (IED):

\begin{equation}\label{eq:IED}
E\left[\mathsf{D}(Q(X_k),Q(Y_k)\big|\bold{Y}_1^{k-1}=\bold{y}_1^{k-1}\right].
\end{equation}

In \eqref{eq:IED}, $Q$ signifies the query function of $X$ and $Y$ that depends on the specific application, while $\mathsf D(a,b): ({R},{R})\to{{R}^+}$ denotes a distortion or distance measure between $a, b$. The expectation ${E}[\cdot]$ is taken over both the underlying distribution of the data $P_X(x)$, as well as over the randomness of the mechanism.
The expected distance between $Q(X)$ and $Q(Y)$ represents a general type of utility measurement. For instance, in a Location-Based Service (LBS), $Q(X)=X$, and the Euclidean distance between $X$ and $Y$ is typically used to gauge performance: $\textsf{Utility}= -{E}[(X-Y)^2]$. Another example is histogram estimation, where the aim is to ascertain how many people belong to each data category or classification according to users' data value. In this case, $Q(X)$ is an indicator function, and the absolute distance utility function can be written as $\textsf{Utility}= -\sum_{i=1}^K {E}[|\mathbbm{1}_{{X\in{S_{i}}}}-\mathbbm{1}_{{Y\in{S_{i}}}}|]$. These examples demonstrate that for different applications, the utility function can be adapted by modifying the $Q$ function and the distortion function $\mathsf D(\cdot,\cdot)$.

The online mechanism we explore in this paper concentrates on minimizing the IED defined in Eq. \eqref{eq:IED}, subject to IIL constraints i.e., the problem is defined as:

\begin{equation}\label{tradeoff1}
    \begin{aligned}
        & \min E\left[\mathsf{D}(Q(X_k),Q(Y_k)\big|\bold{Y}_1^{k-1}=\bold{y}_1^{k-1}\right],\\
        & \text{Such that} ~\mathcal{L}({Y}_k\to{\bold{X}_1^k})\le{\epsilon_l}, ~\forall{k\in[1, T]}.
    \end{aligned}
\end{equation}

% In the following sections, we illustrate the methodology for designing mechanisms that optimize the utility-privacy tradeoff.
% Firstly, we introduce the Conditional Randomized Response (CRR) perturbation mechanism and demonstrate how to apply this mechanism to limit each local leakage. 
% Finally, we leverage an example to elucidate the core concepts and underlying intuitions.
% \subsection{Utility-Privacy tradeoff for instantaneous setting}

\textbf{Privacy metric:} By Bayes rule, the privacy metric in the IIL can be expressed as:
\begin{equation}\label{cons}
\begin{aligned}
&\frac{\text{Pr}(\bold{X}_1^k=\bold{x}_1^k|\bold{Y}_1^k=\bold{y}_1^{k})}{\text{Pr}(\bold{X}^k_1=\bold{x}_1^k|\bold{Y}_1^{k-1}=\bold{y}_1^{k-1})}\\
    =&\frac{\text{Pr}(Y_k=y_k|\bold{X}_1^k=\bold{x}_1^k,\bold{Y}_1^{k-1}=\bold{y}_1^{k-1})}{\text{Pr}(Y_k=y_k|\bold{Y}_1^{k-1}=\bold{y}_1^{k-1})}\\
    =&\frac{ a^I_{k}(y_{k}| \bold{x}_{1}^{k}, \bold{y}_{1}^{k-1})}{\sum_{\bold{\tilde{x}}_1^k\in{\mathcal{X}^k}}a^I_{k}(y_{k}| \bold{\tilde{x}}_1^k, \bold{y}_{1}^{k-1})\beta^I_k(\bold{\tilde{x}}_1^k|\bold{y}_1^{k-1})}.
\end{aligned}
\end{equation}
Note that, the term $\beta_k(\bold{\tilde{x}}_1^k|\bold{y}_1^{k-1}) = \text{Pr}(\bold{X}_1^k=\bold{\tilde{x}}_1^k|\bold{Y}_1^{k-1}=\bold{y}_1^{k-1})$ can be further expressed as:

\begin{small}
\begin{equation}\label{eq:belief_update}
    \begin{aligned}
        &\text{Pr}(\bold{X}_1^k=\bold{\tilde{x}}_1^k|\bold{Y}_1^{k-1}=\bold{y}_1^{k-1})\\
        =&\text{Pr}({X}_k=\tilde{x}_k|\bold{X}_1^{k-1}=\bold{\tilde{x}}_1^{k-1})\text{Pr}(\bold{X}_1^{k-1}=\bold{\tilde{x}}_1^{k-1}|\bold{Y}_1^{k-1}=\bold{y}_1^{k-1})\\
        =&C^{\tilde{x}_k}_{\bold{x}_1^{k-1}}\frac{ a^I_{k-1}(y_{k-1}|\bold{\tilde{x}}_{1}^{k-1}, y_{1}^{k-1})\beta^I_{k-1}(\bold{\tilde{x}}_1^{k-1}|\bold{y}_1^{k-2})}{\sum_{\bold{\bar{x}}_1^{k-1}}a^I_{k-1}(y_{k-1}|\bold{{\bar{x}}}_{1}^{k-1}, y_{1}^{k-2})\beta^I_{k-1}(\bold{{\bar{x}}}_1^{k-1}|\bold{y}_1^{k-2})},\\
    \end{aligned}
\end{equation}
\end{small}
where the nominator is the perturbation probability, and the denominator is a linear combination of the perturbation parameters with coefficients of some calculable posteriors. 

\textbf{Utility function:} The utility function can be further expressed as:
\begin{equation}\label{utility_IED}
\begin{aligned}
&E\left[D(Q(X_k)-Q(Y_k))|\big|\bold{Y}_1^{k-1}=\bold{y}_1^{k-1}\right]\\
% =&\sum_{x,y}D(Q(x)-Q(y))\sum_{\bold{x}_1^{k-1}}\text{Pr}(X_k=x,\bold{X}_1^{k-1}=\bold{x}_1^{k-1},Y_k=y_k|\bold{Y}_1^{k-1}=\bold{y}_1^{k-1})\\
% =&\sum_{x,y}D(Q(x)-Q(y))\sum_{\bold{x}_1^{k-1}}\text{Pr}(Y_k=y_k|\bold{Y}_1^{k-1}=\bold{y}_1^{k-1},X_k=x,\bold{X}_1^{k-1}=\bold{x}_1^{k-1})\text{Pr}(X_k=x,\bold{X}_1^{k-1}=\bold{x}_1^{k-1}|\bold{Y}_1^{k-1}=\bold{y}_1^{k-1})\\
=&\sum_{x,y}D(Q(x)-Q(y))\\
&~~~\cdot\sum_{\bold{x}_1^{k-1}}a^I_k(y_k|\bold{x}_1^{k-1},x,\bold{y}_1^{k-1})\beta^I_k(\bold{x}_1^{k-1},x|\bold{y}_1^{k-1}).\\
\end{aligned}
\end{equation}

As a result, the utility-privacy tradeoff can be expressed as:
\begin{small}
\begin{equation}\label{tradeoff12}
    \begin{aligned}
        & \min \eqref{utility_IED},~
         \text{Such that} ~\eqref{cons}\in[e^{-\epsilon_k},e^{\epsilon_k}].
    \end{aligned}
\end{equation}
\end{small}
Hence, the local optimization problem can be represented as a function of the data correlation, previous data release policy, and belief state.

% which means considering historical $\bold{X}_1^{k-2}$ will not affect the values of $\beta^s_k({x}_{k-1},x|\bold{y}_1^{k-1})$. as a result, for any $a_k(y_k|\bold{x}_1^{k-1},x,\bold{y}_1^{k-1})$, there is an policy $a^s_k(y_k|{x}_{k-1},x,\bold{y}_1^{k-1})$, such that:
% \begin{equation}
%     \text{IEAD}(a_k(y_k|\bold{x}_1^{k-1},x,\bold{y}_1^{k-1}))=\text{IEAD}(a^s_k(y_k|{x}_{k-1},x,\bold{y}_1^{k-1})).
% \end{equation}
% Hence, simplifying $a_k(y_k|\bold{x}_1^{k-1},x,\bold{y}_1^{k-1})$ to $a^s_k(y_k|{x}_{k-1},x,\bold{y}_1^{k-1})$ does not affect the utility. This completes the proof of theorem 2.

\subsection{Conditional Randomize Response Mechanism:}

Next, we're introducing the Conditional Randomized Response (CRR) perturbation mechanism. Then, we derive our optimal mechanism based on CRR.

In a context-aware setting, the input to a privacy-preserving mechanism consists of the data values and their probabilistic distribution. For the online release setting, this distribution isn't fixed but instead depends on all previously observable releases. Therefore, the data distribution can be symbolized as the belief state defined earlier. This observation leads us to the CRR mechanism, which takes as input the current data value, as well as the belief state.  

An illustrative example of the CRR mechanism is shown in Fig. 2. This mechanism perturbs $X_1$ according to probabilities designed based on the adversary's prior knowledge. After $Y_1$ is output, the adversary's prior knowledge about $X_2$ shifts to $\text{Pr}(X_2|Y_1)$. Hence, the mechanism designs the perturbation parameters at time 2 in accordance with the distribution of $\text{Pr}(X_2|Y_1)$.

\begin{figure}[t]
    \centering
    \includegraphics[width=0.45\textwidth]{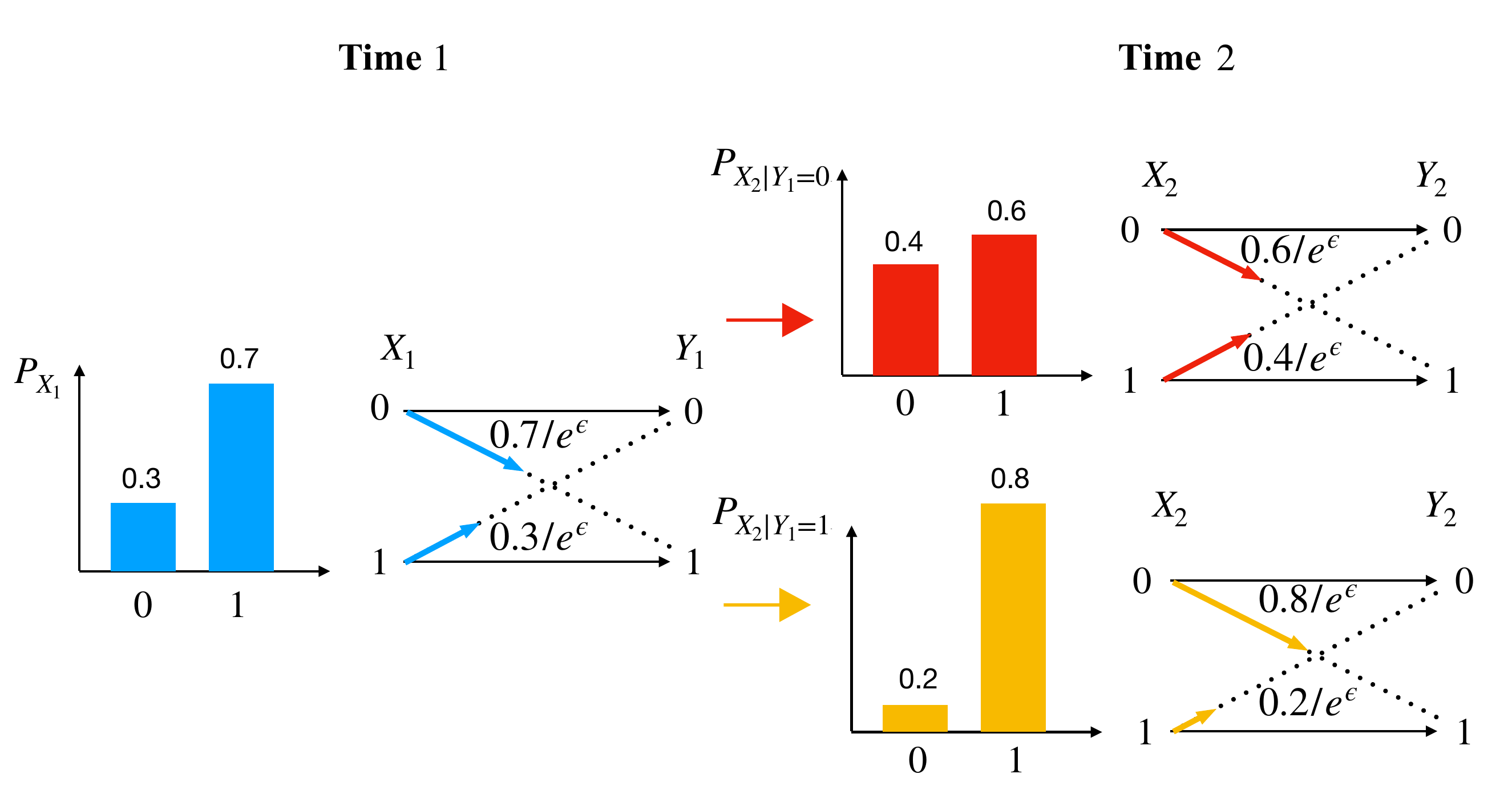}
    \label{Fig_CRR}
    \caption{Illustration of Conditional Randomize Response (CRR) mechanism, each mechanism parameters depend on data prior as well as previous release.}
\end{figure}

% As the perturbation probabilities at each time step depends on previous outputs as well as the data correlations, as a result, the noise required at each time step are also correlated. Define the noise at time $k$ as $N_k(X_1^k,Y_1^{k-1})$, thus in a discrete data horizon with cardinality of $|\mathcal{D}|$, the perturbation mechanism can be rewritten as a noisy channel:
% \begin{equation*}
%   Y_k=X_k+N_k(X_1^k,Y_1^{k-1}) \mod{|\mathcal{D}|}.
% \end{equation*}
% Equivalently, $N_k$ can be viewed as a random variable with the support of $\mathcal{D}$, with $\text{Pr}(N_k=j)$ as  $\sum_{i=0}^\infty\text{Pr}(N_k=i|\mathcal{D}|+j)$. Define the noise power as $|N_k|=E[N^2_k]=\sum_{i=0}^{|\mathcal{D}|}i^2Pr(N_k=i)$.

% Intuitively, to achieve optimal utility, the mechanism should add least amount of noise while satisfying the privacy constraints. Observe that to minimize $|N_k|$ is equivalent to minimize all the probabilities of $\text{Pr}(N_k=i)$ where $i\neq{0}$. taking the value of $N_k$ in, the optimal parameters at time $k$ follows the optimization problem of,  $\forall{i\in[1,|\mathcal{D}-1|]}$:
% \begin{equation}\label{opt}
% \begin{aligned}
% &\min q(x_k+i\mod |\mathcal{D}|)_{\bold{x}_1^k}^{\bold{y}_1^{k-1}};\\  
%  \text{s.t.}~&Eq.\eqref{cons},\\
%  &q(x_k+i\mod |\mathcal{D}|)_{\bold{x}_1^k}^{\bold{y}_1^{k-1}}\in[0,1],\\
%  &\sum_{i=0}^{|D|-1}q(x_k+i\mod |\mathcal{D}|)_{\bold{x}_1^k}^{\bold{y}_1^{k-1}}=1.
%  \end{aligned}
% \end{equation}

We next derive the closed-form optimal solutions which minimize IED at each time step, which follows the next theorem.
\begin{thm}\label{thm3}
For the class of utility function {with distance $D$} satisfying the following properties: 1. non-negativity: $D(X,Y)\ge 0$; 2. Identity of Indiscernibles: $D(X,X)=0$, 3. Symmetry: $D(X,Y)=D(Y,X)$, and 4. Triangle Inequality: $D(X,Y)\le D(X,Z) + D(Z,Y)$, the following perturbation parameters at time $k$ is the optimal solution of the problem defined in \eqref{tradeoff12},
$$  a^I_{k}(y_k| \bold{x}_{1}^k, \bold{y}_{1}^{k-1})  =\left\{
\begin{aligned}
&\frac{\beta^I(x_{k}|\bold{y}_1^{k-1})}{e^{\epsilon_k}}  & ~\text{if}~ y_k\neq x_k \\
&1-\frac{1-\beta^I(x_{k}|\bold{y}_1^{k-1})}{e^{\epsilon_k}}  & ~\text{if}~ y_k=x_k \\
\end{aligned}
\right.$$
\end{thm}
% \begin{prop}
% The optimal solutions also apply to the general utility definition.
% \end{prop}
Detailed proof is provided in the appendix in the supplementary document.
Key insights of the optimal solutions:
\begin{itemize}
    \item The optimal perturbation parameters are found at the boundaries of the privacy constraints, so that the least amount of noise is added while the privacy can be protected.
    \item The conditional probability of $Y_k$ is identical to that of $X_k$ given $\bold{Y}^{k-1}_1$, i.e., $\text{Pr}(X_k=x_k|\bold{Y}^{k-1}_1=\bold{y}^{k-1}_1)=\text{Pr}(Y_k=x_k|\bold{Y}^{k-1}_1=\bold{y}^{k-1}_1)$, which means the output at each time step always has the same marginal distribution of the input data.
    \item The mechanism does not depend on the historical raw data sequence. i.e. $\bold{x}_1^{T-1}$. This is due to the following two aspects: 1. the data utility at time $k$ only depends on $X_k$. 2. A simplification from $\bold{x}_1^{T-1}$ won't violate the privacy constraints.
\end{itemize}
\setlength{\textfloatsep}{0pt}
\begin{algorithm}[tbh]
	\small
	\caption{{SIP mechanism for instantaneous release}}
	\label{algo1}
	\begin{algorithmic}[1]
	\item Input: current time $k$, initial prior $P_1(x)$, data correlation $C^{\bold{u}}_{\bold{v}}$, historical release sequence $\bold{Y}_1^{k-1}$, current $X_k$, previous $\beta^I(x_{k-1}|\bold{y}_1^{k-2})$.
	\item Output: Instantaneous release $Y_k$
	\item if $k = 1$: 
        \item ~~~~Release $Y_1=y_1$ according to the following rule:
        \item ~~~~$y_1= x_1$ w.p. $1-(1-P_1(y_1))/e^{\epsilon}$;
        \item ~~~~other $y_1$ w.p. $P_1(y_1)/e^{\epsilon}$ .
	\item else: 
        \item ~~~~for all $x_k \in \mathcal{X}$:
        \item ~~~~~~~~update $\beta^I(x_{k}|y_{1}^{k-1})$ according to \eqref{eq:belief_update}:
	\item ~~~~Release $Y_k$ according to the following rule:
        \item ~~~~$y_k = x_k$ w.p.  $1-(1-\beta^I(y_k|\bold{y}_1^{k-1}))/e^{\epsilon}$, 
        \item ~~~~other $y_k$ w.p. $ \beta^I(y_k|\bold{y}_1^{k-1})/e^{\epsilon}$.
	\end{algorithmic}
\end{algorithm}

Given the optimal mechanism parameters, the algorithm for instantaneous release is shown in Alg. \ref{algo1}. As the calculating of the belief state takes constant time complexity, the computation complexity of Alg. \ref{algo1} is $\mathcal{O}(|\mathcal{X}|)$.

\begin{figure*}[htp]
\begin{small}
\centering 
{ \includegraphics[width=0.23\textwidth]{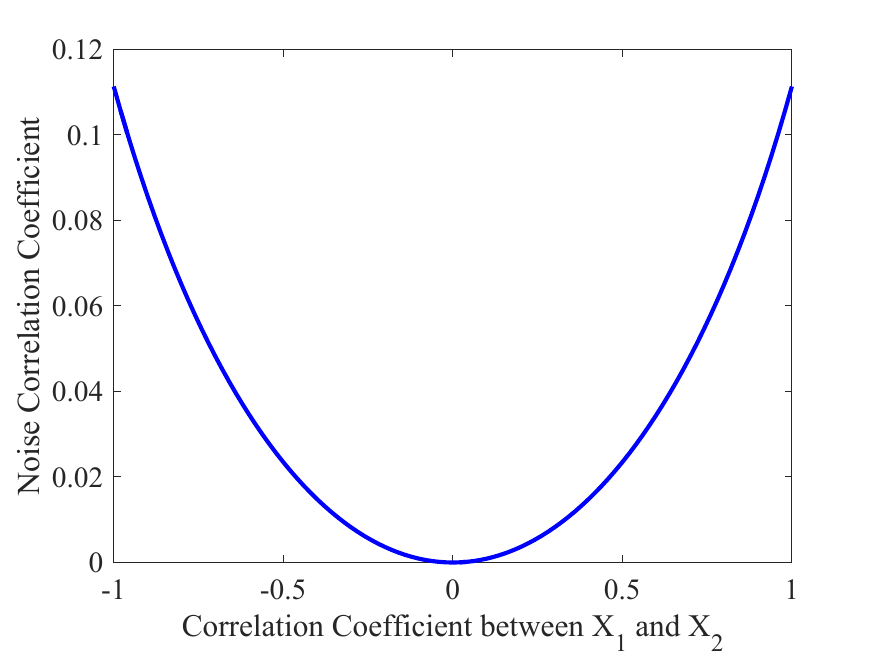} 
\label{Fig_corr} } 
{ \includegraphics[width=0.23\textwidth]{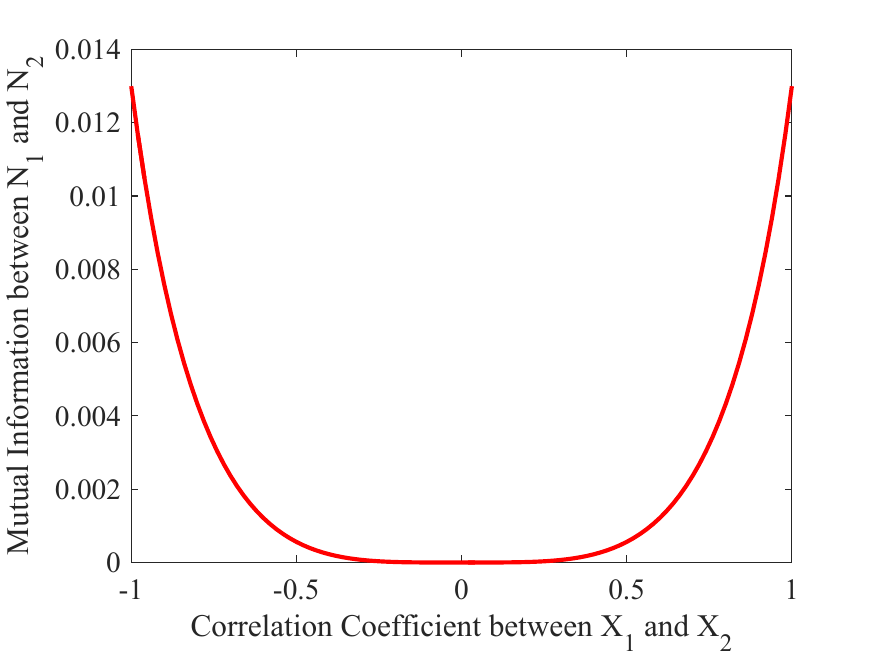} 
\label{Fig_corr2} } 
{ \includegraphics[width=0.23\textwidth]{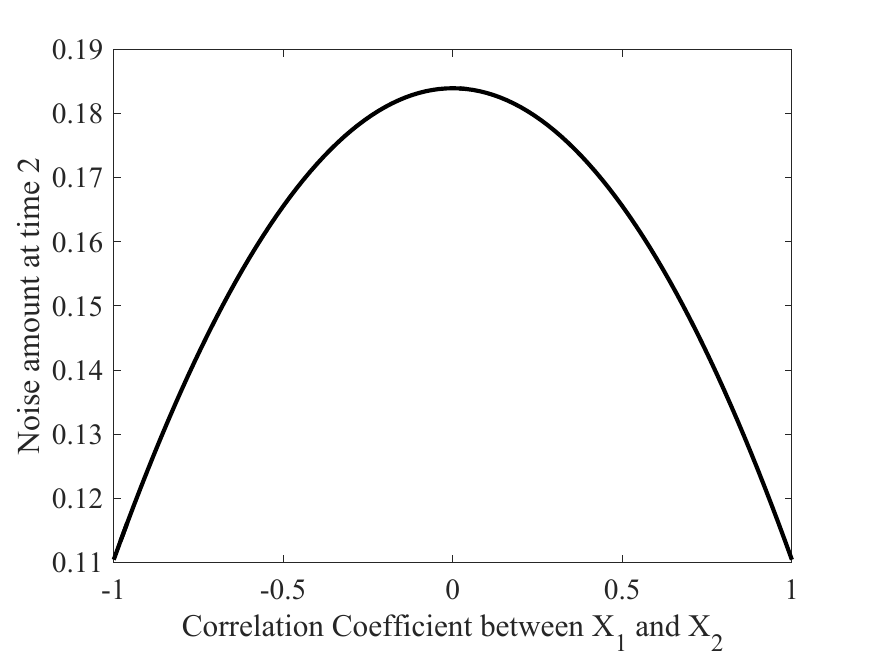} 
\label{Fig_corr2} } 
{ \includegraphics[width=0.23\textwidth]{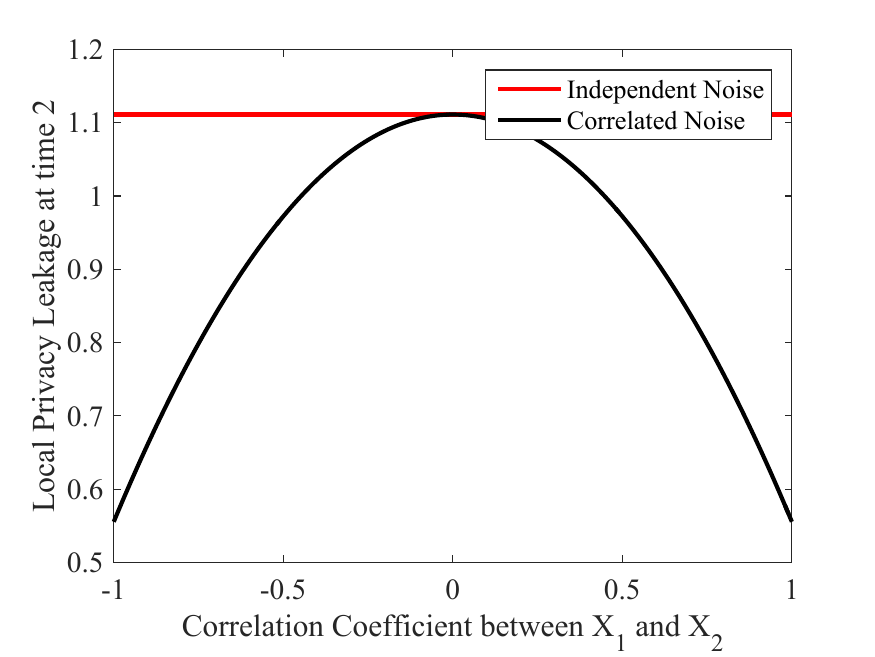} 
\label{Fig_corr2} }
\caption{Noise correlation, Mutual Information, Noise at time 2 and minimized leakage as a function of the correlation coefficient (when $P_1=0.5$)} 
\label{Noise_Corr} 
\end{small}
\end{figure*}
\begin{figure*}[htp]
\begin{small}
\centering 
{ \includegraphics[width=0.23\textwidth]{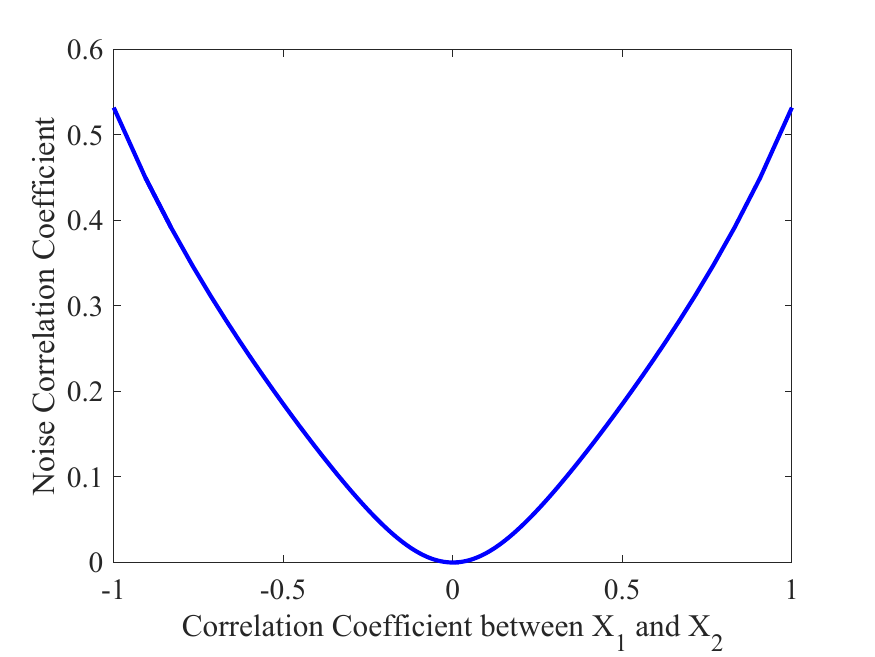} 
\label{Fig_corr} } 
{ \includegraphics[width=0.23\textwidth]{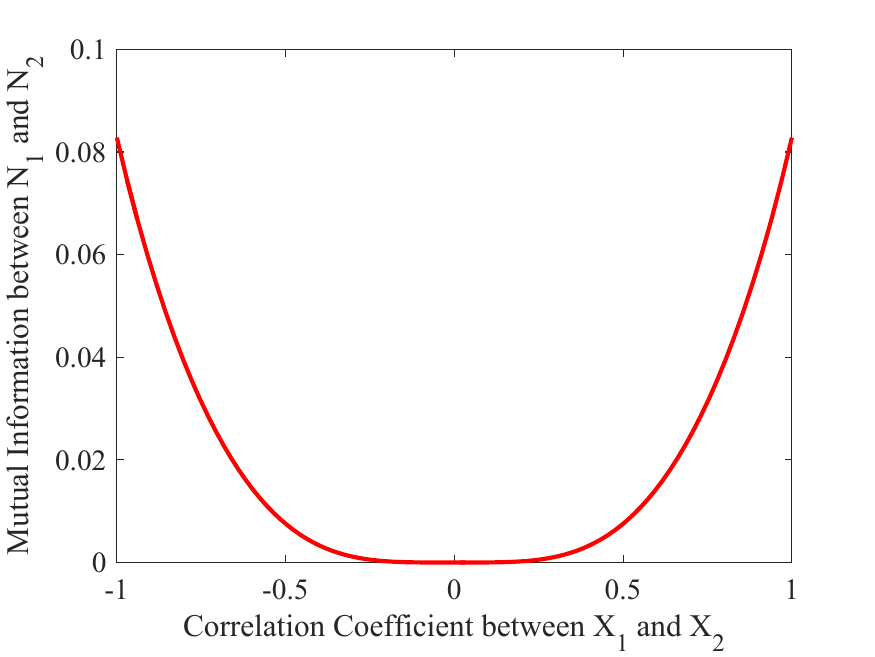} 
\label{Fig_corr2} } 
{ \includegraphics[width=0.23\textwidth]{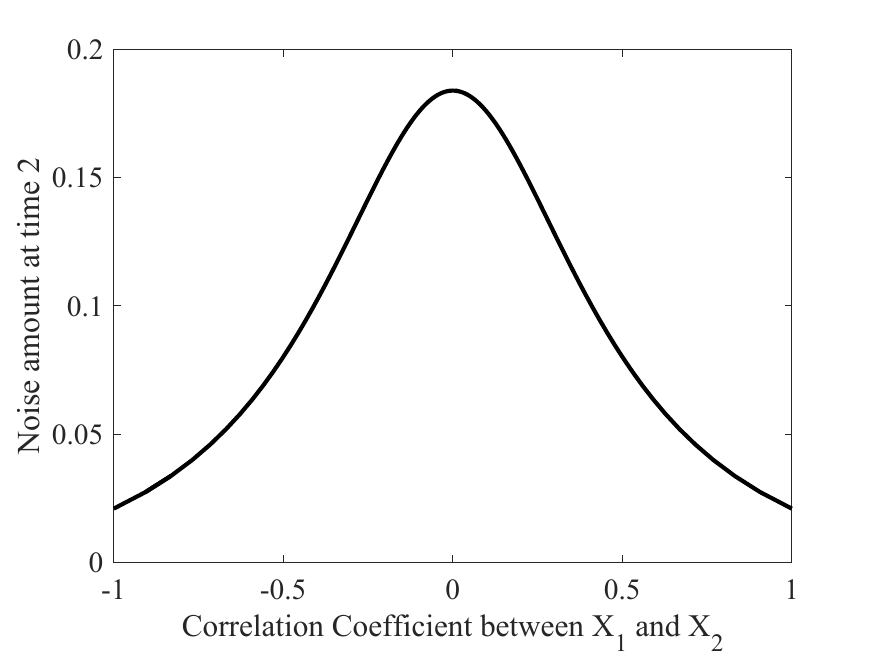} 
\label{Fig_corr3} }
{ \includegraphics[width=0.23\textwidth]{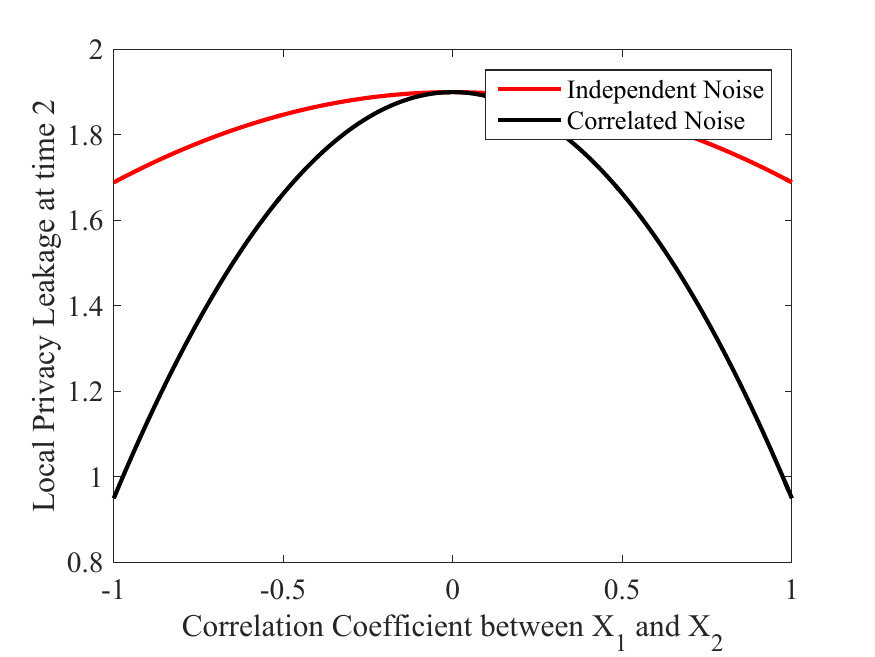} 
\label{Fig_corr4} }
\caption{Noise correlation, Mutual Information, Noise at time 2 and minimized leakage as a function of the correlation coefficient (when $P_1=0.95$)} 
\label{Noise_Corr} 
\end{small}
\end{figure*}

\subsection{Data Correlation Dependent Noise}

We present an illustrative example next to elucidate how data correlation influences the amount of noise added at different time steps. When data is correlated, the necessary noise at each timestamp also relies on one another.

\begin{exmpl}
Imagine releasing two correlated data points $X_1$ and $X_2$, and $X_1,X_2\in\{0,1\}$. Denote the prior distribution of $X_1$ as $P_1=\text{Pr}(X_1=1)$ and $1-P_1=\text{Pr}(X_1=0)$. We assume that the relationship between $X_1$ and $X_2$ is represented by a symmetric correlation channel, expressed as $\text{Pr}(X_2=1|X_1=0)=\text{Pr}(X_2=0|X_1=1)=\phi$. Prior to being published, $X_1$ and $X_2$ are perturbed to $Y_1$ and $Y_2$ respectively by the mechanisms $\mathcal{M}_1$ and $\mathcal{M}_2$. To ensure privacy, the local leakage of $\mathcal{L}_1(Y_1\to{X_1})\le{\epsilon_1}$ and $\mathcal{L}_2(Y_2\to{X_2,X_1})\le{\epsilon_2}$ must hold.
\end{exmpl}

% \subsubsection{Optimal Perturbation Parameters}
% At time step 1, the optimized LIP mechanism can be directly applied as the local leakage can be viewed as the definition of LIP. 
% By previous result, we have: $q(1)_{0}=\text{Pr}(Y_1=1|X_1=0)=\frac{P_1}{e^{\epsilon_1}}$, $q(0)_{1}=\text{Pr}(Y_1=0|X_1=1)=\frac{1-P_1}{e^{\epsilon_2}}$. At time step 2, the optimal parameters can be derived following theorem \ref{thm3} ($q$ is defined as $q(y_2)^{y_1}_{x_1,x_2}$):
% \begin{equation*}
%     \begin{aligned}
%         &q(1)^{1}_{00}=q(1)^{1}_{10}=\frac{\phi(1-P_1)+(1-\phi)(e^{\epsilon_1}-1+P_1)}{e^{\epsilon_1+\epsilon_2}}\\
%         &q(0)^{1}_{01}=q(0)^{1}_{11}=\frac{(1-\phi)(1-P_1)+\phi(e^{\epsilon_1}-1+P_1)}{e^{\epsilon_1+\epsilon_2}}\\
%         &q(1)^{0}_{00}=q(1)^{0}_{10}=\frac{\phi(e^{\epsilon_1  }-P_1)+(1-\phi)P_1}{e^{\epsilon_1+\epsilon_2}}\\
%         &q(0)^{0}_{01}=q(0)^{0}_{11}=\frac{(1-\phi)(e^{\epsilon_1}-P_1)+\phi P_1}{e^{\epsilon_1+\epsilon_2}}
%     \end{aligned}
% \end{equation*}
% \begin{rmk}
% When $X_1$ and $X_2$ are independent, implying $\phi=0.5$, all parameters equal $\frac{0.5}{\epsilon_2}$. Consequently, in situations where $X_1$ and $X_2$ are independent, the perturbation parameters for $X_2$ also do not rely on the values of $X_1$ and $Y_1$. When $\phi$ approaches either 0 or 1, the mechanism tends to release data according to their prior distribution.
% \end{rmk}

\subsubsection{Data Correlation and Noise correlation}
The correlation coefficient of of $X_1$ and $X_2$ can be expressed as:
\begin{equation*}
    \begin{aligned}
       \rho_{X_1,X_2}=&\frac{\sigma_{X_1X_2}}{\sigma_{X_1}\sigma_{X_2}}
       =\frac{E[X_1X_2]-E[X_1]E[X_2]}{\sqrt{\{E[X_1^2]-E^2[X_1]\}\{E[X_2^2]-E^2[X_2]\}}}.
    \end{aligned}
\end{equation*}
Taking values in, we have:
\begin{equation}
    \rho_{X_1,X_2}=\frac{(1-2\phi)\sqrt{P_1(1-P_1)}}{\sqrt{[P_1\phi+(1-P_1)(1-\phi)][(1-P_1)\phi+P_1(1-\phi)]}}
\end{equation}
Observe that when $\phi=0$, $\rho_{X_1,X_2}={1}$, when $\phi=1/2$, $\rho_{X_1,X_2}=0$; when $\phi=1$, $\rho_{X_1,X_2}=-{1}$. Thus by varying the value of $\phi$ from $0$ to $1$, the correlation coefficient also changes from positive to negative.

The noise $N$ at time 1 / time 2 in this example is defined as a binary random variable, such that $N=1$ means the data value is flipped before release; $N=0$ means the data is directly released. We refer to {$\text{Pr}(N=1)$} as the ``amount of noise". Our subsequent objective is to determine the correlation coefficient of $\rho_{N_1,N_2}$ (detailed derivation of $\rho_{N_1,N_2}$ as well as other parameters are shown in the Appendix).

\subsubsection{Relationship between $\rho_{X_1,X_2}$ and $\rho_{N_1,N_2}$}

We initially fix the value of $P_1$, then vary $\phi$ from $0$ to $1$ and observe changes in $\rho_{N_1,N_2}$. We know that $\rho_{x_1,x_2}$ also ranges from $-1$ to $1$.
Besides charting how noise correlation varies, we are also interested in identifying some specific cases. For instance, when $\rho_{X_1,X_2}=1$ or $\rho_{X_1,X_2}=-1$, we want to ascertain the value of $\rho_{N_1,N_2}$ and how it is influenced by the prior of $X_1$. Nevertheless, even if $\rho_{N_1,N_2}=0$, it doesn't conclusively prove independence. Thus, we consider calculating the mutual information between $N_1$ and $N_2$. If $I(N_1;N_2)=0$, we can conclude that $N_1$ and $N_2$ are independent.

Furthermore, we aim to demonstrate how the prior of $X_1$ and the correlation between $X_1$ and $X_2$ affect the amount of noise required at the second timestamp. The hope is that under the correlated release mechanism when data are more strongly correlated, the required amount of noise at the subsequent timestamp decreases.

Finally, we set an upper bound for the noise at time 2, and identify the optimal parameters that minimize the total privacy leakage of $X_1$ and $X_2$. For comparison, we also derive the parameters that minimize the leakage of $X_2$, which is equivalent to treating $X_1$ and $X_2$ as independent (with the prior of $X_2$ changing according to the correlations). The goal of this comparison is to illustrate that by considering data correlation, under a fixed amount of noise, we can further minimize the total leakage of the data stream. In other words, we are adding noise where it matters the most.

The results of our analysis are depicted in Fig. 3 and 4, where Fig. 3 shows the outcomes when $P_1=0.5$ and Fig. 4 presents the results when $P_1=0.95$.

From both Fig. 3 and Fig. 4, we observe a decrease in $\rho_{N_1,N_2}$ as $\rho_{X_1,X_2}$ approaches 0. This suggests that when the data are more correlated, the noise added is more dependent on the previous data release mechanisms. Conversely, when $\rho_{X_1X_2}=0$, $I(N_1;N_2)=0$, which implies that when data $X_1$ and $X_2$ are independent, the noise added at two different timestamps is also independent.
Another observation is that as $X_1$ and $X_2$ become more correlated, the required amount of noise at time 2 decreases correspondingly. This is because $Y_1$ already conveys a substantial amount of information about $X_2$, and the prior of $X_2$ is relatively more skewed than in scenarios where $X_1$ and $X_2$ are less correlated. 
When the amount of noise at time 2 is bounded, introducing correlated noise can further minimize the total privacy leakage. When $P_1=0.5$, the independent noise remains unchanged because the prior of $X_2$ does not fluctuate with $\phi$.
Upon comparing Fig. 3 and 4, we find that the intercept of Fig. 4 is greater than that of Fig. 3. This implies that when the prior is more uniformly distributed, the correlation of the noise decreases, while when the prior distribution of $P_1$ is more skewed, the correlation of the noise increases.

\section{Utility-Privacy tradeoff for batch setting}
Note that we cannot directly extend the optimal solution in Theorem \ref{thm3} to the release vector in the batched release setting, as the optimal parameters only maximize the probability of releasing one output that is identical to the input, which are not optimal for releasing vectors containing multiple data. Because they fail to enlarge the probabilities to release the sequence where most data are identical to the input but only a few are different. In this section, we first investigate the utility and privacy tradeoff in the batched release setting, followed by model simplification. Finally, we present our algorithm in the batched setting.

\subsection{Utility in the Batched Release Setting}

In the batched release setting, the utility measurement is termed the Batched Expected Distance (BED):

\begin{equation}\label{eq:BED}
E\left[\mathsf{D}(Q(\bold{O}_{l}),Q(\bold{R}_{l})\big|\bold{R}_1^{l-1}=\bold{r}_1^{l-1}\right].
\end{equation}

Similar to the instantaneous release setting, for batched release setting, the optimization problem that yields the utility-privacy  tradeoff is defined as:
\begin{equation}\label{tradeoff2}
    \begin{aligned}
        & \min E\left[\mathsf{D}(Q({\bold{O}}_{l}),Q(\bold{R}_{l})\big|\bold{R}_1^{l-1}=\bold{r}_1^{l-1}\right],\\
        & \text{Such that} ~\mathcal{L}(\bold{R}_l\to{\bold{O}_1^l})\le{\epsilon_l}, ~\forall{l\in[1,\mathsf \tau]}.
    \end{aligned}
\end{equation}
The privacy constraint of the batch setting can be expressed as:
\begin{equation}
    \frac{ a^B_{l}(r_{l}| \bold{o}_{1}^{l}, \bold{r}_{1}^{l-1})}{\sum_{\bold{\tilde{o}}_1^l}a^B_{l}(r_{l}| \bold{\tilde{o}}_1^l, \bold{r}_{1}^{l-1})\beta^B_l(\bold{\tilde{o}}_1^l|\bold{r}_1^{l-1})}.
\end{equation}
The belief state $\beta^B_l(\bold{{o}}_1^l|\bold{r}_1^{l-1})$ can be expressed as:
\begin{equation}
    C^{\tilde{o}_l}_{\tilde{\bold{o}}_1^{l-1}}\frac{ a^B_{l-1}(r_{l-1}|\bold{\tilde{o}}_{1}^{l-1}, \bold{r}_{1}^{l-1})\beta^B_{l-1}(\bold{\tilde{o}}_1^{l-1}|\bold{r}_1^{l-2})}{\sum_{\bold{\bar{o}}_1^{l-1}}a^B_{l-1}(r_{l-1}|\bold{{\bar{o}}}_{1}^{l-1}, \bold{r}_{1}^{l-2})\beta^B_{l-1}(\bold{{\bar{o}}}_1^{l-1}|\bold{r}_1^{l-2})},
\end{equation}
where the correlation term can be derived as:
\begin{equation}
\begin{aligned}
    C^{\bold{o}_l}_{\bold{o}_1^{l-1}}=&\text{Pr}(\bold{o}_l|\bold{o}_1^{l-1})\\
    % =&\text{Pr}\left(\bold{X}_{(l-1)w}^{lw}=\bold{x}_{(l-1)w}^{lw}\big| \bold{X}_{1}^{(l-1)w-1}=\bold{x}_{1}^{(l-1)w-1}\right)\\
    =&\prod_{i = (l-1)w+1}^{lw}\text{Pr}\left({X}_i=x_i\big| \bold{X}_{1}^{i-1}=\bold{x}_{1}^{i-1}\right).
\end{aligned}
\end{equation}
The utility function of BED can be further expressed as:
\begin{small}
\begin{equation}\label{utility_BED}
\begin{aligned}
&E\left[D(Q(\bold{O}_l)-Q(\bold{R}_l))|\big|\bold{R}_1^{l-1}=\bold{r}_1^{l-1}\right]\\
=&\sum_{\bold{o},\bold{r}}D(Q(\bold{o})-Q(\bold{r}))\sum_{\bold{o}_1^{l-1}}a^B_l(r_l|\bold{o}_1^{l-1},\bold{o},\bold{r}_1^{l-1})\beta^B_l(\bold{o}_1^{l-1},\bold{o}|\bold{r}_1^{l-1}).\\
\end{aligned}
\end{equation}
\end{small}
% The mechanism parameters can be further simplified by consider a limited historical batched data. The corresponding analysis is similar to that of the instantaneous release setting. We omit the intermediate steps and show the final form of the utility-privacy tradeoff as follows:

\textbf{Mechanism simplification}
Note that the mechanism parameter $a_l^B$ contains the whole raw data sequence, which makes the computational cost grow exponentially. Next, we introduce a subset of policies that requires a memory of length $L$:
\begin{equation}
\begin{aligned}\label{action3}
    &a_{l}^s(\bold{r}_{l}| \bold{o}^l_{l-L+1}, \bold{r}_{1}^{l-1}) =\\ &\operatorname{Pr}\left(\bold{R}_{l}=\bold{r}_{l} \mid \bold{O}_{l-L+1}^l=\bold{o}_{l-L+1}^l, \mathbf{R}_{1}^{l-1}= r_{1}^{l-1}\right).
\end{aligned}
\end{equation}
% Note that, compared with \eqref{action}, \eqref{action2} differs in the condition terms, where \eqref{action} depends on the whole histories of raw data sequence, \eqref{action2} depends only on the data of the past two time instances.

The next Theorem states that considering a subset of the policies \eqref{action3} {will not violate the privacy constraints}

\begin{thm}
 {For a batched release mechanism $\mathcal{M}^s_l$ that is parameterized by $a^s_l$, and satisfies the condition in \eqref{eq:condition},}
\begin{equation}\label{eq:condition}
    \frac{\operatorname{Pr}(\bold{O}_{l-L+1}^l=\bold{o}_{l-L+1}^l|\bold{R}_1^l=\bold{r}_{1}^{l})}{\operatorname{Pr}(\bold{O}_{l-L+1}^l=\bold{o}_{l-L+1}^l|\bold{R}_1^{l-1}=\bold{r}_{1}^{l-1})}\in[e^{-\epsilon},e^{\epsilon}],
\end{equation}
{it is sufficient to show $\mathcal{M}^s_l$ makes the BIL defined in \eqref{eq:BIL} upper bounded by $\epsilon$.}
\end{thm}

\begin{proof}
For the privacy constraints (one term):

% \begin{equation}
% \begin{aligned}
%     &\frac{\text{Pr}(\bold{X}^k_1=\bold{x}^k_1|\bold{Y}_1^k=\bold{y}_{1}^{k})}{\text{Pr}(\bold{X}_1^k=\bold{x}_1^k|\bold{Y}_1^{k-1}=\bold{y}_{1}^{k-1})}\\ = &\frac{\text{Pr}(\bold{X}^{k-L}_1=\bold{x}^{k-L}_1|\bold{Y}_1^k=\bold{y}_{1}^{k},\bold{X}_{k-L+1}^k=\bold{x}_{k-L+1}^k)}{\text{Pr}(\bold{X}_1^{k-L}=\bold{x}_1^{k-L}|\bold{Y}_1^{k-1}=\bold{y}_{1}^{k-1},\bold{X}_{k-L+1}^k=\bold{x}_{k-L+1}^k)}\frac{\text{Pr}(\bold{X}_{k-L+1}^k=\bold{x}_{k-L+1}^k|\bold{Y}_1^k=\bold{y}_{1}^{k})}{\text{Pr}(\bold{X}_{k-L+1}^k=\bold{x}_{k-L+1}^k|\bold{Y}_1^{k-1}=\bold{y}_{1}^{k-1})}.
% \end{aligned}
% \end{equation}

\begin{equation}
\begin{aligned}
    &\frac{\text{Pr}(\bold{O}^l_1=\bold{o}^l_1|\bold{R}_1^l=\bold{r}_{1}^{l})}{\text{Pr}(\bold{O}_1^l=\bold{o}_1^l|\bold{R}_1^{l-1}=\bold{r}_{1}^{l-1})}\\ 
    = &\mathsf{A}\cdot\frac{\text{Pr}(\bold{O}_{l-L+1}^l=\bold{o}_{l-L+1}^l|\bold{R}_1^l=\bold{r}_{1}^{l})}{\text{Pr}(\bold{O}_{l-L+1}^l=\bold{o}_{l-L+1}^l|\bold{R}_1^{l-1}=\bold{r}_{1}^{l-1})},
\end{aligned}
\end{equation}
where

% \begin{equation}
% \begin{aligned}
%     &\frac{\text{Pr}(\bold{X}^{k-L}_1=\bold{x}^{k-L}_1|\bold{Y}_1^k=\bold{y}_{1}^{k},\bold{X}_{k-L+1}^k=\bold{x}_{k-L+1}^k)}{\text{Pr}(\bold{X}_1^{k-L}=\bold{x}_1^{k-L}|\bold{Y}_1^{k-1}=\bold{y}_{1}^{k-1},\bold{X}_{k-L+1}^k=\bold{x}_{k-L+1}^k)}\\
%     =&\frac{\text{Pr}(\bold{X}^{k-L}_1=\bold{x}^{k-L}_1|\bold{Y}_1^{k-1}=\bold{y}_{1}^{k-1},\bold{X}_{k-L+1}^k=\bold{x}_{k-L+1}^k)\text{Pr}(Y_k=y_k|\bold{Y}_1^{k-1}=\bold{y}_{1}^{k-1},\bold{X}_{1}^k=\bold{x}_1^{k})}{\text{Pr}(Y_k=y_k|\bold{Y}_1^{k-1}=\bold{y}_{1}^{k-1},\bold{X}_{k-L+1}^k=\bold{x}_{k-L+1}^k)\text{Pr}(\bold{X}_1^{k-L}=\bold{x}_1^{k-L}|\bold{Y}_1^{k-1}=\bold{y}_{1}^{k-1},\bold{X}_{k-L+1}^k=\bold{x}_{k-L+1}^k)}\\
%     =&\frac{\text{Pr}(Y_k=y_k|\bold{Y}_1^{k-1}=\bold{y}_{1}^{k-1},\bold{X}_{1}^k=\bold{x}_1^{k})}{\text{Pr}(Y_k=y_k|\bold{Y}_1^{k-1}=\bold{y}_{1}^{k-1},\bold{X}_{k-L+1}^k=\bold{x}_{k-L+1}^k)}=1\\
% \end{aligned}
% \end{equation}

\begin{equation}
\begin{aligned}
    &\mathsf{A}\\=&\frac{\text{Pr}(\bold{O}^{l-L}_1=\bold{o}^{l-L}_1|\bold{R}_1^l=\bold{r}_{1}^{l},\bold{O}_{l-L+1}^l=\bold{o}_{l-L+1}^l)}{\text{Pr}(\bold{O}_1^{l-L}=\bold{o}_1^{l-L}|\bold{R}_1^{l-1}=\bold{r}_{1}^{l-1},\bold{O}_{l-L+1}^l=\bold{o}_{l-L+1}^l)}\\
    =&\frac{\text{Pr}(\bold{O}^{l-L}_1=\bold{o}^{l-L}_1|\bold{R}_1^{l-1}=\bold{r}_{1}^{l-1},\bold{O}_{l-L+1}^l=\bold{o}_{l-L+1}^l)}{\text{Pr}(\bold{R}_l=\bold{r}_l|\bold{R}_1^{l-1}=\bold{r}_{1}^{l-1},\bold{O}_{l-L+1}^l=\bold{o}_{l-L+1}^l)}\\
    &~~~~~~~~\cdot\frac{\text{Pr}(\bold{R}_l=\bold{r}_l|\bold{R}_1^{l-1}=\bold{r}_{1}^{l-1},\bold{O}_{1}^l=\bold{o}_1^{l})}{\text{Pr}(\bold{O}_1^{l-L}=\bold{o}_1^{l-L}|\bold{R}_1^{l-1}=\bold{r}_{1}^{l-1},\bold{O}_{l-L+1}^l=\bold{o}_{l-L+1}^l)}\\
    =&\frac{\text{Pr}(\bold{R}_l=\bold{r}_l|\bold{R}_1^{l-1}=\bold{r}_{1}^{l-1},\bold{O}_{1}^l=\bold{o}_1^{l})}{\text{Pr}(\bold{R}_l=\bold{r}_l|\bold{R}_1^{l-1}=\bold{r}_{1}^{l-1},\bold{O}_{l-L+1}^l=\bold{o}_{l-L+1}^l)}=1.\\
\end{aligned}
\end{equation}
The last equation holds because the new policy does not depend on the sequence of $\bold{O}_1^{l-L}$. Then, by considering the subset of policies, the BIL becomes:
\begin{equation}
\begin{aligned}
    &\mathcal{L}(\bold{R}_l\to{\bold{O}_{l-L+1}^l})=\\
    &\max_{\bold{o}_{l-L+1}^l, \bold{r}_l}\left|\log\frac{\text{Pr}(\bold{O}_{l-L+1}^l=\bold{o}_{l-L+1}^l|\bold{R}_1^l=\bold{r}_{1}^{l})}{\text{Pr}(\bold{O}_{l-L+1}^l=\bold{o}_{l-L+1}^l|\bold{R}_1^{l-1}=\bold{r}_{1}^{l-1})}\right|.
\end{aligned}
\end{equation}
\end{proof}

The utility function of BED, by adopting a policy in \eqref{action2}, can be expressed as:
\begin{small}
\begin{equation}\label{utility_action2}
\begin{aligned}
&E\left[D(Q(\bold{O}_l)-Q(\bold{R}_l))\big|\bold{R}_1^{l-1}=\bold{r}_1^{l-1}\right]\\
=&\sum_{\bold{o},\bold{r}}D(Q(\bold{o})-Q(\bold{r}))\\
&~~\cdot\sum_{\bold{o}^{l-1}_{l-L+1}}\text{Pr}(\bold{O}_l=o,\bold{O}^{l-1}_{l-L+1}=\bold{o}^{l-1}_{l-L+1},R_l=r_l|\bold{R}_1^{l-1}=\bold{r}_1^{l-1})\\
=&\sum_{\bold{o},\bold{r}}D(Q(\bold{o})-Q(\bold{r}))\\
&~~\cdot\sum_{\bold{o}^{l-1}_{l-L+1}}\big\{\text{Pr}(R_l=r_l|\bold{R}_1^{l-1}=\bold{r}_1^{l-1},\bold{O}_l=\bold{o},\bold{O}^{l-1}_{l-L+1}=\bold{o}^{l-1}_{l-L+1})\\
&~~\cdot\text{Pr}(\bold{O}_l=o,\bold{O}^{l-1}_{l-L+1}=\bold{o}^{l-1}_{l-L+1}|\bold{R}_1^{l-1}=\bold{r}_1^{l-1})\big\}\\
=&\sum_{\bold{o},\bold{r}}D(Q(\bold{o})-Q(\bold{r}))\\
&~~\cdot\sum_{\bold{o}^{l-1}_{l-L+1}}a^s_l(r_l|\bold{o}^{l-1}_{l-L+1},o,\bold{r}_1^{l-1})\beta^s_l(\bold{o}^{l-1}_{l-L+1},o|\bold{r}_1^{l-1}).\\
\end{aligned}
\end{equation}
\end{small}

On the other hand, the simplified belief state using \eqref{action} can be further expressed as:
\begin{small}
\begin{equation}\label{belif2}
    \begin{aligned}
        &\text{Pr}(\bold{O}^{l}_{l-L+1}=\bold{o}^{l}_{l-L+1}|\bold{R}_1^{l-1}=\bold{r}_{1}^{l-1})\\
        =&C_{{l-L+1}}^{l}\text{Pr}(\bold{O}^{l-1}_{l-L+1}=\bold{o}^{l-1}_{l-L+1}|\bold{R}_1^{l-1}=\bold{r}_{1}^{l-1})\\
        % =&C_{{l-L+1}}^{l}\frac{\sum_{o_{l-L}}\text{Pr}(R_{l-1}=r_{l-1}|\bold{O}^{l-1}_{l-L}=\bold{o}^{l-1}_{l-L},\bold{R}_1^{l-2}=\bold{r}_1^{l-2})\text{Pr}(\bold{O}^{l-1}_{l-L}=\bold{o}^{l-1}_{l-L}|\bold{R}_1^{l-2}=\bold{r}_1^{l-2})}{\text{Pr}(R_{l-1}=r_{l-1}|\bold{R}_1^{l-2}=\bold{r}_1^{l-2})}\\
        % =&C_{{l-L+1}}^{l}\frac{\sum_{o_{l-L}}\text{Pr}(R_{l-1}=r_{l-1}|\bold{O}^{l-1}_{l-L}=\bold{o}^{l-1}_{l-L},\bold{R}_1^{l-2}=\bold{r}_1^{l-2})\text{Pr}(\bold{O}^{l-1}_{l-L}=\bold{o}^{l-1}_{l-L}|\bold{R}_1^{l-2}=\bold{r}_1^{l-2})}{\sum_{\bar{\bold{o}}^{l-1}_{l-L}}\text{Pr}(R_{l-1}=r_{l-1}|{\bold{O}}^{l-1}_{l-L}=\bar{\bold{o}}^{l-1}_{l-L},\bold{R}_1^{l-2}=\bold{r}_1^{l-2})\text{Pr}({\bold{O}}^{l-1}_{l-L}=\bar{\bold{o}}^{l-1}_{l-L}|\bold{R}_1^{l-2}=\bold{r}_1^{l-2})}\\
        =&C_{{l-L+1}}^{l}\frac{\sum_{\bold{o}_{l-L}}a_{l-1}(\bold{r}_{l-1}|\bold{o}^{l-1}_{l-L},\bold{r}_1^{l-2})\beta^s_{l-1}(\bold{o}^{l-1}_{l-L}|\bold{R}_1^{l-2}=\bold{r}_1^{l-2})}{\sum_{\bar{\bold{o}}^{l-1}_{l-L}}a_{l-1}(\bold{r}_{l-1}|\bar{\bold{o}}^{l-1}_{l-L},\bold{r}_1^{l-2})\beta^s_{l-1}(\bar{\bold{o}}^{l-1}_{l-L}|\bold{R}_1^{l-2}=\bold{r}_1^{l-2})}\\
        =&C_{{l-L+1}}^{l}\frac{\sum_{\bold{o}_{l-L}}a^s_{l-1}(\bold{r}_{l-1}|\bold{o}^{l-1}_{l-L},\bold{r}_1^{l-2})\beta^s_{l-1}(\bold{o}^{l-1}_{l-L}|\bold{R}_1^{l-2}=\bold{r}_1^{l-2})}{\sum_{\bar{\bold{o}}^{l-1}_{l-L}}a^s_{l-1}(\bold{r}_{l-1}|\bar{\bold{o}}^{l-1}_{l-L},\bold{r}_1^{l-2})\beta^s_{l-1}(\bar{\bold{o}}^{l-1}_{l-L}|\bold{R}_1^{l-2}=\bold{r}_1^{l-2})}.\\
    \end{aligned}
\end{equation}
\end{small}
Where we use $C$ in \eqref{belif2} to denote the correlation among batched sequences. The simplification allows us to reduce computation complexity without violating the privacy guarantee. In the following, we let $L=1$. Note that by restricting the policy to only using a subset of historical input batches, the solution to \eqref{eq:batch} will be a sub-optimal solution of the original utility-privacy tradeoff formulation in \eqref{tradeoff2}.  Hence, the simplified utility-privacy tradeoff of the batched model becomes:
\begin{equation}\label{eq:batch}
\begin{aligned}
&\sum_{o,r}D(Q(\bold{o})-Q(\bold{r}))a^s_l(\bold{r}_l|\bold{o},\bold{r}_1^{l-1})\beta^s_l(o|\bold{r}_1^{l-1}),\\
& \text{Such that}~~ \frac{ a^s_{l}(\bold{r}| \bold{o}, \bold{r}_{1}^{l-1})}{\sum_{{\tilde{\bold{o}}}}a^s_{l}(\bold{r}| {\tilde{\bold{o}}}, \bold{r}_{1}^{l-1})\beta^s_l({\tilde{\bold{o}}}|\bold{r}_1^{l-1})}\in[e^{-\epsilon},e^{\epsilon}].
\end{aligned}
\end{equation}

 Notably, the objective function in \eqref{eq:batch} is a linear combination of $a^s_{l}(\bold r| \bold{o}, \bold{r}_{1}^{l-1})$ for all $\bold r, \bold o$, enabling the attainment of the global optimal solution through convergence-based algorithms such as the gradient descent. Herein, we utilize the gradient descent algorithm to numerically solve the optimization problem in \eqref{eq:batch}, as outlined in Alg. \ref{alg:BRM}. In essence, we initialize the perturbation parameters from a uniform distribution. In each iteration, we calculate a partial derivative of the utility function with respect to each parameter. Upon updating each $a_l$ with a step length $\phi$, we verify if the current parameters satisfy the privacy constraints. If they do not, we halt the update of the current parameters. Observe that the mechanism needs to update all parameters $\bold{o},\bold{r}\in\mathcal{X}^w$. The computation complexity of Alg. 2. is $\mathcal{O}(|\mathcal{X}|^{2w}).$

\begin{algorithm}[tbh]
	\small
	\caption{{SIP mechanism for Batched release}}
	\label{alg:BRM}
	\begin{algorithmic}[1]
	\item Input: current time $l$, initial prior $\text{P}_{\bold{O}_1}(\bold{o})$, transitional matrix $C^{l+1}_{l}$ for all $0<l<B$, historical release sequence $\bold{R}_1^{l-1}$, current $\bold{o}_l$, Utility function $U$, step length $\phi$.
	\item Output: Batched release $R_l$
	\item if $l \neq 1$: 
        \item ~~~~Update belief state for all $o$ according to \eqref{belif2}
        \item Initialize $a_l(\bold{r}_l|\bold{o}_l,\bold{r}_1^{l-1})=1/|\mathcal{X}|^{w}$ for all $\bold{o}_l$ and $\bold{r}_l$, $U^*=\infty$, $Act = \text{ones}({|\mathcal{X}|^{w}}, |\mathcal{X}|^{w}])$.
        \item While sum$(Act)\neq 0:$
        \item ~~~for $\bold{o}\in |\mathcal{X}|^{w}$:
        \item ~~~~~~for $\bold{r}\in |\mathcal{X}|^{w}$: 
        \item ~~~~~~~~~Calculate the derivative: $d(\bold{o},\bold{r})={\partial U}/{\partial a_l(\bold{r}|\bold{o},\bold{r}_1^{l-1})}$
        \item ~~~~~~~~~$a_l(\bold{r}|\bold{o},\bold{r}_1^{l-1})\gets a_l(\bold{r}|\bold{o},\bold{r}_1^{l-1})+\phi \cdot d(\bold{o},\bold{j})$
	\item ~~~~~~~~~Check privacy constraints, $s(\bold{o},\bold{r})=0$ 
        \item ~~~~~~~~~for $\tilde{\bold{o}}\in|\mathcal{X}|^{w}$:
        \item ~~~~~~~~~~~~$s(\bold{o},\bold{r}) \gets s(\bold{o},\bold{r}) + a_l(\bold{r}|\tilde{\bold{o}},\bold{\bold{r}}_1^{l-1})\beta_l(\tilde{\bold{o}}|\bold{\bold{r}}_1^{l-1})$
        \item ~~~~~~~~~if $s(\bold{o},\bold{r})/ [a_l(\bold{r}|\tilde{\bold{o}},\bold{r}_1^{l-1})\beta_l(\tilde{\bold{o}}|\bold{r}_1^{l-1})] \notin [e^{-\epsilon},e^{\epsilon}]$:
        % \item ~~~~~~~~~~~~$a^*_l=a_l(r|o,\bold{r}_1^{l-1})$
        % \item ~~~~~~~~~else:
        \item ~~~~~~~~~~~~$Act[\bold{o}][\bold{r}]=0$
        \item Release $\bold{R}_l$ according to $a_l(\bold{r}_l|\bold{o}_l,\bold{r}_1^{l-1})$
	\end{algorithmic}
\end{algorithm}

\begin{figure}[htp]
\begin{small}
\centering 
{ \includegraphics[width=0.4\textwidth]{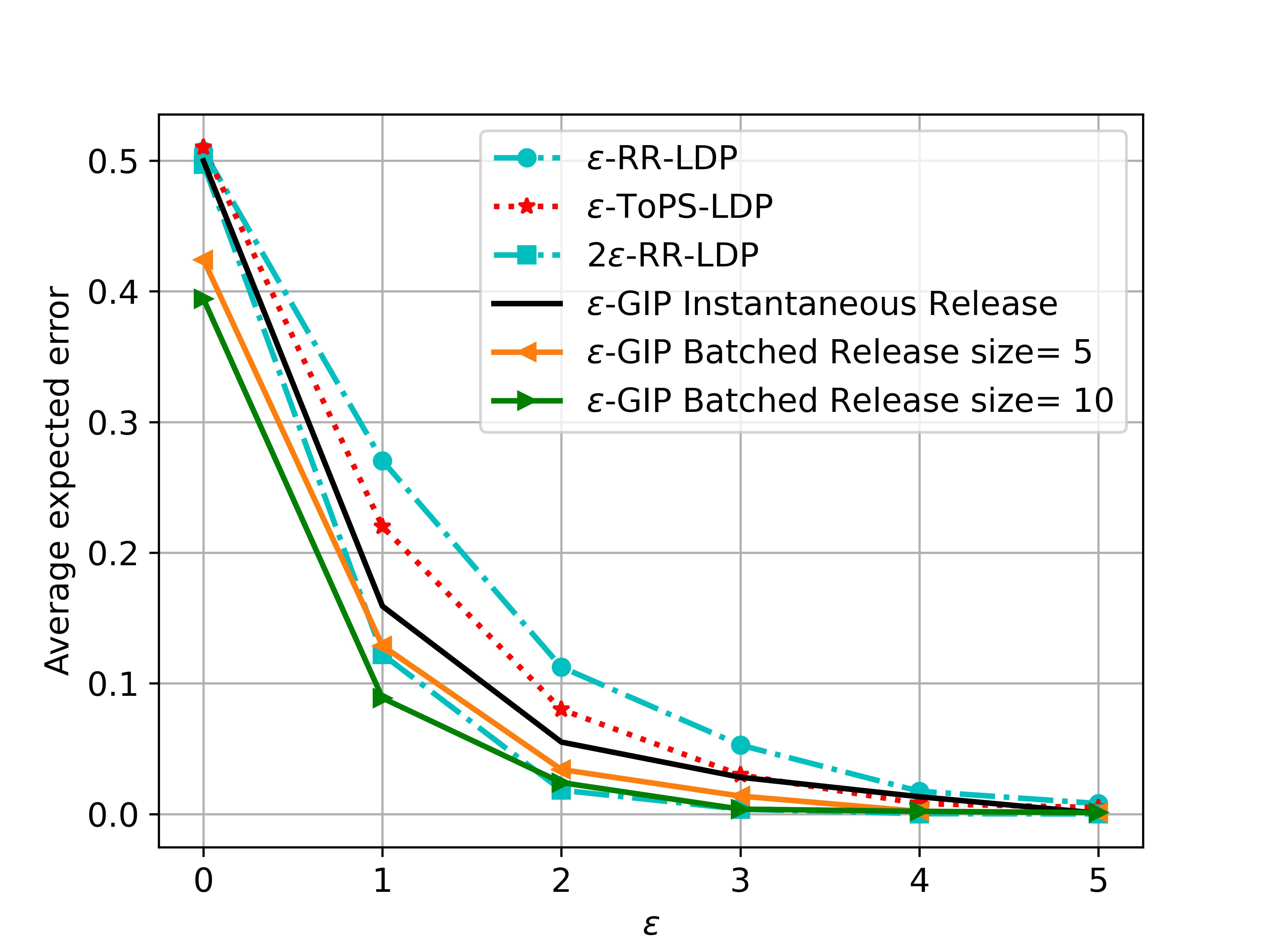}
\label{up_val} } 
{ \includegraphics[width=0.4\textwidth]{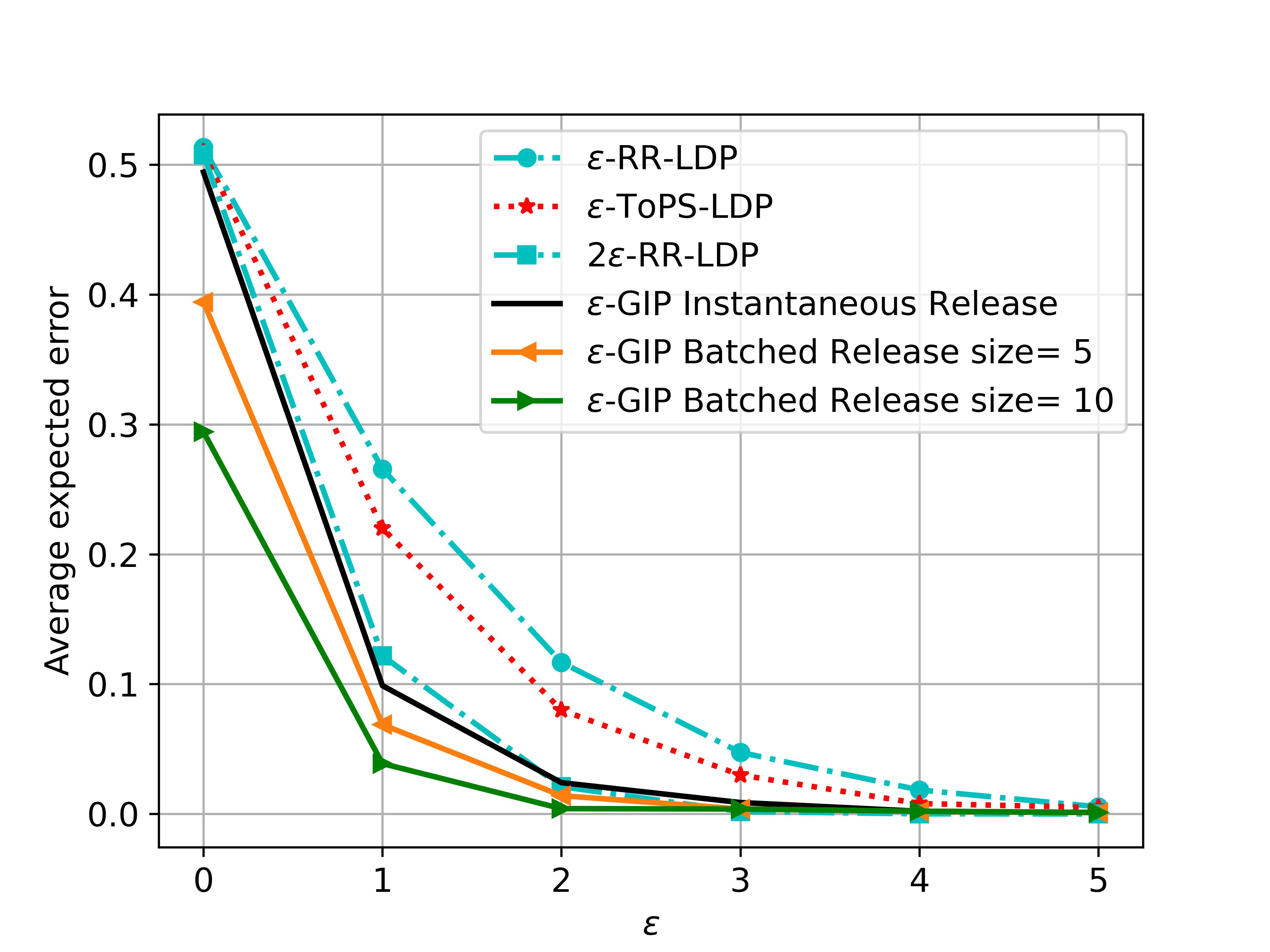} 
\label{up_val2} } 
\caption{Privacy-utility tradeoff comparison for different data release settings using synthetic data with weak (Case 1) and strong (Case 2) data correlations.} 
\label{UP_compare} 
\end{small}
\end{figure}
\begin{figure}[htp]
\begin{small}
\centering 
{ \includegraphics[width=0.4\textwidth]{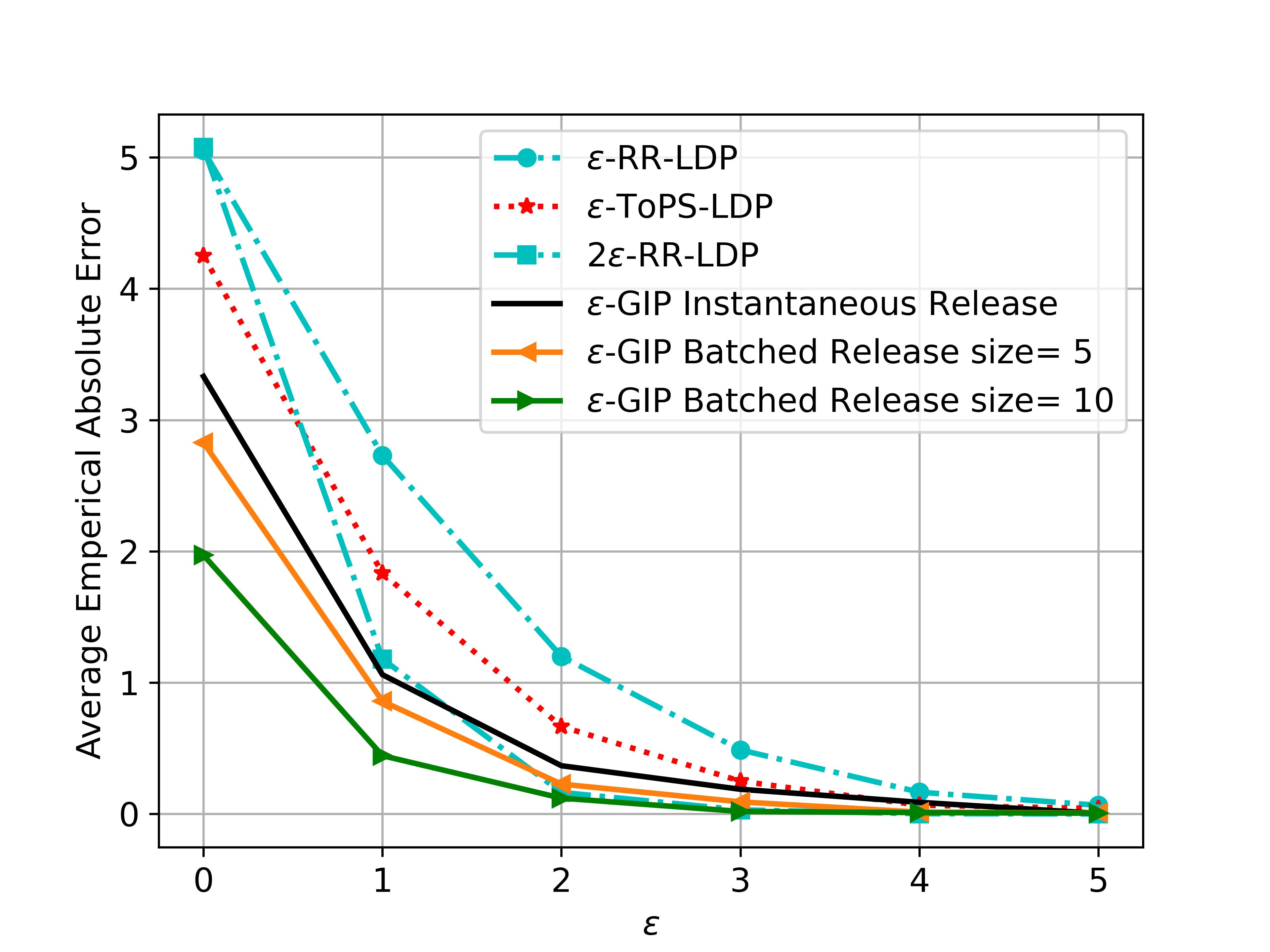}
\label{up_cat} } 
\caption{Utility privacy tradeoff comparison between different data release settings with Kosarak} 
\label{UP_compare2} 
\end{small}
\end{figure}

\begin{figure*}[t]
\label{fig:leakage_com}
\centering 
\subfigure[Pre-processed data] 
{ \includegraphics[width=0.23\textwidth]{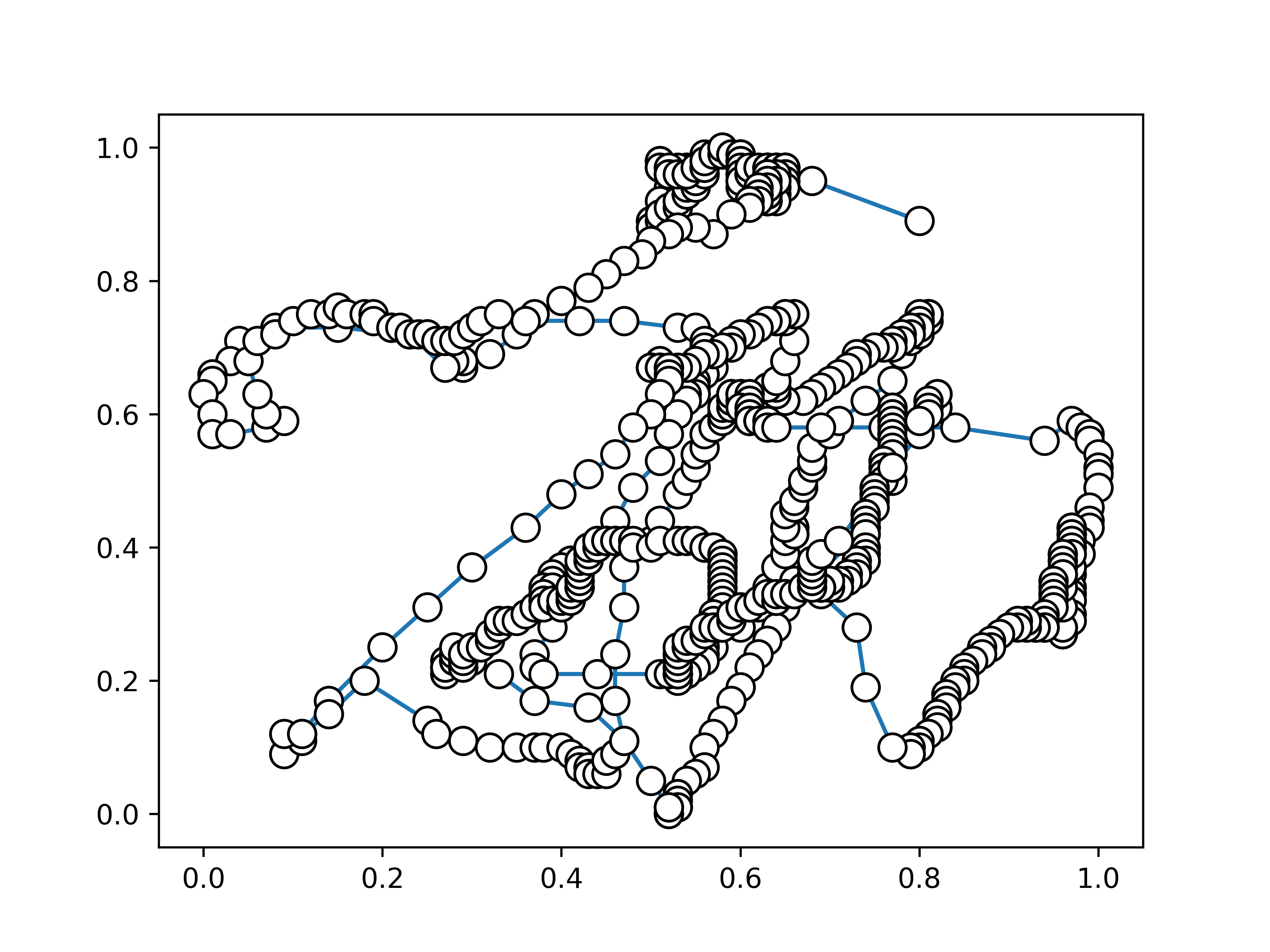}
\label{real1} }
\subfigure[Independent RR-LDP]
{ \includegraphics[width=0.23\textwidth]{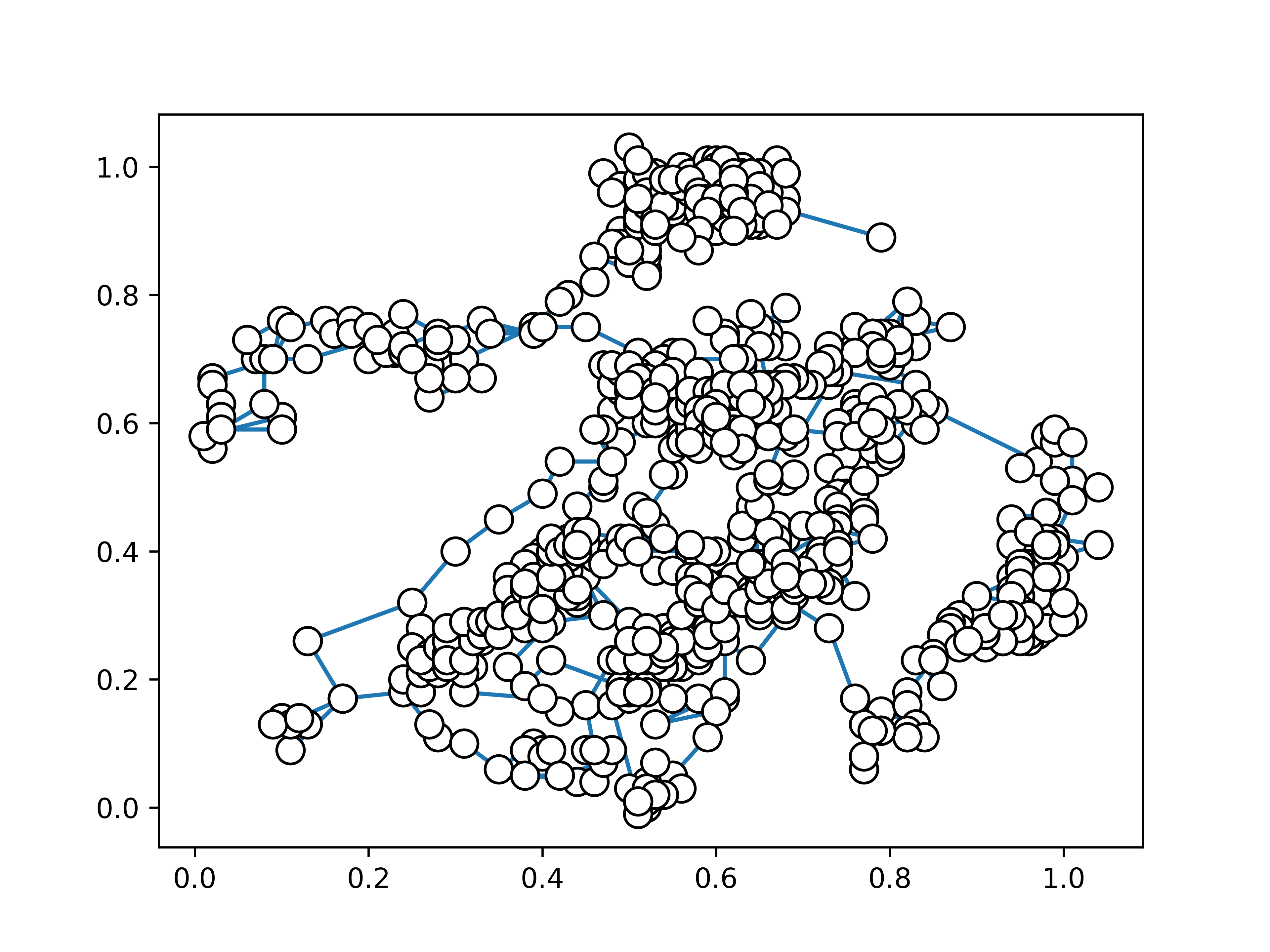}
\label{real2}}
\subfigure[Correlated Instantaneous release]
{\includegraphics[width=0.23\textwidth]{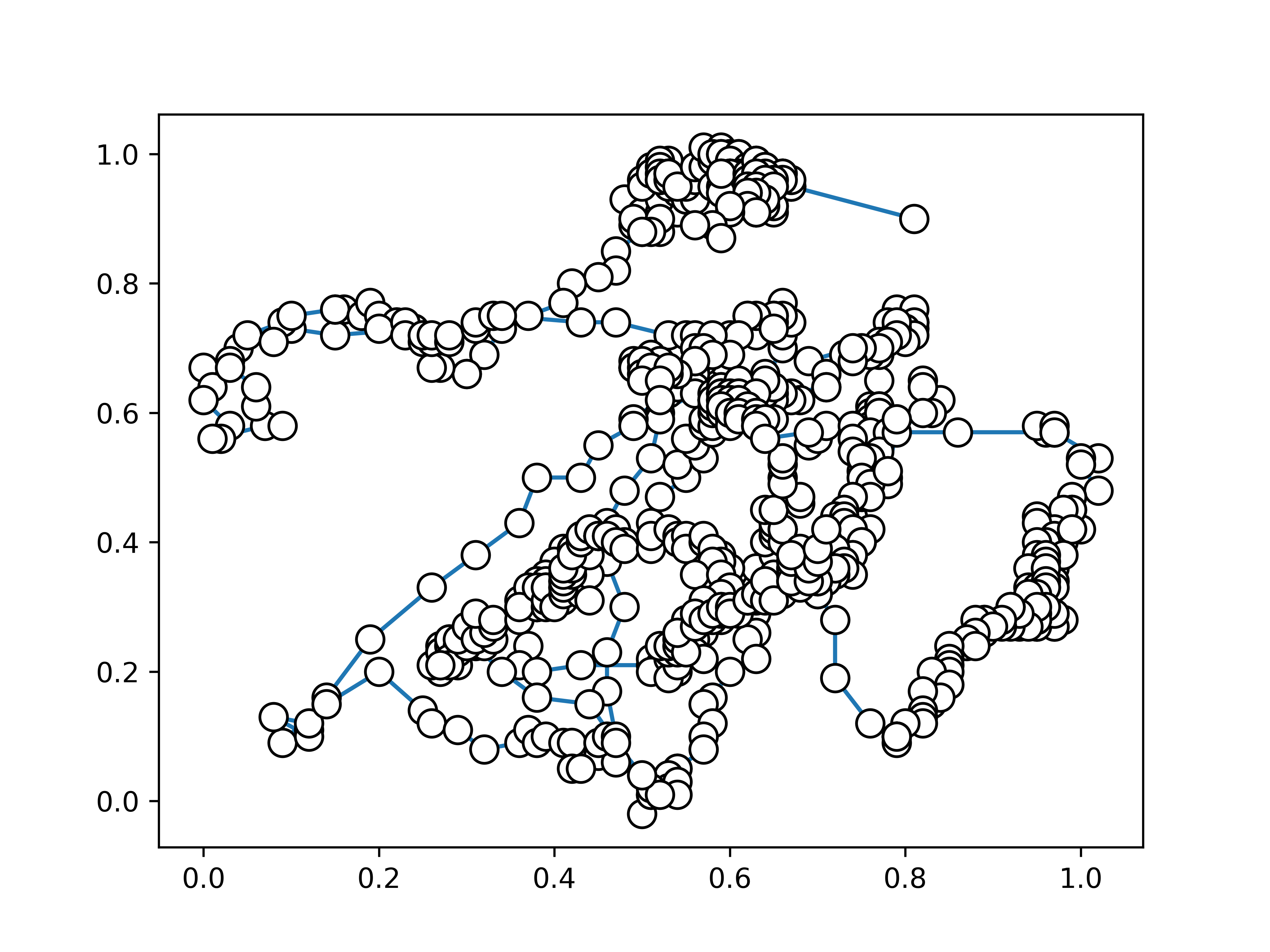}
\label{real3}}
\subfigure[Correlated Batched release]
{\includegraphics[width=0.23\textwidth]{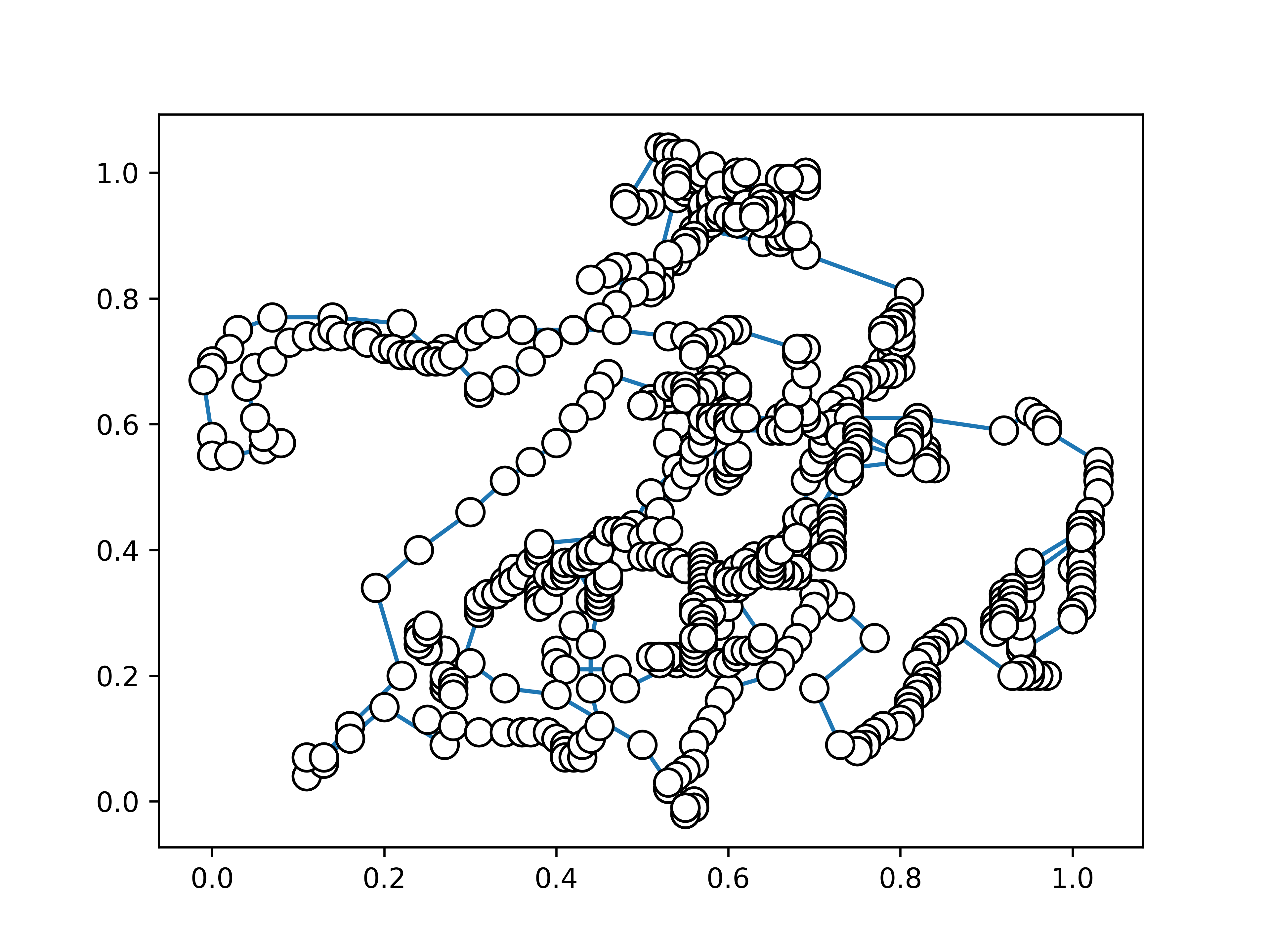}
\label{real4}}
\caption{Visualization of the eye tracking data: (a) quantized normalized eye-gazing data; (b) Release from independent randomized response LDP mechanisms (c) Release from correlated randomized response mechanism with the instantaneous release mechanism, (d) Release from correlated randomized response mechanism with  the batched sequential release mechanism (batch size =10). For (b),(c),(d), total budget $\epsilon=10$;}
\end{figure*}

\begin{figure}[htp]
\label{fig:upreal}
\centering 
\subfigure[Privacy leakage (one step Markov)] 
{ \includegraphics[width=0.4\textwidth]{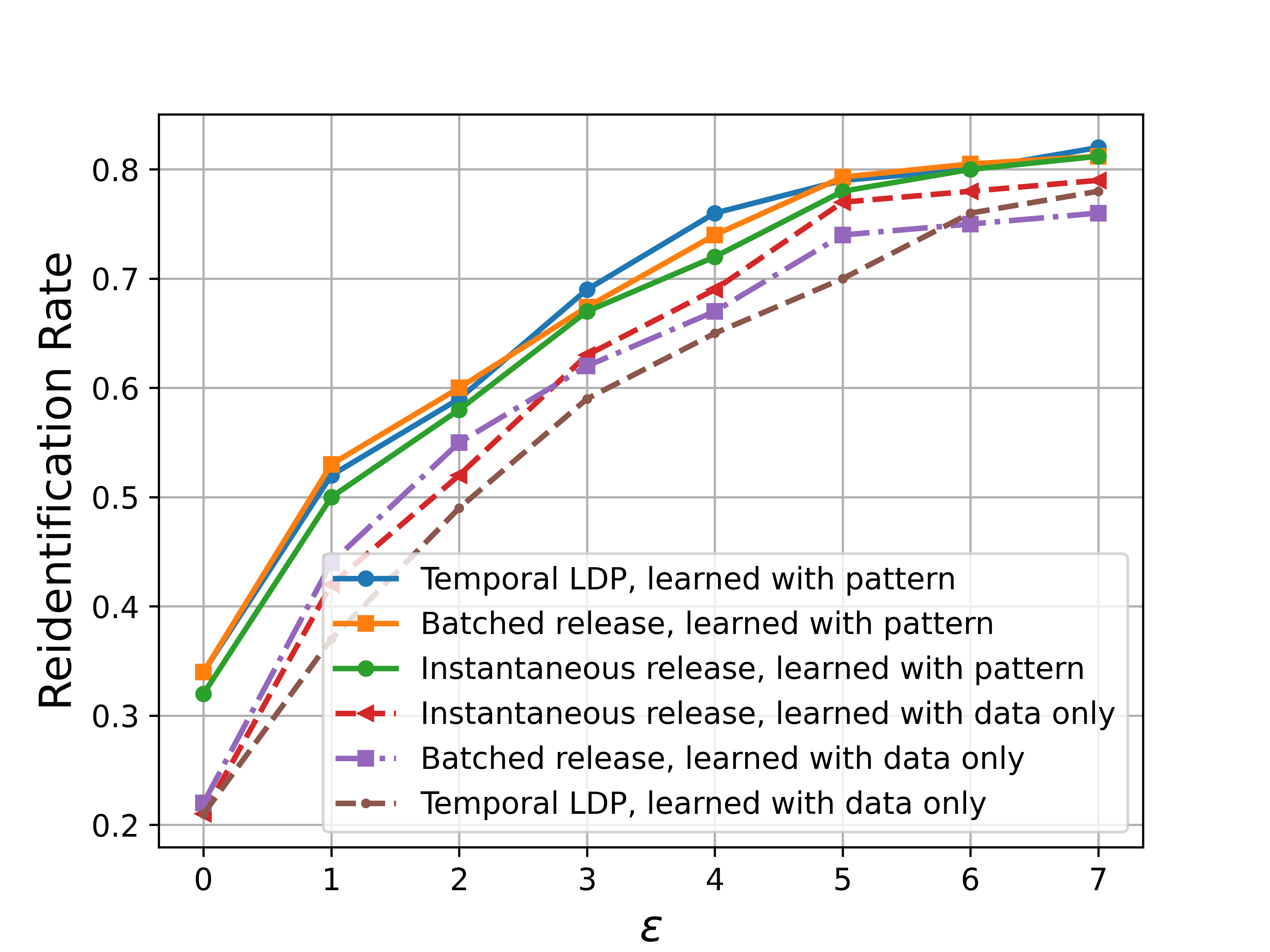}
\label{real5} }
\subfigure[Privacy leakage (two-step Markov)]
{ \includegraphics[width=0.4\textwidth]{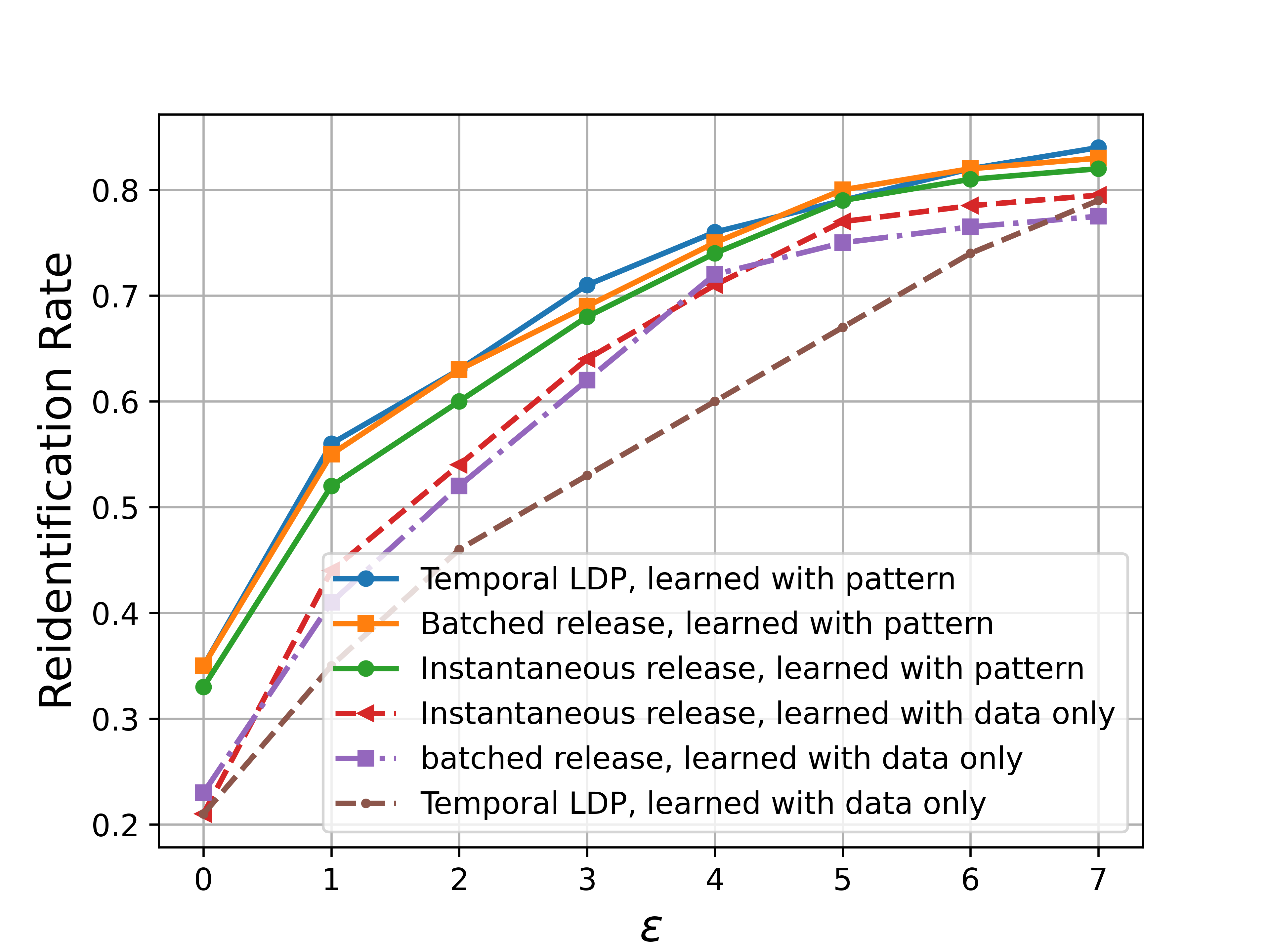}
\label{real6}}
\caption{Privacy leakage comparisons, privacy evaluated by the identification rate of individuals in the dataset; (a) considers a one-step Markov Chain, and (b) considers a two-step Markov Chain.}
\end{figure}

\begin{figure}[htp]
\label{fig:upreal2}
\centering 
\subfigure[Utility measure (one step Markov)]
{\includegraphics[width=0.4\textwidth]{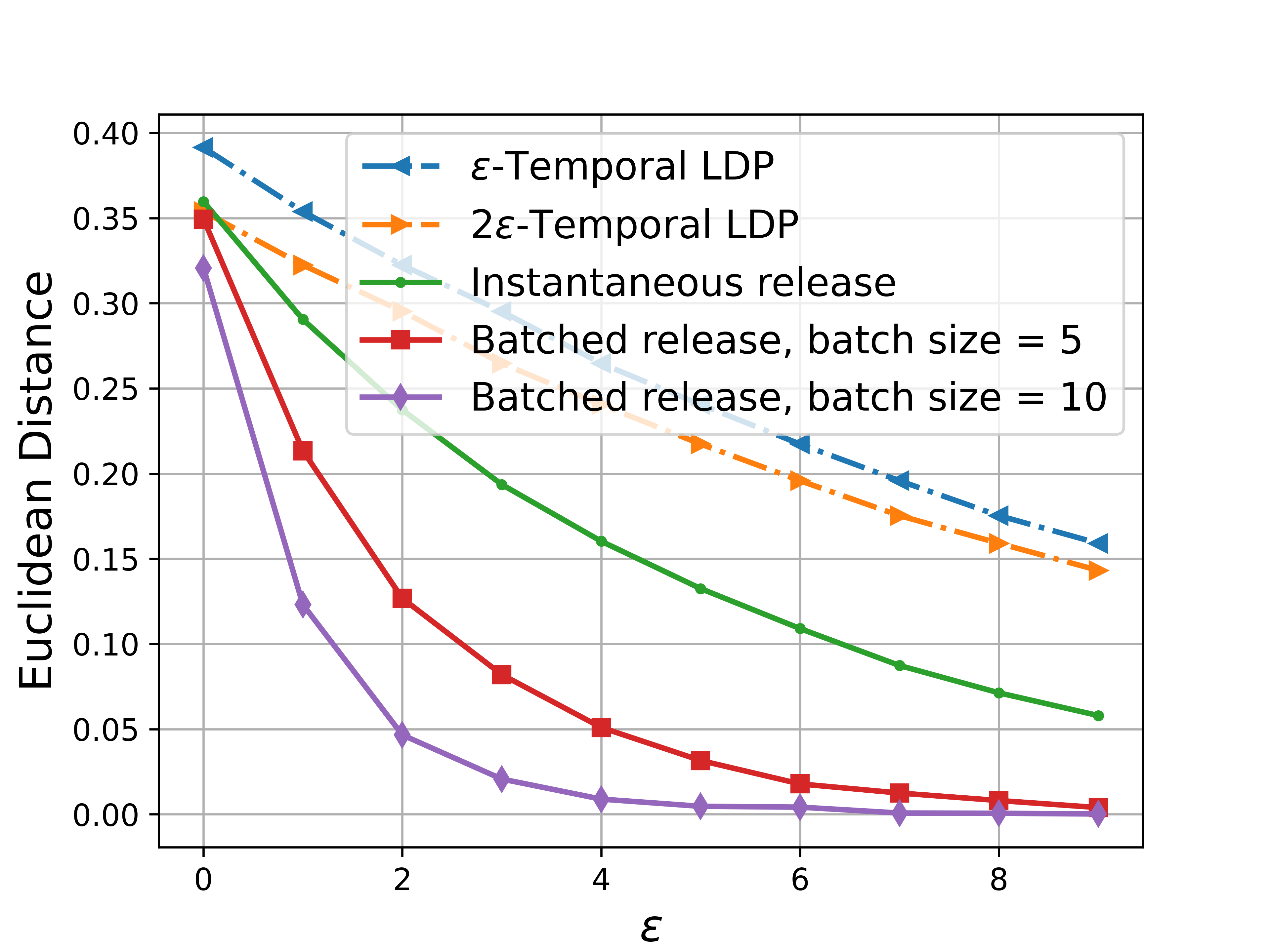}
\label{real7}}
\subfigure[Utility measure (two-step Markov)]
{\includegraphics[width=0.4\textwidth]{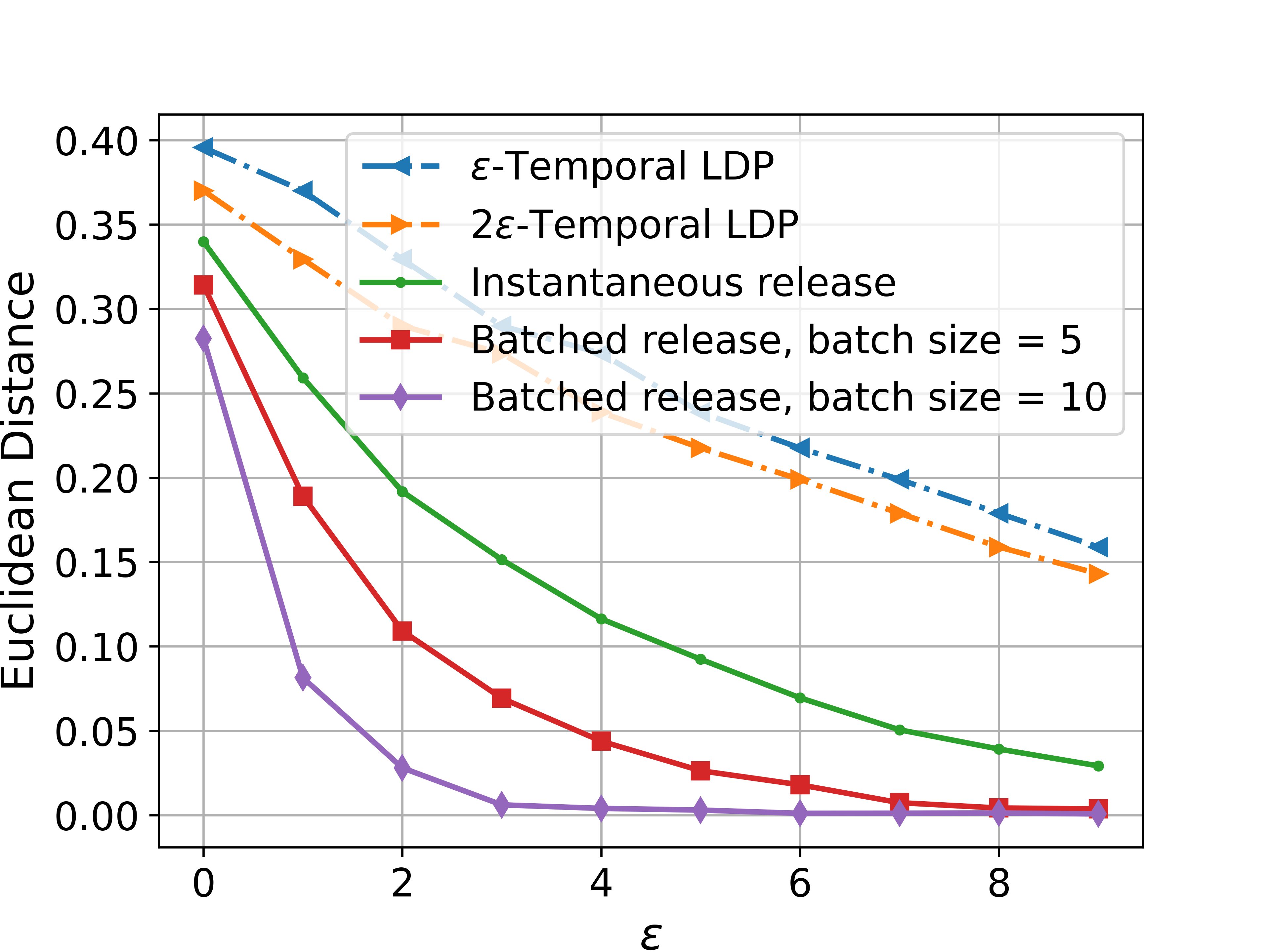}
\label{real8}}
\caption{Utility comparisons, utility measured by the averaged Euclidean distance between the raw and the released data at each time stamp.  (a) considers a one-step Markov Chain, and (b) considers a two-step Markov Chain.}
\end{figure}

\section{Experimental Results}\label{sec:eva}
In this section, we evaluate the performance of our proposed mechanism and compare it with other solutions from related works.

\subsection{Evaluations with Synthetic Data}
In our first evaluation, we conduct a comparative analysis between the proposed instantaneous release mechanism, batched release mechanisms, and the following two  mechanisms from the literature:
1) Randomized Response Mechanism Based on Local Differential Privacy (RR-LDP): This mechanism perturbs data values independently at each timestamp to satisfy $\epsilon_{k}$-LDP. Thanks to the sequential composition theorem of LDP, the global leakage measured by LDP after time $T$ is the sum of $\epsilon_k$ from $k=1$ to $T$. 2) The LDP Mechanism from \cite{Conti_release_LDP} (denoted as ToPS): This mechanism comprises three phases: threshold estimation, perturbation, and smoothing. Initially, the mechanism releases samples using the Laplacian mechanism and uses the first set of samples to estimate a threshold. A hierarchical method is then employed to handle the ranged values. Lastly, a post-processing (smoothing) phase ensures that the data sequence aligns with a specified distribution. Our comparison aims to measure the effectiveness of these techniques, highlighting their relative strengths and weaknesses in handling privacy-preserving data releases.

We first simulate by generating a binary time series that satisfies the first-order Markov chain with an initial probability of $P_1 = \text{Pr}(X_1 = 1)$ and transitional matrix
\begin{equation}
    \begin{bmatrix}
q_{00} & q_{01}\\
q_{10} & q_{11}
\end{bmatrix}
\end{equation}
where $q_{ij} = \text{Pr}(X_{k}=j|X_{k-1}=i)$. We then compare the mechanisms mentioned above under the following two settings: (a) $P_1 = 0.5$, $q_{00} = q_{11} = 0.5$, which means the data in the sequence are more correlated with each other. (b) $P_1 = 0.9$, $q_{00} = q_{11} = 0.9$, which means the dependence in the sequence is strong. The comparison results are shown in Fig. \ref{UP_compare}. 

Observe that generally, the  utility provided by the instantaneous release mechanism is sandwiched between $\epsilon$ and $2\epsilon$-RR-LDP mechanisms given any $\epsilon$. However, as the data dependence increases, $\epsilon$-SIP with instantaneous release even outperforms $\epsilon$-LDP. On the other hand, batched release models always outperform instantaneous release models, and the advantage increases with the batch size. It is also worth noting that when the data correlation is strong.  The advantages of context-aware models (SIP-based mechanisms) are even more obvious.

\subsection{Experiment with Real Data} 

\subsubsection{Click Streams Data (Kosarak)} 
Investigating the clicking streams for a website can be very helpful in learning product popularity or guiding web page design. On the other hand, individuals' clicking data
is privately sensitive as it may infer one's personal interest, working hours, daily behavior, etc. 

\textbf{Dataset:} Kosarak is a dataset of click streams on a Hungarian website that contains around one million users and 41270 categories. The data is formatted so that each user clicks through multiple categories. 
After data cleaning and preprocessing, there are $990002$ users in the dataset, each containing a data stream with lengths from $1$ to $2498$. We first extract each user's data stream and calculate the frequency of the occurrence of different patterns (every consecutive data value). For the instantaneous release model, we calculate the frequency of every pair of data values, and estimate the conditional distribution from the frequency. Then for batched release model, we first truncate the data streaming into several trucks according to the predefined batch size. Then we analyzed the frequency of every consecutive batch value. Then, we estimate the distribution from there.
Finally, we summarized a frequency lookup table, which is then transferred  to a correlation library. Given every possible current value, we are able to look up the following data within the correlation library with certain probabilities. Then we apply the mechanisms proposed in the previous section to the preprocessed model. 

\textbf{Utility and Privacy:} For our experiment with Kosarak, we randomly select $1000$ users' data sequence as the input. Then we perturb these data sequences with the mechanisms proposed in the previous sections. The utility of this experiment is measured by the empirical absolute distance between the input sequence and the output. The privacy on the other hand, is customized by varying the value of $\epsilon$. The comparison is shown in Fig.\ref{UP_compare2}. 

From the above figures, several key observations can be made:
1) Despite $\epsilon$-SIP operating within the bounds of $\epsilon$-LDP and $2\epsilon$-LDP, our proposed mechanisms consistently outperform $2\epsilon$-LDP.
2) In the case of batched release models, we observe a trend where larger batch sizes contribute to improved utility. This concept also extends to the instantaneous release model, which can be considered a special case of batched release with a batch size of one.
3) While the Truncated Output Perturbation under Local Differential Privacy (ToPs-LDP) mechanism improves data utility compared to the Randomized Response under Local Differential Privacy (RR-LDP), primarily due to sensitivity reduction achieved through truncation, it still performs worse than our SIP-based mechanisms. This disparity stems from ToPs-LDP's inability to capture data correlation, a deficiency not present in our proposed mechanisms.

\subsubsection{Experiments with Eye-tracking Data}
Eye-tracking data is usually collected online by AR/VR devices. Usually, cameras are embedded in these devices to track users' eyeball moving, and these axis data are uploaded to the server in return for services, such as online video games, online social, etc. However, studies have figured out several privacy concerns in eye-tracking data consumption: In \cite{10.1145/3314111.3319915}, Steil et al. pointed out that eye-tracking data can reveal one's many private sensitive information, such as age, gender, health status, sexual preference, personal trails, etc. Further, a group of researchers from Stanford University have shown they’re able to reliably identify individuals Using a pool of 511 participants. Their system is said to be capable of identifying $95\%$ of users correctly “when trained on less than 5 min of tracking data per person\cite{Miller_2020}.

\textbf{Dataset:} In this experiment, we compare the performance of our mechanisms with independent LDP-based mechanisms with the dataset of "MOJO", which is collected by "Mojo vision". The Mojo dataset contains ten users' eye-gazing data. Each user's data sequence can be viewed as a three-dimensional vector containing $X$, $Y$, and $Z$ axes, where $Z$ label measures image rotation, and we ignore this factor in the following experiment. The length of each individual's eye-gazing data sequence is different according to his/her recorded period. On average, each one of them has been collected for $5$ hours, with a sample rate of $50$ per second. 
We first normalize the coordinates to be within $[0,1]$. Then we equally divided the area into a $100\times 100$ space. Then we quantify each coordinate and cast the coordinates to the grid. We take a portion of one individual's normalized and quantified eye gazing sequence $(1/1000)$ of the raw data sequence. 

We assume the correlation within the data stream has Markov properties. Further, we assume a first-order and second-order Markov chain in the stream. Then the correlation is measured by the conditional probabilities, which can be estimated by frequency checking in the eye-gazing stream. We assume the initial probability is uniformly distributed.

We consider three types of data release mechanisms: (1) Temporal LDP mechanism proposed in \cite{Quantify_DP}. We first assign a total budget $\epsilon=10$ to the data sequence every second. Then, the local privacy budget can be calculated by the sequential composition Theorem. Note that support of the randomized response mechanism is selected according to the realizations with non-zero conditional probabilities given the previous values (according to the Markov properties). (2) Conditional randomized response mechanism for the instantaneous release model. The parameters involved in this mechanism are summarized in this paper. (3) Batched data release mechanism, where we assume the batch sizes are 5 and 10, respectively. The perturbation parameters are calculated by numerically solving the optimization problem defined in \eqref{eq:batch}. A visualization of the release of different mechanisms is shown in Fig. 7.

\textbf{Privacy Leakage:} Since the privacy protection guarantees provided by different mechanisms are different, we want to see how they protect the re-identification rate of each individual under fixed epsilons. We first design RNN models that leverage data correlations in predictions. Models are trained with half of the eye-gazing data streams from each user. The goal is to learn the eye-moving patterns of each individual. Then, we consider two scenarios regarding the RNN's prediction model: one considers the data correlation estimated from frequency (for the whole data stream) and leverages it in the data prediction; the other one makes predictions solely depending on the observed data. 

The privacy leakage comparison of different models is shown in Fig. 8. Observe that when the reidentification model is trained with data only, the temporal LDP-based mechanism provides the most strict privacy protection. When the adversary trains the model with the pattern information, while the overall leakage of all mechanisms increases, the increase of SIP-based mechanisms is relatively smaller, and the leakage of SIP-based mechanisms is even smaller  than LDP-based mechanisms. The reason is prior information has already been  considered in the definition of SIP and noise is added more effectively to mitigate the privacy leakage against such adversaries.

\textbf{Query Utility:} The utility of different mechanisms is measured by the average Euclidean distance between the raw and released data at each time stamp. The utility comparisons of different models are presented  in Fig. 9. 

Observe that even though different mechanisms achieve similar privacy protections, their utilities are very different. SIP-based batched release models outperform other models significantly and the gap is even larger with bigger batch size. On the other hand, the instantaneous release model provides better utility even than the $2\epsilon$-temporal LDP. This is because even though temporal LDP leverages correlations to reduce sensitivity, the definition and the mechanism do not consider the correlation, while SIP is totally context-aware.

\section{Conclusion and Future Works}

In this paper, we tackle the challenge of releasing streaming data while preserving privacy. Initially, we  introduce Sequence Information Privacy (SIP), which is an extension of local information privacy to sequential data. SIP inherently acknowledges data correlations within its definition. Subsequently, we study its relationships with existing privacy concepts such as Local Differential Privacy and Pufferfish Privacy. We demonstrate that our SIP can be sequentially decomposed into individual local leakages, making the optimization of global utility-privacy tradeoff equivalent to independently solving each local instance. Based on two data release settings, instantaneous and batched release, we propose perturbation mechanisms that optimize this utility-privacy tradeoff. Our evaluation, using both synthetic and real data, compares the utility-privacy tradeoffs provided by our proposed mechanisms with those from existing works. Results indicate that our mechanisms can significantly enhance data utility without compromising data privacy. 

In terms of future work, we are interested in the following directions.One direction is to remove the assumptions that the prior distribution and data correlation are known. We can make the mechanism release the first several data points in a context-free manner (for example, using LDP). As more observations are made, the data prior/correlation becomes more certain, and the mechanism leverages the context to achieve context-aware utility-privacy tradeoffs. Another direction to consider is that the batched release mechanism still involves high computational complexity. We would like to investigate further ways to reduce the computational complexity without violating privacy constraints.

\bibliographystyle{IEEEtran}
\bibliography{ref}

\begin{appendices}
\section{Proof for theorem 1.}
\begin{proof}
One of the metrics in the total leakage can be expressed as:
\begin{equation}\label{eq:proofthm1}
    \begin{aligned}
    &\frac{\text{Pr}(\bold{X}_1^T=\bold{x}_1^T|\bold{Y}_1^T=\bold{y}_1^T)}{\text{Pr}(\bold{X}_1^T=\bold{x}_1^T)}\\
    =&\frac{\text{Pr}(\bold{X}_1^T=\bold{x}_1^T|\bold{Y}_1^T=\bold{y}_1^T)}{\text{Pr}(\bold{X}_1^T=\bold{x}_1^T|\bold{Y}_1^{T-1}=\bold{y}_1^{T-1})}\frac{\text{Pr}(\bold{X}_1^T=\bold{x}_1^T|\bold{Y}_1^{T-1}=\bold{y}_1^{T-1})}{\text{Pr}(\bold{X}_1^T=\bold{x}_1^T)}\\
    \le &e^{\epsilon_T}\frac{\text{Pr}(X_T={x}_T|\bold{Y}_1^{T-1}=\bold{y}_1^{T-1},\bold{X}_1^{T-1}=\bold{x}_1^{T-1})}{\text{Pr}({X}_T={x}_T|\bold{X}_1^{T-1}=\bold{x}_1^{T-1})}\\
    &~~~~~~~~~~~~~~~~~~~~~~~~~~\cdot \frac{\text{Pr}(\bold{X}_1^{T-1}=\bold{x}_1^{T-1}|\bold{Y}_1^{T-1}=\bold{y}_1^{T-1})}{\text{Pr}(\bold{X}_1^{T-1}=\bold{x}_1^{T-1})}\\
    =&e^{\epsilon_T}\frac{\text{Pr}(\bold{X}_1^{T-1}=\bold{x}_1^{T-1}|\bold{Y}_1^{T-1}=\bold{y}_1^{T-1})}{\text{Pr}(\bold{X}_1^{T-1}=\bold{x}_1^{T-1})}\\
    \le &\prod_{k=1}^Te^{\epsilon_k},
    \end{aligned}
\end{equation}
where $\text{Pr}(X_T={x}_T|\bold{Y}_1^{T-1}=\bold{y}_1^{T-1},\bold{X}_1^{T-1}=\bold{x}_1^{T-1})=\text{Pr}(X_T={x}_T|\bold{X}_1^{T-1}=\bold{x}_1^{T-1})$, because $$Y_k\independent\bold{X}_{k+1}^T|\{\bold{X}_1^k,\bold{Y}_1^{k-1}\}.$$
Similarly, with respect to the other metric, we have:
\begin{equation}
    \begin{aligned}
    \frac{\text{Pr}(\bold{X}_1^T=\bold{x}_1^T|\bold{Y}_1^T=\bold{y}_1^T)}{\text{Pr}(\bold{X}_1^T=\bold{x}_1^T)}\ge \prod_{j=1}^Te^{-\epsilon_j},
    \end{aligned}
\end{equation}
which concludes our proof.
\end{proof}

\section{Proof of Theorem 2}
\begin{proof}
    From \eqref{eq:proofthm1}, $\mathcal{L}(\bold{Y}_1^T\to \bold{X}_1^T) = \sum_{k=1}^T \mathcal{L}(Y_k\to \bold{X}_1^k)$, where by Bayes rule:
    \begin{equation}
    \begin{aligned}
        &\frac{\text{Pr}(\bold{X}_1^k=\bold{x}_1^k|\bold{Y}_1^k=\bold{y}_1^k)}{\text{Pr}(\bold{X}_1^k=\bold{x}_1^k|\bold{Y}_1^{k-1}=\bold{y}_1^{k-1})}\\
        =&\frac{\text{Pr}(Y_k=y_k|\bold{X}_1^k=\bold{x}_1^k,\bold{Y}_1^{k-1}=\bold{y}_1^{k-1})}{\text{Pr}(Y_k=y_k|\bold{Y}_1^{k-1}=\bold{y}_1^{k-1})}
    \end{aligned}
    \end{equation}
    Denote $\text{Pr}(Y_k=y_k|\bold{X}_1^k=\bold{x}_1^k,\bold{Y}_1^{k-1}=\bold{y}_1^{k-1})$ as $P_k$ and $\text{Pr}(Y_k=y_k|\bold{Y}_1^{k-1}=\bold{y}_1^{k-1})=Q_k$. Define $c_k=\log \frac{P_k}{Q_k}$, then $ \mathcal{L}(Y_k\to \bold{X}_1^k)$ implies $|c_k|\le \epsilon$. 
    On the other hand, 
    \begin{equation}
        \begin{aligned}
        E[c_k] = &\sum_{y_k} P_k \log \frac{P_k}{Q_k}\\
        \le& \sum_{y_k} P_k \log \frac{P_k}{Q_k} + \sum_{y_k} Q_k \log \frac{Q_k}{P_k}\\
        =& \sum_{y_k} P_k (\log\frac{P_k}{Q_k}+\log\frac{Q_k}{P_k}) + \sum_{y_k} (Q_k - P_k) \log \frac{Q_k}{P_k}\\
        =&\sum_{y_k} (Q_k - P_k) \log \frac{Q_k}{P_k}\\
        \le& \sum_{y_k} |Q_k-P_k| |\log\frac{Q_k}{P_k}|\\
        \le& \epsilon \sum_{y_k} |e^{\epsilon}P_k - P_k|\\
        = & \epsilon (e^{\epsilon} -1 ).
        \end{aligned}
    \end{equation}
    With Azuma-hoeffding inequality, if each $|c_k|\le{\epsilon}$ and each $E[c_k]\le \epsilon (e^{\epsilon} -1 )$. 
    \begin{equation}
        \text{Pr}\left(\sum_{k=1}^T c_k > T\epsilon(e^{\epsilon}-1)+\sqrt{T}\epsilon\sqrt{2\ln(1/\delta)}\right)\le \delta.
    \end{equation}
    This concludes the proof for the first half, the second half of the batch release mechanism follows the same idea.
\end{proof}

\section{Proof of Theorem 3}
\begin{proof}
\textbf{From $\epsilon$-LDP $\to$ $\epsilon$-SIP:}

    By the Bayesian rule, the privacy metric in the sequence information leakage can be expressed as:
    \begin{equation}
        \begin{aligned}
        \frac{\text{Pr}(\bold{Y}_1^T=\bold{y}_1^T|\bold{X}_1^T=\bold{x}_1^{T})}{\text{Pr}(\bold{Y}_1^T=\bold{y}_1^T)}, \frac{\text{Pr}(\bold{Y}_1^T=\bold{y}_1^T)}{\text{Pr}(\bold{Y}_1^T=\bold{y}_1^T|\bold{X}_1^T=\bold{x}_1^{T})}
        \end{aligned}
    \end{equation}
   When $\mathcal{M}$ satisfies $\epsilon$-LDP, 
   \begin{equation}
   \begin{aligned}
       e^{-\epsilon} \text{Pr}(\bold{Y}_1^T=\bold{y}_1^T|\bold{X}_1^T=\bold{x}_1^{T})
       &\le \text{Pr}(\bold{Y}_1^T=\bold{y}_1^T|\bold{X}_1^T=\tilde{\bold{x}}_1^{T}) \\
       &\le e^{\epsilon}\text{Pr}(\bold{Y}_1^T=\bold{y}_1^T|\bold{X}_1^T=\bold{x}_1^{T})
    \end{aligned}
   \end{equation}
   Therefore: 
    \begin{equation}
        \begin{aligned}
        \text{Pr}(\bold{Y}_1^T=\bold{y}_1^T)=&\sum_{\tilde{\bold{x}}_1^{T}}\text{Pr}(\bold{Y}_1^T=\bold{y}_1^T|\bold{X}_1^T=\tilde{\bold{x}}_1^{T})\text{Pr}(\bold{X}_1^T=\tilde{\bold{x}}_1^{T})\\
        \le & \sum_{\tilde{\bold{x}}_1^{T}}e^{\epsilon}\text{Pr}(\bold{Y}_1^T=\bold{y}_1^T|\bold{X}_1^T={\bold{x}}_1^{T})\text{Pr}(\bold{X}_1^T=\tilde{\bold{x}}_1^{T})\\
        =&e^{\epsilon}\text{Pr}(\bold{Y}_1^T=\bold{y}_1^T|\bold{X}_1^T={\bold{x}}_1^{T}).
        \end{aligned}
    \end{equation}
    For the other direction:
    \begin{equation}
        \text{Pr}(\bold{Y}_1^T=\bold{y}_1^T)\ge e^{-\epsilon}\text{Pr}(\bold{Y}_1^T=\bold{y}_1^T|\bold{X}_1^T={\bold{x}}_1^{T}).
    \end{equation}
    As a result:
    \begin{equation}
        \frac{\text{Pr}(\bold{Y}_1^T=\bold{y}_1^T)}{\text{Pr}(\bold{Y}_1^T=\bold{y}_1^T|\bold{X}_1^T={\bold{x}}_1^{T})}\in{[e^{-\epsilon},e^{\epsilon}]},
    \end{equation}
    and the sequence information leakage is upper bounded by $\epsilon$.
    
\textbf{From $\epsilon$-SIP $\to$ $2\epsilon$-LDP:}

Suppose $\bold{x}_1^T$ and $\tilde{\bold{x}}_1^T$ are two realizations of $\bold{X}_1^T$. If $\mathcal{M}$ guarantees $\mathcal{L}(\bold{Y}_1^T\to{\bold{X}_1^T})\le {\epsilon}$:
    \begin{equation}
        \frac{\text{Pr}(\bold{Y}_1^T=\bold{y}_1^T)}{\text{Pr}(\bold{Y}_1^T=\bold{y}_1^T|\bold{X}_1^T={\bold{x}}_1^{T})}\in{[e^{-\epsilon},e^{\epsilon}]}.
    \end{equation}
% and 
% \begin{equation}
%         \frac{\text{Pr}(\bold{Y}_1^T=\bold{y}_1^T)}{\text{Pr}(\bold{Y}_1^T=\bold{y}_1^T|\bold{X}_1^T=\tilde{\bold{x}}_1^{T})}\in{[e^{-\epsilon},e^{\epsilon}]}.
%     \end{equation}
    Therefore:
    \begin{equation}
    \begin{aligned}
        &\frac{\text{Pr}(\bold{Y}_1^T=\bold{y}_1^T)}{\text{Pr}(\bold{Y}_1^T=\bold{y}_1^T|\bold{X}_1^T={\bold{x}}_1^{T})}\cdot \frac{\text{Pr}(\bold{Y}_1^T=\bold{y}_1^T|\bold{X}_1^T=\tilde{\bold{x}}_1^{T})}{\text{Pr}(\bold{Y}_1^T=\bold{y}_1^T)}\\
        =&\frac{\text{Pr}(\bold{Y}_1^T=\bold{y}_1^T|\bold{X}_1^T=\tilde{\bold{x}}_1^{T})}{\text{Pr}(\bold{Y}_1^T=\bold{y}_1^T|\bold{X}_1^T={\bold{x}}_1^{T})}\in[e^{-2\epsilon},e^{2\epsilon}]
        \end{aligned}
        \end{equation}
\end{proof}

\begin{figure*}[t]
\vbox{
 \begin{equation}
 \label{eq49}
 \small
\begin{split}
 \nabla{f}=\left[\begin{array}{c}
       \beta_k(\bold{x}_1^{k-1},1|\bold{y}_1^{k-1})\\\beta_k(\bold{x}_2^{k-1},x|\bold{y}_1^{k-1})\\...\\\beta_k(\bold{x}_1^{k-1},|\mathcal{X}||\bold{y}_1^{k-1})
    \end{array}\right]
    \left[\begin{array}{ccc}
       D[Q(1),Q(1)]&D[Q(1),Q(2)]...&D[Q(1),Q(|\mathcal{Y}|)]\\
       D[Q(2),Q(1)]&D[Q(2),Q(2)]...&D[Q(2),Q(|\mathcal{Y}|)]\\
       ...&...&...\\
       D[Q(|\mathcal{X}|),Q(1)]&D[Q(|\mathcal{X}|),Q(2)]...&D[Q(|\mathcal{X}|),Q(|\mathcal{Y}|)]\\
    \end{array}\right]\\
\end{split}
\end{equation}}
% \hline
\end{figure*}

\begin{figure*}[t]
\vbox{
 \begin{equation}
 \label{eq50}
 \small
\begin{split}
 \begin{aligned}
       &g_1=\frac{\sum_{x\in{\mathcal{X}}}\beta_k(\bold{x}_1^{k-1},x|\bold{y}_1^{k-1})a_k(1|\bold{x}_1^{k-1},x,\bold{y}_1^{k-1})}{a_k(1|\bold{x}_1^{k-1},1,\bold{y}_1^{k-1})}-e^{\epsilon},~~~~~~g_2=e^{-\epsilon}-\frac{\sum_{x\in{\mathcal{X}}}\beta_k(\bold{x}_1^{k-1},x|\bold{y}_1^{k-1})a_k(1|\bold{x}_1^{k-1},x,\bold{y}_1^{k-1})}{a_k(1|\bold{x}_1^{k-1},1,\bold{y}_1^{k-1})},\\
       %&g_3=\frac{\sum_{x\in{\mathcal{X}}}P_X(x)M_{x1}}{M_{21}}-e^{\epsilon},~~~~~~~~g_4=e^{-\epsilon}-\frac{\sum_{x\in{\mathcal{X}}}P_X(x)M_{x1}}{M_{21}},\\
       &...~~~~~~~~~~~~~~~~~~~~~~~~~~~~~~~~~~~~~~~~~~~~~~~~~~~~~...\\
       &g_{2|\mathcal{X}|-1}=\frac{\sum_{x\in{\mathcal{X}}}\beta_k(\bold{x}_1^{k-1},x|\bold{y}_1^{k-1})a_k(1|\bold{x}_1^{k-1},x,\bold{y}_1^{k-1})}{a_k(1|\bold{x}_1^{k-1},|\mathcal{X}|,\bold{y}_1^{k-1})}-e^{\epsilon},~~~g_{2|\mathcal{X}|}=e^{-\epsilon}-\frac{\sum_{x\in{\mathcal{X}}}\beta_k(\bold{x}_1^{k-1},x|\bold{y}_1^{k-1})a_k(1|\bold{x}_1^{k-1},x,\bold{y}_1^{k-1})}{a_k(1|\bold{x}_1^{k-1},|\mathcal{X}|,\bold{y}_1^{k-1})},\\
       &...~~~~~~~~~~~~~~~~~~~~~~~~~~~~~~~~~~~~~~~~~~~~~~~~~~~~~~...\\
       &g_{2|\mathcal{X}||\mathcal{Y}|-1}=\frac{\sum_{x\in{\mathcal{X}}}\beta_k(\bold{x}_1^{k-1},x|\bold{y}_1^{k-1})a_k(|\mathcal{Y}||\bold{x}_1^{k-1},x,\bold{y}_1^{k-1})}{a_k(|\mathcal{Y}||\bold{x}_1^{k-1},|\mathcal{X}|,\bold{y}_1^{k-1})}-e^{\epsilon},~g_{2|\mathcal{X}||\mathcal{Y}|}=e^{-\epsilon}-\frac{\sum_{x\in{\mathcal{X}}}\beta_k(\bold{x}_1^{k-1},x|\bold{y}_1^{k-1})a_k(|\mathcal{Y}||\bold{x}_1^{k-1},x,\bold{y}_1^{k-1})}{a_k(|\mathcal{Y}||\bold{x}_1^{k-1},|\mathcal{X}|,\bold{y}_1^{k-1})}.\\
    \end{aligned}
\end{split}
\end{equation}}
% \hline
\end{figure*}

\begin{figure*}[t]
\vbox{
 \begin{equation}
 \label{eq52}
 \small
\begin{split}
  \nabla g_{1}=\left[\begin{array}{cccc}
       \frac{\sum_{x\neq{1}}\beta_k(\bold{x}_1^{k-1},x|\bold{y}_1^{k-1})a_k(1|\bold{x}_1^{k-1},x,\bold{y}_1^{k-1})}{a^2_k(1|\bold{x}_1^{k-1},1,\bold{y}_1^{k-1})}&0&...&0\\
       \frac{\beta_k(\bold{x}_1^{k-1},2|\bold{y}_1^{k-1})}{a_k(1|\bold{x}_1^{k-1},1,\bold{y}_1^{k-1})}&0&...&0\\
       ...&...&...&...\\
       \frac{\beta_k(\bold{x}_1^{k-1},|\mathcal{X}||\bold{y}_1^{k-1})}{a_k(1|\bold{x}_1^{k-1},1,\bold{y}_1^{k-1})}&0&...&0\\
    \end{array}\right],
    \nabla g_{2}=\left[\begin{array}{cccc}
       \frac{\sum_{-x\neq{1}}\beta_k(\bold{x}_1^{k-1},x|\bold{y}_1^{k-1})a_k(1|\bold{x}_1^{k-1},x,\bold{y}_1^{k-1})}{a^2_k(1|\bold{x}_1^{k-1},1,\bold{y}_1^{k-1})}&0&...&0\\
       -\frac{\beta_k(\bold{x}_1^{k-1},2|\bold{y}_1^{k-1})}{a_k(1|\bold{x}_1^{k-1},1,\bold{y}_1^{k-1})}&0&...&0\\
       ...&...&...&...\\
       -\frac{\beta_k(\bold{x}_1^{k-1},|\mathcal{X}||\bold{y}_1^{k-1})}{a_k(1|\bold{x}_1^{k-1},1,\bold{y}_1^{k-1})}&0&...&0\\
    \end{array}\right],
\end{split}
\end{equation}}
% \hline
\end{figure*}
\begin{figure*}[t]
\vbox{
 \begin{equation}
 \label{eq53}
 \small
\begin{split}
  \begin{aligned}
    &\sum_{i=1}^{2|\mathcal{X}|}\mu_{i}\nabla{g_i}\\=&\left[\begin{array}{c}
       \frac{(\mu_1-\mu_2)\sum_{x\neq{1}}a_k(1|\bold{x}_1^{k-1},x,\bold{y}_1^{k-1})\beta_k(\bold{x}_1^{k-1},x|\bold{y}_1^{k-1})}{a^2_k(1|\bold{x}_1^{k-1},1,\bold{y}_1^{k-1})}+\sum_{i=1}^{|\mathcal{X}|}(\mu_{2i-1}-\mu_{2i})\frac{\beta_k(\bold{x}_1^{k-1},1|\bold{y}_1^{k-1})}{a_k(1|\bold{x}_1^{k-1},i,\bold{y}_1^{k-1})}-(\mu_{1}-\mu_{2})\frac{\beta_k(\bold{x}_1^{k-1},1|\bold{y}_1^{k-1})}{a_k(1|\bold{x}_1^{k-1},1,\bold{y}_1^{k-1})}\\
       \frac{(\mu_3-\mu_4)\sum_{x\neq{1}}a_k(1|\bold{x}_1^{k-1},x,\bold{y}_1^{k-1})\beta_k(\bold{x}_1^{k-1},x|\bold{y}_1^{k-1})}{a^2_k(1|\bold{x}_1^{k-1},2,\bold{y}_1^{k-1})}+\sum_{i=1}^{|\mathcal{X}|}(\mu_{2i-1}-\mu_{2i})\frac{\beta_k(\bold{x}_1^{k-1},2|\bold{y}_1^{k-1})}{a_k(1|\bold{x}_1^{k-1},i,\bold{y}_1^{k-1})}-(\mu_{3}-\mu_{4})\frac{\beta_k(\bold{x}_1^{k-1},2|\bold{y}_1^{k-1})}{a_k(1|\bold{x}_1^{k-1},2,\bold{y}_1^{k-1})}\\
       ...\\
       \frac{(\mu_{2|\mathcal{X}-1|}-\mu_{2|\mathcal{X}|})\sum_{x\neq{1}}a_k(1|\bold{x}_1^{k-1},x,\bold{y}_1^{k-1})\beta_k(\bold{x}_1^{k-1},x|\bold{y}_1^{k-1})}{a_k(1|\bold{x}_1^{k-1},|\mathcal{X}|,\bold{y}_1^{k-1})}+\sum_{i=1}^{|\mathcal{X}|}(\mu_{2i-1}-\mu_{2i})\frac{\beta_k(\bold{x}_1^{k-1},|\mathcal{X}||\bold{y}_1^{k-1}))}{a_k(1|\bold{x}_1^{k-1},i,\bold{y}_1^{k-1})}-(\mu_{2|\mathcal{X}|-1}-\mu_{2|\mathcal{X}|})\frac{\beta_k(\bold{x}_1^{k-1},|\mathcal{X}||\bold{y}_1^{k-1}))}{a_k(1|\bold{x}_1^{k-1},|\mathcal{X}|,\bold{y}_1^{k-1})}\\
    \end{array}\right],
    \end{aligned}
\end{split}
\end{equation}}
% \hline
\end{figure*}
\begin{figure*}[t]
\vbox{
 \begin{equation}
 \label{eq49}
 \small
\begin{split}
\begin{aligned}
&\left[\begin{array}{c}
       \beta_k(\bold{x}_1^{k-1},1|\bold{y}_1^{k-1})\\\beta_k(\bold{x}_1^{k-1},2|\bold{y}_1^{k-1})\\...\\\beta_k(\bold{x}_1^{k-1},|\mathcal{X}||\bold{y}_1^{k-1})
    \end{array}\right]
    \left[\begin{array}{c}
       D[Q(1),Q(1)]\\
       D[Q(2),Q(1)]\\
       ...\\
       D[Q(|\mathcal{X}|),Q(1)]\\
    \end{array}\right]\\
    &-\left[\begin{array}{c}
       \frac{(\mu_1-\mu_2)\beta_k(\bold{x}_1^{k-1},1|\bold{y}_1^{k-1})(1-\beta_k(\bold{x}_1^{k-1},1|\bold{y}_1^{k-1}))}{a^2_k(1|\bold{x}_1^{k-1},1,\bold{y}_1^{k-1})}+\sum_{i\neq1}\mu_{2i-1}e^{\epsilon}\\
       \frac{\mu_3\beta_k(\bold{x}_1^{k-1},2|\bold{y}_1^{k-1})(1-\beta_k(\bold{x}_1^{k-1},2|\bold{y}_1^{k-1}))}{a^2_k(1|\bold{x}_1^{k-1},2,\bold{y}_1^{k-1})}+\sum_{i\neq2}\mu_{2i-1}\frac{\beta_k(\bold{x}_1^{k-1},2|\bold{y}_1^{k-1})}{\beta_k(\bold{x}_1^{k-1},1|\bold{y}_1^{k-1})}e^{\epsilon}\\
       ...\\
       \frac{\mu_{2|\mathcal{X}|-1}\beta_k(\bold{x}_1^{k-1},|\mathcal{X}||\bold{y}_1^{k-1})(1-\beta_k(\bold{x}_1^{k-1},|\mathcal{X}||\bold{y}_1^{k-1}))}{a^2_k(1|\bold{x}_1^{k-1},|\mathcal{X}|,\bold{y}_1^{k-1})}+\sum_{i\neq|\mathcal{X}|}\mu_{2i-1}\frac{\beta_k(\bold{x}_1^{k-1},|\mathcal{X}||\bold{y}_1^{k-1}))}{\beta_k(\bold{x}_1^{k-1},1|\bold{y}_1^{k-1})}e^{\epsilon}\end{array}\right]-
       \left[\begin{array}{c}
       \lambda_1\\
       \lambda_2\\
       ...\\
       \lambda_{|\mathcal{X}|}
    \end{array}\right]=0,
\end{aligned}
\end{split}
\end{equation}}
% \hline
\end{figure*}

\section{Optimal parameters in the example}

Suppose at time step 1, $y_1=1$ is the output.

For the time step 2. We have:
\begin{equation*}
\begin{aligned}
    &\frac{\text{Pr}(Y_2=y_2|X_1=x_1,X_2=x_2,Y_1=1)}{\text{Pr}(Y_2=y_2|Y_1=1)}\\
    =&\frac{q(y_2)_{x_1x_2}^{1}}{\sum_{x'_1,x'_2}q(y_2)_{x_1x_2}^{1}\text{Pr}(X_1=x_1,X_2=x_2|Y_1=1)}
\end{aligned}
\end{equation*}
To find the feasible region for $q(y_2)_{x_1x_2}^{1}$, we need to calculate each possible posterior of $\text{Pr}(X_1=x_1,X_2=x_2|Y_1=1)$, which can be further expressed as:
\begin{equation*}
    \begin{aligned}
        &\text{Pr}(X_1=x_1,X_2=x_2|Y_1=1)\\
        =&\text{Pr}(X_2=x_2|X_1=x_1)\text{Pr}(X_1=x_1|Y_1=1)\\
        =&\text{Pr}(X_2=x_2|X_1=x_1)\frac{q(1)_{x_1}\text{Pr}(X_1=x_1)}{\sum_{x'_1=0}^1 q(1)_{x'_1}\text{Pr}(X_1=x'_1)}\\
    \end{aligned}
\end{equation*}
As there are four combinations of $X_1$ and $X_2$ given $Y_1=1$: $(X_1=0,X_2=0)$, $(X_1=0,X_2=1)$, $(X_1=1,X_2=0)$, $(X_1=1,X_2=1)$, with the conditional marginal probabilities of:
\begin{equation}
\begin{aligned}
    &\text{Pr}(X_1=0,X_2=0|Y_1=1)\\=&\frac{(1-\phi)q(1)_0(1-P_1)}{q(1)_0(1-P_1)+q(1)_1P_1}=\frac{(1-\phi)(1-P_1)}{e^{\epsilon}}
\end{aligned}
\end{equation}
\begin{equation}
\begin{aligned}
    &\text{Pr}(X_1=0,X_2=1|Y_1=1)\\=&\frac{\phi q(1)_0(1-P_1)}{q(1)_0(1-P_1)+q(1)_1P_1}=\frac{\phi(1-P_1)}{e^{\epsilon_1}}
\end{aligned}
\end{equation}
\begin{equation}
\begin{aligned}
    &\text{Pr}(X_1=1,X_2=0|Y_1=1)\\=&\frac{\phi q(1)_1P_1}{q(1)_0(1-P_1)+q(1)_1P_1}=\frac{\phi(e^{\epsilon_1}-1+P_1)}{e^{\epsilon_1}}
\end{aligned}
\end{equation}
\begin{equation}
\begin{aligned}
    &\text{Pr}(X_1=1,X_2=1|Y_1=1)\\=&\frac{(1-\phi) q(1)_1P_1}{q(1)_0(1-P_1)+q(1)_1P_1}=\frac{(1-\phi)(e^{\epsilon_1}-1+P_1)}{e^{\epsilon_1}}
\end{aligned}
\end{equation}
As we want to make $X_2$ as close to $Y_1$ as possible to increase the utility of the data, which means we want to increase the probabilities of $q(1)^{1}_{x_11}$ and $q(0)^{1}_{x_10}$ while decreasing the probabilities of $q(1)^{1}_{x_10}$ and $q(0)^{1}_{x_11}$, with linear constraints, we can found the boundaries of the parameters:

\begin{equation*}
    \begin{aligned}
        &q(1)^{1}_{00}=q(1)^{1}_{10}=\frac{\phi(1-P_1)+(1-\phi)(e^{\epsilon_1}-1+P_1)}{e^{\epsilon_1+\epsilon_2}}\\
        &q(0)^{1}_{01}=q(0)^{1}_{11}=\frac{(1-\phi)(1-P_1)+\phi(e^{\epsilon_1}-1+P_1)}{e^{\epsilon_1+\epsilon_2}}
    \end{aligned}
\end{equation*}
Similarly, we have given $Y_1=0$:
\begin{equation*}
    \begin{aligned}
        &q(1)^{0}_{00}=q(1)^{0}_{10}=\frac{\phi(e^{\epsilon_1  }-P_1)+(1-\phi)P_1}{e^{\epsilon_1+\epsilon_2}}\\
        &q(0)^{0}_{01}=q(0)^{0}_{11}=\frac{(1-\phi)(e^{\epsilon_1}-P_1)+\phi P_1}{e^{\epsilon_1+\epsilon_2}}
    \end{aligned}
\end{equation*}

\section{Values in noise correlations}
Thus the noise distribution at time step 1 can be represented as: $\text{Pr}(N=1)=\text{Pr}(N_1=1|X_1=0)\text{Pr}(X_1=0)+\text{Pr}(N_1=1|X_1=1)\text{Pr}(X_1=1)$.

The noise distribution at time step 2 can be expressed as:
\begin{equation*}
    \begin{aligned}
       \text{Pr}(N_2=1)=&\text{Pr}(X_2\neq{Y_2})\\
       =&\text{Pr}(Y_2=1,X_2=0)+\text{Pr}(Y_2=0,X_2=1)\\
       =&\text{Pr}(Y_2=1|X_2=0)\text{Pr}(X_2=0)\\+&\text{Pr}(Y_2=0|X_2=1)\text{Pr}(X_2=1);
    \end{aligned}
\end{equation*}
where $\text{Pr}(Y_2=1|X_2=0)$ can be expressed as:
\begin{equation*}
    \begin{aligned}
       &\text{Pr}(Y_2=1|X_2=0)\\
      =&\sum_{x_1,y_1=0}^1\text{Pr}(Y_2=1|X_1=x_1,Y_1=y_1,X_2=0)\\&~~~~~~~~~~~~~~~\cdot\text{Pr}(X_1=x_1|X_2=0)\text{Pr}(Y_1=y_1|X_1=x_1).
    \end{aligned}
\end{equation*}
on the other hand, $\text{Pr}(X_2=0)$ can be expressed as:
\begin{equation}
\begin{aligned}
    &\text{Pr}(X_2=0)=\text{Pr}(X_2=0|X_1=0)\text{Pr}(X_1=0)\\&~~~~~~+\text{Pr}(X_2=0|X_1=1)\text{Pr}(X_1=1).
\end{aligned}
\end{equation}

By the same way, we can calculate other terms in the expression of $\text{Pr}(N_2=1)$. 
where $E[N_1]=\text{Pr}(N_1=1)$, $E[N_2]=\text{Pr}(N_2=1)$, $\sqrt{E[N_1^2]-E^2[N_1]}=\text{Pr}(N_1=1)\text{Pr}(N_1=0)$, $\sqrt{E[N_2^2]-E^2[N_2]}=\text{Pr}(N_2=1)\text{Pr}(N_2=0)$.

Now we derive the value of $E[N_1N_2]$, which can be treated as a new binary random variable, thus $E[N_1N_2]=\text{Pr}(N_1=1,N_2=1)=\text{Pr}(N_2=1|N_1=1)\text{Pr}(N_1=1)$, and $\text{Pr}(N_2=1|N_1=1)$ can be expressed as:
\begin{equation}
    \sum_{i=0,j=0}^1\text{Pr}(Y_2=i,X_2=1-i|Y_1=j,X_1=1-j),
\end{equation}
Which means there are four combinations of $i$ and $j$, and for each of them we have:
\begin{equation}
    \begin{aligned}
       &\text{Pr}(Y_2=i,X_2=1-i|Y_1=j,X_1=1-j)\\
       =&\frac{\text{Pr}(Y_2=i,X_2=1-i,Y_1=j,X_1=1-j)}{\text{Pr}(Y_1\neq{X_1})}\\
       =&\frac{q(i)^j_{1-j,1-i}\text{Pr}(X_2=1-i|X_1=1-j)q(j)_{1-j}P_{1-j}}{\text{Pr}(Y_1\neq{X_1})}
    \end{aligned}
\end{equation}
where $\text{Pr}(Y_1\neq{X_1})=\sum_{i=0}^1\text{Pr}(Y_1=i|X_1=1-i)\text{Pr}(X_1=1-i)$.
Taking values in, we have 
\begin{equation}
    \begin{aligned}
        &\text{Pr}(N_2=1|N_1=1)\\
    =&\frac{[(1-\phi)q(1)^1_{00}+\phi q(0)^1_{01}](1-P_1)q(1)_0}{(1-P_1)q(1)_0+P_1q(0)_1}\\
    +&\frac{[\phi q(1)^{0}_{10}+(1-\phi)q(0)^{0}_{11}]q(0)_1P_1}{(1-P_1)q(1)_0+P_1q(0)_1}
    \end{aligned}
\end{equation}
So far, all the values in the expression of $\rho_{N_1N_2}$ can be calculated.

\section{Proof for key insights in Theorem 4}
\begin{proof}

The metric in the local information leakage can be expressed as:
\begin{equation}
    \begin{aligned}
        &\frac{\text{Pr}(\bold{X}_1^k=\bold{x}_1^k|\bold{Y}_1^k=\bold{y}_1^k)}{\text{Pr}(\bold{X}_1^k=\bold{x}_1^k|\bold{Y}_1^{k-1}=\bold{y}_1^{k-1})}\\
        =&\frac{\text{Pr}(X_k=x_k|\bold{Y}_1^k=\bold{y}_1^k)}{\text{Pr}(X_k=x_k|\bold{Y}_1^{k-1}=\bold{y}_1^{k-1})}\\
        &~~~~~~~\cdot\frac{\text{Pr}(\bold{X}_1^{k-1}=\bold{x}_1^{k-1}|\bold{Y}_1^k=\bold{y}_1^k,X_k=x_k)}{\text{Pr}(\bold{X}_1^{k-1}=\bold{x}_1^{k-1}|\bold{Y}_1^{k-1}=\bold{y}_1^{k-1},X_k=x_k)}
    \end{aligned}
\end{equation}
Based on the optimal perturbation parameters, we have 
\begin{equation}
\begin{aligned}
    &\text{Pr}(Y_k=y_k|\bold{X}_1^k=\bold{x}_1^k,\bold{Y}_1^{k-1}=\bold{y}_1^{k-1})\\=& \text{Pr}(Y_k=y_k|X_k=x_k,\bold{Y}_1^{k-1}=\bold{y}_1^{k-1}).
\end{aligned}
\end{equation}
Thus 
\begin{equation}
    \frac{\text{Pr}(\bold{X}_1^k=\bold{x}_1^k|\bold{Y}_1^k=\bold{y}_1^k)}{\text{Pr}(\bold{X}_1^k=\bold{x}_1^k|\bold{Y}_1^{k-1}=\bold{y}_1^{k-1})}=\frac{\text{Pr}(X_k=x_k|\bold{Y}_1^k=\bold{y}_1^k)}{\text{Pr}(X_k=x_k|\bold{Y}_1^{k-1}=\bold{y}_1^{k-1})}
\end{equation}
\end{proof}

\section{Proof of Theorem 4}\label{opt_sol}

% We fist analyze the utility function of $E[D(Q(x),Q(y))]$, where the distance $D(Q(x),Q(y))$ is non-negative. $D(Q(x),Q(y))=0$ when $Q(x)=Q(y)$ is a sufficient but not necessary condition to minimize $D(Q(x),Q(y))$. 
%  As a result, one optimal solution is to minimize all the $q_{xy}$s, $\forall{x\neq{y}}$. 
% The minimized $q_{xy}$ are at the boundaries of $q_{xy}=e^{-\epsilon}\sum_{x'\in{\mathcal{X}}}P_X(x') q_{x'y}$.
% We have $q_{xy}=P_X(y)e^{-\epsilon}$, by the third constraint, we have $q_{xx}=1-[1-P_X(x)]e^{-\epsilon}$. We next check whether the first and second constraints is satisfied:

% Firstly, 
% \begin{equation*}
%     1-\frac{(1-P_X(x))}{e^{\epsilon}}-\frac{P_X(y)}{e^{\epsilon}}\ge{0}
% \end{equation*}
% which means the first constraint always hold for all $x\in{\mathcal{X}}$ and $y\in{\mathcal{Y}}$.

% Secondly, 
% \begin{equation*}
% \begin{aligned}
%     e^{\epsilon}-\frac{\operatorname{Pr}(Y=y)}{{q_{yy}}}=&e^{\epsilon}-\frac{e^{\epsilon}P_X(y)}{e^{\epsilon}+P_X(y)}\ge{0}
%     \end{aligned}
% \end{equation*}
% which means the second constraint holds for all $x\in{\mathcal{X}}$ and $y\in{\mathcal{Y}}$.

{Privacy Constraints for Perturbation Parameters}
The privacy constraints can be further expressed as:
\begin{equation}
\begin{aligned}
    &\frac{\text{Pr}(\bold{X}_{k-L+1}^k=\bold{x}_{k-L+1}^k|\bold{Y}_1^k=\bold{y}_{1}^{k})}{\text{Pr}(\bold{X}_{k-L+1}^k=\bold{x}_{k-L+1}^k|\bold{Y}_1^{k-1}=\bold{y}_{1}^{k-1})}\\
    =&\frac{a^s_{k}(y_{k}| \bold{x}_{k-L+1}^k, \bold{y}_{1}^{k-1})}{\text{Pr}(Y_k=y_k|\bold{Y}_1^{k-1}=\bold{y}_{1}^{k-1})}\\
    =&\frac{ a^s_{k}(y_{k}| \bold{x}_{k-L+1}^k \bold{y}_{1}^{k-1})}{\sum_{\bar{\bold{x}}_{k-L+1}^k\in{\mathcal{D}}}a^s_{k}(y_{k}|\bar{\bold{x}}_{k-L+1}^k, \bold{y}_{1}^{k-1})\beta^s_k(\bar{\bold{x}}_{k-L+1}^k|\bold{y}_1^{k-1})},
\end{aligned}
\end{equation}
from which, we know that, for any $L>=1$, the parameter $a^s_{k}(y_{k}| \bold{x}_{k-L+1}^k \bold{y}_{1}^{k-1})$ is bounded by:
\begin{equation}
    {\text{Pr}(Y_k=y_k|\bold{Y}_1^{k-1}=\bold{y}_{1}^{k-1})}[{e^{-\epsilon}}, e^{\epsilon}]
\end{equation}
which implies  a range of
\begin{equation}
    \left[\frac{\beta^s_k({\bold{x}}_{k-L+1}^k|\bold{y}_1^{k-1})}{e^{\epsilon}},{1-\frac{1-\beta^s_k({\bold{x}}_{k-L+1}^k|\bold{y}_1^{k-1})}{e^{\epsilon}}}\right]
\end{equation}

Next, we show that the solution satisfies KKT condition. The Jacobian matrix of the objective function can be written as \eqref{eq49}.

% \begin{equation*}
%     \nabla{f}=\left[\begin{array}{c}
%        \beta_k(\bold{x}_1^{k-1},1|\bold{y}_1^{k-1})\\\beta_k(\bold{x}_2^{k-1},x|\bold{y}_1^{k-1})\\...\\\beta_k(\bold{x}_1^{k-1},|\mathcal{X}||\bold{y}_1^{k-1})
%     \end{array}\right]
%     \left[\begin{array}{ccc}
%        D[Q(1),Q(1)]&D[Q(1),Q(2)]...&D[Q(1),Q(|\mathcal{Y}|)]\\
%        D[Q(2),Q(1)]&D[Q(2),Q(2)]...&D[Q(2),Q(|\mathcal{Y}|)]\\
%        ...&...&...\\
%        D[Q(|\mathcal{X}|),Q(1)]&D[Q(|\mathcal{X}|),Q(2)]...&D[Q(|\mathcal{X}|),Q(|\mathcal{Y}|)]\\
%     \end{array}\right]
% \end{equation*}
The second order derivative are all zeros, which means the objective function is convex (linear).
The inequality constraints can be rewritten as

As an example, the Jacobian Matrix of the first row of the inequality constraints can be expressed as \eqref{eq50}.

% \begin{equation*}
%        \nabla g_{1}=\left[\begin{array}{cccc}
%        \frac{\sum_{x\neq{1}}\beta_k(\bold{x}_1^{k-1},x|\bold{y}_1^{k-1})a_k(1|\bold{x}_1^{k-1},x,\bold{y}_1^{k-1})}{a^2_k(1|\bold{x}_1^{k-1},1,\bold{y}_1^{k-1})}&0&...&0\\
%        \frac{\beta_k(\bold{x}_1^{k-1},2|\bold{y}_1^{k-1})}{a_k(1|\bold{x}_1^{k-1},1,\bold{y}_1^{k-1})}&0&...&0\\
%        ...&...&...&...\\
%        \frac{\beta_k(\bold{x}_1^{k-1},|\mathcal{X}||\bold{y}_1^{k-1})}{a_k(1|\bold{x}_1^{k-1},1,\bold{y}_1^{k-1})}&0&...&0\\
%     \end{array}\right],
%     \nabla g_{2}=\left[\begin{array}{cccc}
%        \frac{\sum_{-x\neq{1}}\beta_k(\bold{x}_1^{k-1},x|\bold{y}_1^{k-1})a_k(1|\bold{x}_1^{k-1},x,\bold{y}_1^{k-1})}{a^2_k(1|\bold{x}_1^{k-1},1,\bold{y}_1^{k-1})}&0&...&0\\
%        -\frac{\beta_k(\bold{x}_1^{k-1},2|\bold{y}_1^{k-1})}{a_k(1|\bold{x}_1^{k-1},1,\bold{y}_1^{k-1})}&0&...&0\\
%        ...&...&...&...\\
%        -\frac{\beta_k(\bold{x}_1^{k-1},|\mathcal{X}||\bold{y}_1^{k-1})}{a_k(1|\bold{x}_1^{k-1},1,\bold{y}_1^{k-1})}&0&...&0\\
%     \end{array}\right],
% \end{equation*}

Observe that constraints from $g_1$ to $g_{2|\mathcal{X}|}$ are non-negative in the first column, then \eqref{eq53}

% \begin{equation*}
% \begin{aligned}
%     &\sum_{i=1}^{2|\mathcal{X}|}\mu_{i}\nabla{g_i}\\=&\left[\begin{array}{c}
%        \frac{(\mu_1-\mu_2)\sum_{x\neq{1}}a_k(1|\bold{x}_1^{k-1},x,\bold{y}_1^{k-1})\beta_k(\bold{x}_1^{k-1},x|\bold{y}_1^{k-1})}{a^2_k(1|\bold{x}_1^{k-1},1,\bold{y}_1^{k-1})}+\sum_{i=1}^{|\mathcal{X}|}(\mu_{2i-1}-\mu_{2i})\frac{\beta_k(\bold{x}_1^{k-1},1|\bold{y}_1^{k-1})}{a_k(1|\bold{x}_1^{k-1},i,\bold{y}_1^{k-1})}-(\mu_{1}-\mu_{2})\frac{\beta_k(\bold{x}_1^{k-1},1|\bold{y}_1^{k-1})}{a_k(1|\bold{x}_1^{k-1},1,\bold{y}_1^{k-1})}\\
%        \frac{(\mu_3-\mu_4)\sum_{x\neq{1}}a_k(1|\bold{x}_1^{k-1},x,\bold{y}_1^{k-1})\beta_k(\bold{x}_1^{k-1},x|\bold{y}_1^{k-1})}{a^2_k(1|\bold{x}_1^{k-1},2,\bold{y}_1^{k-1})}+\sum_{i=1}^{|\mathcal{X}|}(\mu_{2i-1}-\mu_{2i})\frac{\beta_k(\bold{x}_1^{k-1},2|\bold{y}_1^{k-1})}{a_k(1|\bold{x}_1^{k-1},i,\bold{y}_1^{k-1})}-(\mu_{3}-\mu_{4})\frac{\beta_k(\bold{x}_1^{k-1},2|\bold{y}_1^{k-1})}{a_k(1|\bold{x}_1^{k-1},2,\bold{y}_1^{k-1})}\\
%        ...\\
%        \frac{(\mu_{2|\mathcal{X}-1|}-\mu_{2|\mathcal{X}|})\sum_{x\neq{1}}a_k(1|\bold{x}_1^{k-1},x,\bold{y}_1^{k-1})\beta_k(\bold{x}_1^{k-1},x|\bold{y}_1^{k-1})}{a_k(1|\bold{x}_1^{k-1},|\mathcal{X}|,\bold{y}_1^{k-1})}+\sum_{i=1}^{|\mathcal{X}|}(\mu_{2i-1}-\mu_{2i})\frac{\beta_k(\bold{x}_1^{k-1},|\mathcal{X}||\bold{y}_1^{k-1}))}{a_k(1|\bold{x}_1^{k-1},i,\bold{y}_1^{k-1})}-(\mu_{2|\mathcal{X}|-1}-\mu_{2|\mathcal{X}|})\frac{\beta_k(\bold{x}_1^{k-1},|\mathcal{X}||\bold{y}_1^{k-1}))}{a_k(1|\bold{x}_1^{k-1},|\mathcal{X}|,\bold{y}_1^{k-1})}\\
%     \end{array}\right],
%     \end{aligned}
% \end{equation*}
Similarly, the inequality constraints from $g_{2|\mathcal{X}|+1}$ to $g_{3|\mathcal{X}|}$ determines the second column so on so forth. The equality constraints can be formulated as
\begin{equation*}
    \left\{h_j=\sum_{y\in{\mathcal{Y}}}a_k(y|\bold{x}_1^{k-1},j,\bold{y}_1^{k-1})-1\right\}_{j=1}^{|\mathcal{X}|}.
\end{equation*}
with Jacobian matrix of 
\begin{equation*}
    \nabla h_{j}=\left[\begin{array}{cccc}
       0&0&...&0\\
       ...&...&...&...\\
       1&1&...&1\\
       ...&...&...&...\\
       0&0&...&0\\
    \end{array}\right],
\end{equation*}
with non-negative values at the $j$-th row. As a result, 
\begin{equation*}
    \sum_{j=1}^{|\mathcal{X}|}\lambda_j\nabla{h_j}=\left[\begin{array}{cccc}
       \lambda_1&\lambda_1&...&\lambda_1\\
       ...&...&...&...\\
       \lambda_j&\lambda_j&...&\lambda_j\\
       ...&...&...&...\\
       \lambda_{|\mathcal{X}|}&\lambda_{|\mathcal{X}|}&...&\lambda_{|\mathcal{X}|}\\
    \end{array}\right].
\end{equation*}
Denote the solution we found in proposition 1 as $a^*$, then our objective is to derive a set of $\{\mu_{i}\}_{i=1}^{2|\mathcal{X}||\mathcal{Y}|}$ and a set of $\{\lambda_{j}\}_{j=1}^{|\mathcal{X}|}$ that satisfy:
\begin{equation*}
\begin{aligned}
   &\nabla{f(a^*)}-\sum_{i=1}^{2|\mathcal{X}||\mathcal{Y}|}\mu_i\nabla{g_i(a^*)}-\sum_{j=1}^{|\mathcal{X}|}\lambda_j\nabla{h_j(a^*)}=0\\
   &\mu_ig_i(a^*)=0, \forall{i=1,2,...,|\mathcal{X}||\mathcal{Y}|},\\
   &\mu_i\ge{0}, \forall{i=1,2,...,|\mathcal{X}||\mathcal{Y}|}.
\end{aligned}
\end{equation*}

Observe that, $g_{2i-1}(a^*)=0$ for all $\{1\le{i\le{2|\mathcal{X}||\mathcal{Y}|}}\}/\{1,|\mathcal{X}|+1,...,|\mathcal{X}|(|\mathcal{Y}-1|+|\mathcal{Y}-1|)\}$. As a result, $\mu_{2i}=0$ for all $\{1\le{i\le{2|\mathcal{X}||\mathcal{Y}|}}\}/\{1,|\mathcal{X}|+1,...,|\mathcal{X}|(|\mathcal{Y}-1|+|\mathcal{Y}-1|)\}$. The first column of the Lagrangian function can be expressed as:

% \begin{small}
% \begin{equation*}
% \begin{aligned}
% &\left[\begin{array}{c}
%        \beta_k(\bold{x}_1^{k-1},1|\bold{y}_1^{k-1})\\\beta_k(\bold{x}_1^{k-1},2|\bold{y}_1^{k-1})\\...\\\beta_k(\bold{x}_1^{k-1},|\mathcal{X}||\bold{y}_1^{k-1})
%     \end{array}\right]
%     \left[\begin{array}{c}
%        D[Q(1),Q(1)]\\
%        D[Q(2),Q(1)]\\
%        ...\\
%        D[Q(|\mathcal{X}|),Q(1)]\\
%     \end{array}\right]\\
%     &-\left[\begin{array}{c}
%        \frac{(\mu_1-\mu_2)\beta_k(\bold{x}_1^{k-1},1|\bold{y}_1^{k-1})(1-\beta_k(\bold{x}_1^{k-1},1|\bold{y}_1^{k-1}))}{a^2_k(1|\bold{x}_1^{k-1},1,\bold{y}_1^{k-1})}+\sum_{i\neq1}\mu_{2i-1}e^{\epsilon}\\
%        \frac{\mu_3\beta_k(\bold{x}_1^{k-1},2|\bold{y}_1^{k-1})(1-\beta_k(\bold{x}_1^{k-1},2|\bold{y}_1^{k-1}))}{a^2_k(1|\bold{x}_1^{k-1},2,\bold{y}_1^{k-1})}+\sum_{i\neq2}\mu_{2i-1}\frac{\beta_k(\bold{x}_1^{k-1},2|\bold{y}_1^{k-1})}{\beta_k(\bold{x}_1^{k-1},1|\bold{y}_1^{k-1})}e^{\epsilon}\\
%        ...\\
%        \frac{\mu_{2|\mathcal{X}|-1}\beta_k(\bold{x}_1^{k-1},|\mathcal{X}||\bold{y}_1^{k-1})(1-\beta_k(\bold{x}_1^{k-1},|\mathcal{X}||\bold{y}_1^{k-1}))}{a^2_k(1|\bold{x}_1^{k-1},|\mathcal{X}|,\bold{y}_1^{k-1})}+\sum_{i\neq|\mathcal{X}|}\mu_{2i-1}\frac{\beta_k(\bold{x}_1^{k-1},|\mathcal{X}||\bold{y}_1^{k-1}))}{\beta_k(\bold{x}_1^{k-1},1|\bold{y}_1^{k-1})}e^{\epsilon}\end{array}\right]-
%        \left[\begin{array}{c}
%        \lambda_1\\
%        \lambda_2\\
%        ...\\
%        \lambda_{|\mathcal{X}|}
%     \end{array}\right]=0,
% \end{aligned}
% \end{equation*}
% \end{small}
For the other values of $\mu$ and $\lambda$ in the first column, one feasible solution are
\begin{align}
\mu_{2i-1} & =\frac{D(Q(i),Q(1))\beta^2_k(\bold{x}_1^{k-1},1|\bold{y}_1^{k-1})}{\beta_k(\bold{x}_1^{k-1},1|\bold{y}_1^{k-1})(1-\beta_k(\bold{x}_1^{k-1},i|\bold{y}_1^{k-1}))e^{\epsilon}} \nonumber \\ 
\mu_2 &=\frac{D(Q(i),Q(1))(e^{\epsilon}-1)(e^{\epsilon}-1+2\beta_k(\bold{x}_1^{k-1},1|\bold{y}_1^{k-1}))}{\beta_k(\bold{x}_1^{k-1},1|\bold{y}_1^{k-1})(1-\beta_k(\bold{x}_1^{k-1},i|\bold{y}_1^{k-1}))e^{\epsilon}}, \nonumber \\ 
\lambda_i & =\beta_k(\bold{x}_1^{k-1},i|\bold{y}_1^{k-1})\sum_{j=1}^{|\mathcal{Y}|}D(Q(i),Q(j)).\nonumber
\end{align}
Similarly, we can find feasible solutions of $\mu$s and $\lambda$s for the second column. Note that all the $\mu$s are non-negative, which means the KKT condition is satisfied, further implying that $q^*$ is a globally optimal solution.

%\section{Relation to Input Sensitivity}

\end{appendices}

\end{document}